\documentclass[prl,twocolumn,longbibliography,nofootinbib,groupedaddress,preprintnumbers,notitlepage,fleqn,showpacs,superscriptaddress]{revtex4-1} 

\usepackage{soul} 

\usepackage{xcolor}

\usepackage[bookmarks = true, citecolor = blue, colorlinks = true, linkcolor = magenta, urlcolor = blue]{hyperref}

\usepackage[T1]{fontenc}			
\usepackage[sc,osf]{mathpazo}   	

\usepackage{amsmath}  			
\usepackage{amsfonts}  			
\usepackage{graphicx}   			
\usepackage{mathrsfs, amsthm, amssymb}
\usepackage{bbm, bm}

\usepackage{nicefrac}    			
\usepackage{longtable}
\usepackage{array}

\usepackage{fourier-orns}  		

\usepackage[braket, qm]{qcircuit}	

\newtheorem{theorem}{Theorem}
\newtheorem{definition}{Definition}

\newtheorem{lemma}{Lemma}

\usepackage{contour}
\usepackage{ulem}

\contourlength{0.8pt}

\newcommand{\myuline}[1]{%
  \uline{\phantom{#1}}%
  \llap{\contour{white}{#1}}%
}
\renewcommand{\emph}{\textit}

\usepackage{xcolor}
\definecolor{dark green}{rgb}{0.0,0.5,0.0}

\begin{document}


\title[]{
Measuring Bell non-locality in the presence of signaling
}

\author{Mark \surname{Broom}}
\affiliation{City St George's, University of London, London EC1V 0HB, United Kingdom}

\author{Talel \surname{Naccache}}
\affiliation{City St George's, University of London, London EC1V 0HB, United Kingdom}

\author{Emmanuel M. \surname{Pothos}}
\affiliation{City St George's, University of London, London EC1V 0HB, United Kingdom}

\author{Christoph \surname{Gallus}}
\affiliation{Technische Hochschule Mittelhessen, D-35390 Gießen, Germany}

\author{Pawel \surname{Blasiak}}
\email[Corresponding author: ]{pawel.blasiak@ifj.edu.pl}
\affiliation{Institute of Nuclear Physics, Polish Academy of Sciences, 31342 Krak\'ow, Poland}


\begin{abstract}
Scientific inquiry seeks causal explanations of observed phenomena. The Bell experiment provides a paradigmatic case, revealing correlations between spatially separated systems that no local model can reproduce. Such correlations, known as Bell non-locality, are typically analyzed under the non-signaling assumption, which requires that local statistics be independent of distant measurement choices. Yet real experiments, as well as applications beyond physics, often involve signaling, raising the question of how non-locality should be characterized without this constraint. We introduce a general method for quantifying Bell non-locality in the presence of signaling, designed to relax locality as little as necessary. Our approach is guided by the question: \textit{how often can locality be preserved across repeated trials in explaining the observed correlations?} The task reduces to finding the optimal convex decomposition of the observed correlations into local and genuinely non-local components. We solve this problem within the linear-programming framework, obtaining a closed-form solution valid for arbitrary correlations. We further evaluate a corresponding measure of signaling, demonstrating the generality of the method and the non-trivial character of the results. By extending the notion of non-locality beyond the non-signaling regime, our framework reshapes the basis for experimental analysis in physics and offers tools applicable outside physics.

\vspace{1cm}
\end{abstract}


\maketitle





\section{\textsf{Introduction}}
\vspace{-0.3cm}

At the heart of science lies the pursuit of uncovering causal relations in observed phenomena, an endeavor that finds a particularly stringent test in the correlations revealed by Bell inequalities~\cite{Be64,Be93}. Originally formulated within the context of physics, these inequalities confront deep-rooted assumptions about the workings of nature, showing that locality cannot be maintained alongside the statistical predictions of quantum mechanics. Over the past decades, increasingly refined experiments have confirmed their violation, firmly establishing Bell non-locality as an intrinsic feature of physical reality~\cite{FrCl72,AsDaRo82,HeBeDrReKaBlRu15,GiVeWeHaHoPhSt15,ShMeChBiWaStGe15,BIGBellCollaboration18,RaHaHoGaFrLeLi18}, while simultaneously opening the way for modern quantum information science and its applications. This stands out as a prominent example within the broader paradigm of causal discovery across disciplines~\cite{PeMa18,Pe09,SpGlSc00,AnPi09,RuIm15,HeRo20}. In this paper, we extend the original Bell framework to provide a unified basis for analyzing Bell-type scenarios across the full range of possible physical contexts and applications beyond physics.

Bell non-locality has been studied primarily in the non-signaling regime, where the local statistics of one party is independent of the measurement choices made by the other~\cite{Sc19,BrCaPiScWe14,Gi14c}. This assumption is grounded in the conventional setting of spatially separated systems, where standard quantum mechanics implies no faster-than-light communication between the parties. There are, however, compelling reasons to move beyond this regime. 
First, real experiments are rarely perfectly non-signaling, with imperfections in design and noise introducing unavoidable deviations~\cite{AdKh17,HeKaBlDrReVeSc16}. This makes the analysis appear sometimes contrived, as signaling alone suffices to demonstrate non-locality, highlighting the need for a unified framework that treats all conditions on an equal footing.
Second, not all fundamental tests in physics rely on spatial separation, with temporal scenarios~\cite{EmLaNo14,KoBr13,Fr10,BrKoMaPaPr15} or instrumental tests~\cite{RiGiChCoWhFe16,ChCaAgDiAoGiSc18,AgPoPoMiGaSuPo22} providing some notable examples. In such cases, disturbance is to be expected, underscoring the need for a unified approach that consistently incorporates effective signaling.
Third, in applications beyond physics -- such as Bayesian networks, epidemiology, social systems, econometrics, and cognition -- signaling is often inherent to the design, making the standard Bell non-signaling framework too restrictive~\cite{CeDz18,BrFeHoDeObGiMo23,GaPoBlYeWo23}. A broader framework is therefore essential to align the study of Bell non-locality with general causal-discovery paradigms in both physical applications and other sciences~\cite{PeMa18,Pe09,SpGlSc00,AnPi09,RuIm15,HeRo20}, where signaling is often intrinsic and cannot be ruled out on purely principled grounds. 

In this paper, we study Bell non-locality without requiring non-signaling and quantify how much locality can account for an observed behavior even in the presence of signaling. Our approach is guided by a simple question:
\begin{eqnarray}\nonumber
\hspace{0.45cm}\textit{\parbox{0.9\columnwidth}{Given that the observed statistics arise from a repeated experiment, \myuline{how often} can locality be preserved across trials while still explaining the observed correlations?}}&
\end{eqnarray}
We call this the \textit{(fractional) measure of locality}, or more briefly the \textit{local fraction}, since it reflects the proportion of experimental trials that admit a local explanation; see~\cite{ElPoRo92,Ha91,BaKePi06,CoRe08,CoRe16,PoBrGi12,BlPoYeGaBo21}.

Note that reducing the answer to the coarse distinction between \textit{‘always’} and \textit{‘sometimes’} brings us back to Bell’s original question: \textit{is the observed statistics local or non-local?} This binary test is resolved by the standard criterion: \textit{correlations admit a local explanation if and only if they are non-signaling and satisfy the Bell-CHSH inequalities}~\cite{ClHoShHo69,Fi82a}. Here, however, we go further by introducing a quantitative measure that captures the \textit{degree} of departure from locality in the observed correlations. Crucially, it provides a direct means of quantifying non-locality without requiring any additional assumptions. In doing so, we go beyond the scope of Ref.~\cite{BlPoYeGaBo21}, thereby laying the foundation for a quantitative framework of non-locality beyond the non-signaling assumption and opening new avenues for both theoretical and experimental advances.

Our result establishes a closed-form solution for arbitrary correlations in the two-setting, two-outcome Bell scenario. At its core, the challenge is to evaluate a non-trivial geometric distance from a given behavior to the local polytope, constructed via convex decomposition and constrained by the space of admissible behaviors. In the general case, where signaling is allowed, these constraints reduce to simple normalization conditions.  Geometrically, we show that the problem reduces to evaluating a set of scalar products with specific vectors that project the observed behavior onto the relevant directions of the correlation space. We refer to these operators with prefixes indicating the class to which they belong, e.g., \textit{f-vectors} or \textit{e-vectors} and the corresponding closeness measure as \textit{f-measure}, etc. This representation is readily amenable to computational implementation. See Fig.~\ref{Fig-Solution-Summary} for a summary. 

The scope and practical utility of the approach are partly revealed through comparison with an analogous measure that quantifies the degree of signaling. In direct analogy with the local fraction, we define a quantity that captures the extent of non-signaling in an observed behavior, which we call the \textit{(fractional) measure of non-signaling} or simply the \textit{non-signaling fraction}. This formulation places the problem within a broader landscape, emphasising the structural richness of the solution and offering a clear framework for comparative analysis. Finally, we demonstrate the indispensable role of each component in the solution and analyze example datasets of behaviors to illustrate how often each element of the solution set realizes the local or non-signaling fraction.


\begin{figure}[t]
\centering
\includegraphics[width=1\columnwidth]{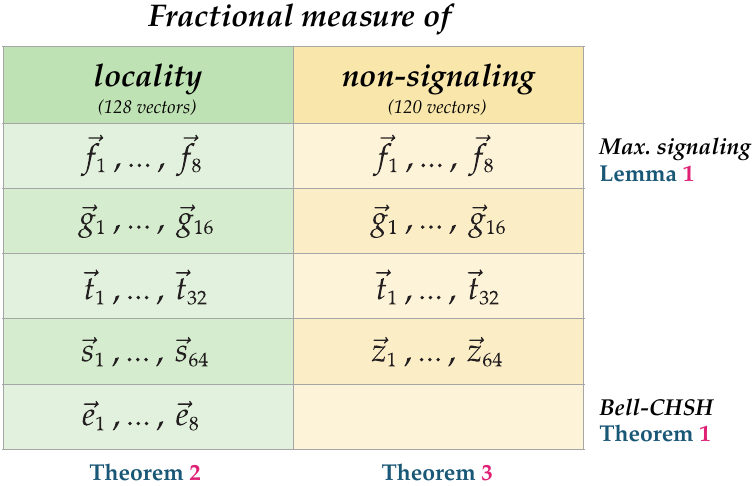}
\caption
{\label{Fig-Solution-Summary}{\bf\textsf{Summary of the results.}} The \textit{fractional measure} reflects the proportion of experimental trials that can still be understood within a specific explanatory framework, be it \textit{local} or \textit{non-signaling}. We calculate it explicitly and show that the result takes the form of a minimum over simpler linear measures, each represented by a vector that projects the observed statistics onto the relevant directions of the correlation space. Accordingly, the table presents the collections of vectors for each measure, with 128 vectors specifying the \textit{local fraction} (\textbf{Theorem~\ref{theoremL}}) and 120 vectors specifying the \textit{non-signaling fraction} (\textbf{Theorem~\ref{theoremNS}}). Importantly, none of these vectors can be omitted and some coincide with familiar quantities: notably, the \textit{f-vectors} correspond to maximal signaling (\textbf{Lemma~\ref{Lemma-muNS-f}}), while the \textit{e-vectors} align with the Bell-CHSH expressions (\textbf{Theorem~\ref{Theorem-muL-e}}).}
\end{figure}

\section{\textsf{Results}}\vspace{-0.3cm}

\subsection{\textsf{Bell scenario, locality and non-signaling}}\vspace{-0.3cm}

\begin{figure*}[t]
\centering
\includegraphics[width=0.8\columnwidth]{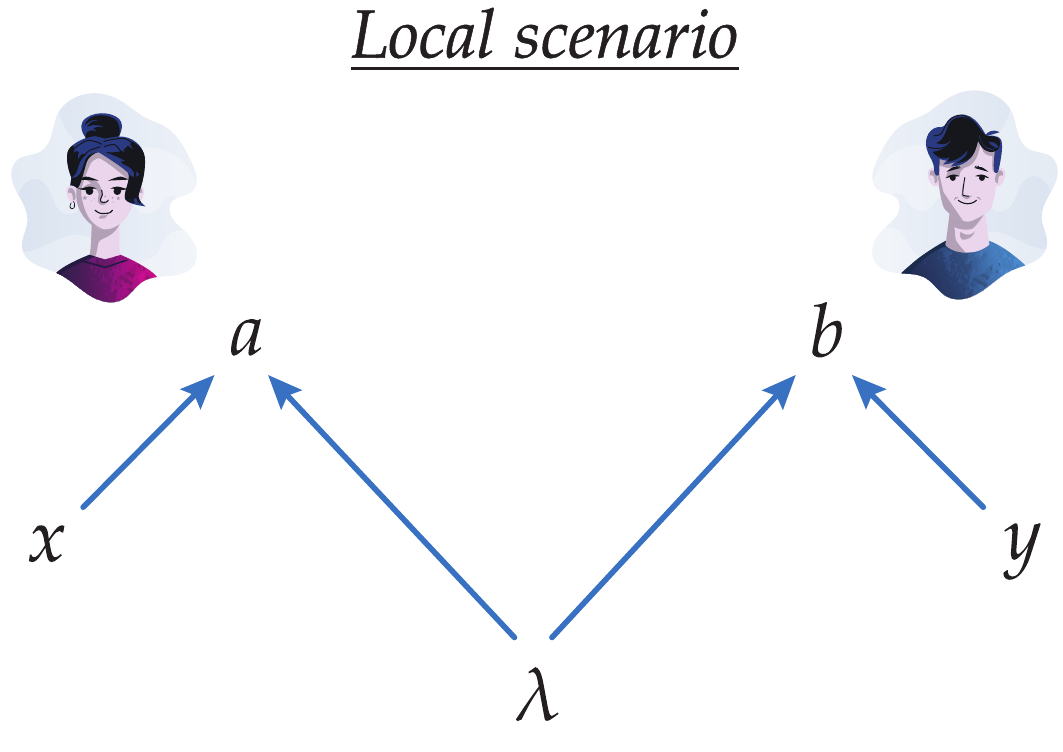}\qquad\qquad\qquad\qquad\qquad
\includegraphics[width=0.8\columnwidth]{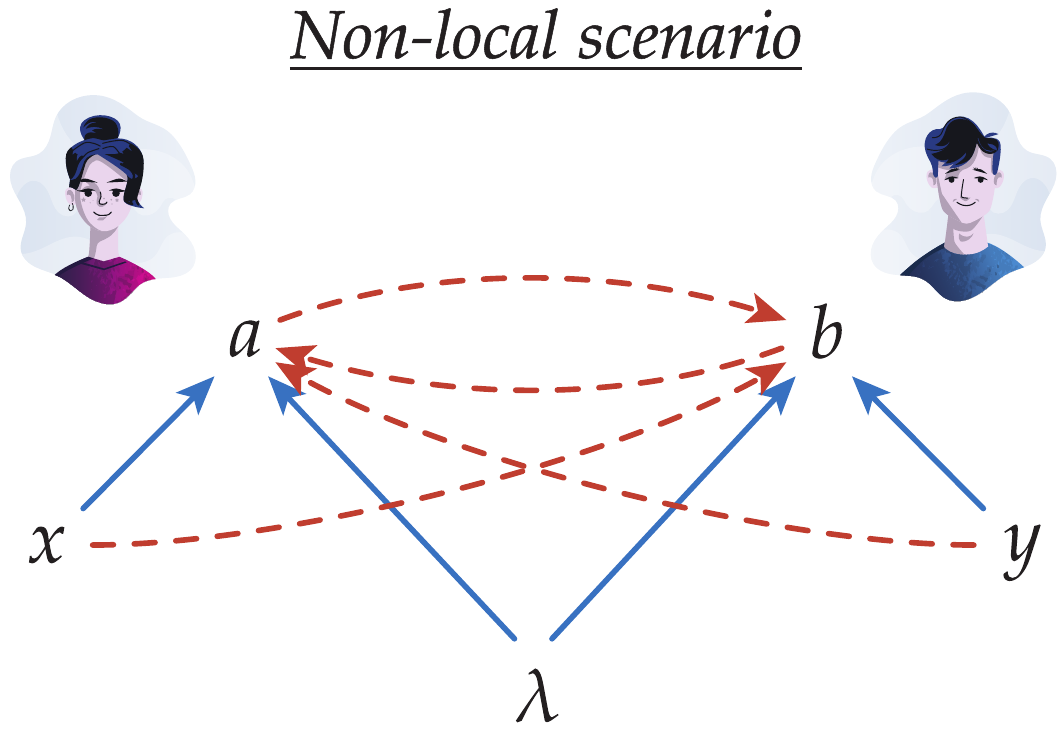}
\caption
{\label{Bell-DAGs}{\bf\textsf{Causal explanations of Bell experiment.}} Two parties, Alice and Bob, perform experiments at spatially separated locations. For each \textit{freely chosen} pair of measurement settings $x$ and $y$, they obtain outcomes $a$ and $b$, respectively, giving rise to the observed statistics $\{P_{\scriptscriptstyle  ab|xy}\}_{\scriptscriptstyle  xy}$.  Assuming that causes propagate only \textit{forward-in-time}, such statistics can be described by \textit{directed acyclic graphs (DAGs)}, where nodes represent measurement settings, outcomes, or shared variables, and arrows capture potential cause-and-effect relations. On the left, the causal DAG depicts a situation where the locality condition rules out any direct causal influence between the parties, so correlations can only originate from a shared common cause in the past, denoted by $\lambda$. On the right, the presence of any red arrow indicates a direct causal influence and therefore constitutes a violation of the locality assumption. While the \textit{local} scenario cannot reproduce arbitrary correlations ${P_{\scriptscriptstyle ab|xy}}$, allowing red arrows in the \textit{non-local} scenario---potentially active in every experimental trial---can generate any behavior. In this work, we investigate to what extent the correlations observed in Bell experiments can be accounted for by \textit{minimally relaxing locality}, that is, by introducing non-local (red) arrows as sparingly as possible across repeated trials.
}
\end{figure*}

Consider the simplest Bell scenario, where two parties -- Alice and Bob -- perform measurements on spatially separated systems. In each trial, Alice chooses one of two measurement settings $x\in\{0,1\}$ and likewise Bob chooses one of two settings $y\in\{0,1\}$, obtaining respective outcomes $a,b\in\{0,1\}$. Repeating the experiment many times generates statistics that defines four conditional distributions $\{P_{\scriptscriptstyle  ab|xy}\}_{\scriptscriptstyle  xy\,=\,00,01,10,11}$, where $P_{\scriptscriptstyle ab|xy}$ is the probability of outcomes $a,b$ given measurement choices $x,y$. We refer to this full collection as a \textit{behavior}, represented compactly by the sixteen-component vector $\vec{P}\equiv(P_{\scriptscriptstyle ab|xy})$. In general, any behavior is admissible provided it is non-negative, $\vec{P}\geqslant0$, and properly normalized $\sum_{\scriptscriptstyle  ab}\,P_{\scriptscriptstyle  ab|xy}=1$ for each $x,y$. The set of all such behaviors defines the \textit{correlation (or behavior) space}, denoted by $\mathcal{P}$.

Bell’s scenario becomes especially revealing when physical constraints are introduced. Specifically, spatial separation is conventionally understood to preclude any direct causal influence between Alice and Bob. In this view, the observed correlations (the \textit{behavior}) are attributed to a shared cause in the past, represented by a hidden variable $\lambda$ (Fig.~\ref{Bell-DAGs}, left). This principle, known as the \textit{(Bell) locality assumption}, is formally expressed by the factorization condition~\cite{Be64,Be93,Sc19,BrCaPiScWe14,Gi14c}
\begin{eqnarray}\label{Bell-factorisation}
P_{\scriptscriptstyle ab|xy}&=&\sum_{\scriptscriptstyle \lambda}\ P_{\scriptscriptstyle a|x\lambda}\cdot P_{\scriptscriptstyle b|y\lambda}\cdot P_{\scriptscriptstyle \lambda}\,.
\end{eqnarray}
Typically, this condition is justified on causal grounds; however, for our purposes, it also suffices to regard it as a purely statistical constraint.\footnote{\label{footnote-1}
For a detailed discussion of the assumptions underlying Bell’s theorem, see e.g. Ref.~\cite{WiCa17}. Here, for clarity of exposition, we adopt the conventional causal approach, illustrated in Fig.~\ref{Bell-DAGs} (left). However, we should emphasize that our analysis relies solely on the factorization condition in Eq.\,(\ref{Bell-factorisation}), which defines the local polytope $\mathcal{P}_{\scriptscriptstyle Loc}$ --- the central object of our study. From this perspective, causal notions serve only a subsidiary (or illustrative) role in the paper. See Refs.~\cite{BlPoYeGaBo21,BlGa24} for treatments of other assumptions, and Ref.~\cite{ViRaCa25} for a different, purely probabilistic analysis of the assumptions underlying Bell inequalities.} 

The situation in which locality is imposed stands in sharp contrast to one where direct influence is allowed and the assumption is relaxed (Fig.~\ref{Bell-DAGs}, right). The central question is whether these two explanations --- \textit{local} and \textit{non-local} --- can be distinguished solely on the basis of observed data, without access to the hidden variable~$\lambda$. As it turns out, the \textit{locality assumption} imposes a set of testable constraints on the observed behavior~$\vec{P}$.

Firstly, the correlations must satisfy the non-signaling condition. This requires that Alice cannot deduce Bob’s measurement choice ($y=0$ or $1$) from the statistics available on her side alone, and likewise Bob cannot deduce Alice’s choice ($x=0$ or $1$) from his own statistics. Formally, these requirements reduce to four independence conditions on how the marginal probabilities on one side depend on the measurement choices of the other. These conditions can be compactly expressed as a single constraint on so-called \textit{maximal signaling}
\begin{eqnarray}\label{NS}
\Delta&:=&\max_{\scriptscriptstyle{i\text{\,=\,}1,...,4}}\ |\Delta_i|\ =\ 0\,,
\end{eqnarray}
where 
\begin{eqnarray}\label{NS-AB0}
\Delta_{\scriptscriptstyle  1}&=&P_{\scriptscriptstyle a\text{=}0|00}\,-\,P_{\scriptscriptstyle a\text{=}0|01}\quad\quad\ \ \text{[Bob\,$\rightarrow$\,Alice}_{\,x\text{=}0}\text{]}\,,\\\label{NS-AB1}
\Delta_{\scriptscriptstyle  2}&=&P_{\scriptscriptstyle a\text{=}0|10}\,-\,P_{\scriptscriptstyle a\text{=}0|11}\quad\quad\ \ \text{[Bob\,$\rightarrow$\,Alice}_{\,x\text{=}1}\text{]}\,,\\\label{NS-BA0}
\Delta_{\scriptscriptstyle  3}&=&P_{\scriptscriptstyle b\text{=}0|00}\,-\,P_{\scriptscriptstyle b\text{=}0|10}\quad\quad\ \ \text{[Alice\,$\rightarrow$\,Bob}_{\,y\text{=}0}\text{]}\,,\\\label{NS-BA1}
\Delta_{\scriptscriptstyle  4}&=&P_{\scriptscriptstyle b\text{=}0|01}\,-\,P_{\scriptscriptstyle b\text{=}0|11}\quad\quad\ \ \text{[Alice\,$\rightarrow$\,Bob}_{\,y\text{=}1}\text{]}\,,\end{eqnarray}
with $P_{\scriptscriptstyle a\text{=}0|xy}=\sum_{\scriptscriptstyle b}P_{\scriptscriptstyle 0b|xy}$ and $P_{\scriptscriptstyle b\text{=}0|xy}=\sum_{\scriptscriptstyle a}P_{\scriptscriptstyle a0|xy}$ denoting the marginal probabilities observed on Alice’s and Bob’s sides, respectively. Requiring all $\Delta_{\scriptscriptstyle i}$ to vanish (or equivalently $\Delta=0$) ensures that neither party can extract any information about the other’s measurement choice from their local statistics.

Secondly, the correlations must satisfy the four \textit{Bell-CHSH inequalities}~\cite{ClHoShHo69}, expressed in compact form as
\begin{eqnarray}\label{CHSH}
S&:=&\max_{\scriptscriptstyle{i\text{\,=\,}1,...,4}}\ |S_i|\ \leqslant\ 2\,,
\end{eqnarray}
where the four linear \textit{CHSH expressions} are defined as
\begin{eqnarray}\label{S1}
S_{\scriptscriptstyle  1}&=&\ \ \ \langle ab\rangle_{\scriptscriptstyle  00}+\langle ab\rangle_{\scriptscriptstyle  01}+\langle ab\rangle_{\scriptscriptstyle  10}-\langle ab\rangle_{\scriptscriptstyle  11}\,,\\\label{S2}
S_{\scriptscriptstyle  2}&=&\ \ \ \langle ab\rangle_{\scriptscriptstyle  00}+\langle ab\rangle_{\scriptscriptstyle  01}-\langle ab\rangle_{\scriptscriptstyle  10}+\langle ab\rangle_{\scriptscriptstyle  11}\,,\\\label{S3}
S_{\scriptscriptstyle  3}&=&\ \ \ \langle ab\rangle_{\scriptscriptstyle  00}-\langle ab\rangle_{\scriptscriptstyle  01}+\langle ab\rangle_{\scriptscriptstyle  10}+\langle ab\rangle_{\scriptscriptstyle  11}\,,\\\label{S4}
S_{\scriptscriptstyle  4}&=&-\langle ab\rangle_{\scriptscriptstyle  00}+\langle ab\rangle_{\scriptscriptstyle  01}+\langle ab\rangle_{\scriptscriptstyle  10}+\langle ab\rangle_{\scriptscriptstyle  11}\,,
\end{eqnarray}
and $\langle ab\rangle_{\scriptscriptstyle  xy}=\sum_{\scriptscriptstyle a,b}\,(-1)^{\scriptscriptstyle a+b}\,P_{\scriptscriptstyle  ab|xy}$ denotes the correlation coefficient for a given choice of measurements $x,y$. Note that while the algebraic maximum in the absence of locality is $S=4$, the locality assumption restricts the bound to $S\leqslant2$. Remarkably, as Bell showed~\cite{Be64,Be93}, quantum mechanics permits correlations that go beyond the local bound, up to Tsirelson’s bound of $S=2\sqrt{2}$~\cite{Ts80}.
\vspace{-0.3cm}

\subsection{\textsf{Geometric picture and Fine's theorem}}\vspace{-0.3cm}

According to \textit{Fine's theorem}~\cite{Fi82a}, a behavior $\vec{P}$ admits a local hidden-variable model consistent with Eq.~(\ref{Bell-factorisation}) if and only if it satisfies both the non-signaling conditions Eq.\,(\ref{NS}) and the Bell-CHSH inequalities Eq.\,(\ref{CHSH}). Any violation of these conditions rules out a purely common-cause explanation, leaving a non-local mechanism as the only viable account. See Fig.~\ref{Bell-DAGs} (left vs right).


In other words, the \textit{local polytope} $\mathcal{P}_{\scriptscriptstyle Loc}$ is fully specified by a set of linear constraints
\begin{eqnarray}\label{loc-polytope}
\mathcal{P}_{\scriptscriptstyle Loc}&\ \Leftrightarrow\ &\bigg\{{\,\Delta=0\,,\atop \ S\leqslant2\,,}
\end{eqnarray}
which consist of four equalities and eight inequalities.
Similarly, the \textit{non-signaling polytope} $\mathcal{P}_{\scriptscriptstyle NS}$ is characterized by the four equality constraints
\begin{eqnarray}\label{NS-polytope}
\mathcal{P}_{\scriptscriptstyle NS} &\ \Leftrightarrow\ & \Delta=0 \,.
\end{eqnarray}

This yields two non-trivial subsets of the full \textit{correlation space (polytope)} $\mathcal{P}$, forming the natural hierarchy
\begin{eqnarray}\label{PPR}
\mathcal{P}_{\scriptscriptstyle Loc}\subsetneq \mathcal{P}_{\scriptscriptstyle NS} \subsetneq \mathcal{P} \subsetneq \mathbb{R}^{16}\,.
\end{eqnarray}
See Fig.\,\ref{Fig-Polytopes} for illustration.
Note that these polytopes can equivalently be characterized by their vertices: $\mathcal{P}_{\scriptscriptstyle Loc}$ is generated by the 16 local deterministic strategies~\cite{Fi82a}, $\mathcal{P}_{\scriptscriptstyle NS}$ further includes the 8 external PR-box vertices~\cite{PoRo94}, and $\mathcal{P}$ is spanned by all 256 deterministic strategies. A detailed account is provided in the \textbf{Methods} section.
\begin{figure}[t]
\centering
\includegraphics[width=0.95\columnwidth]{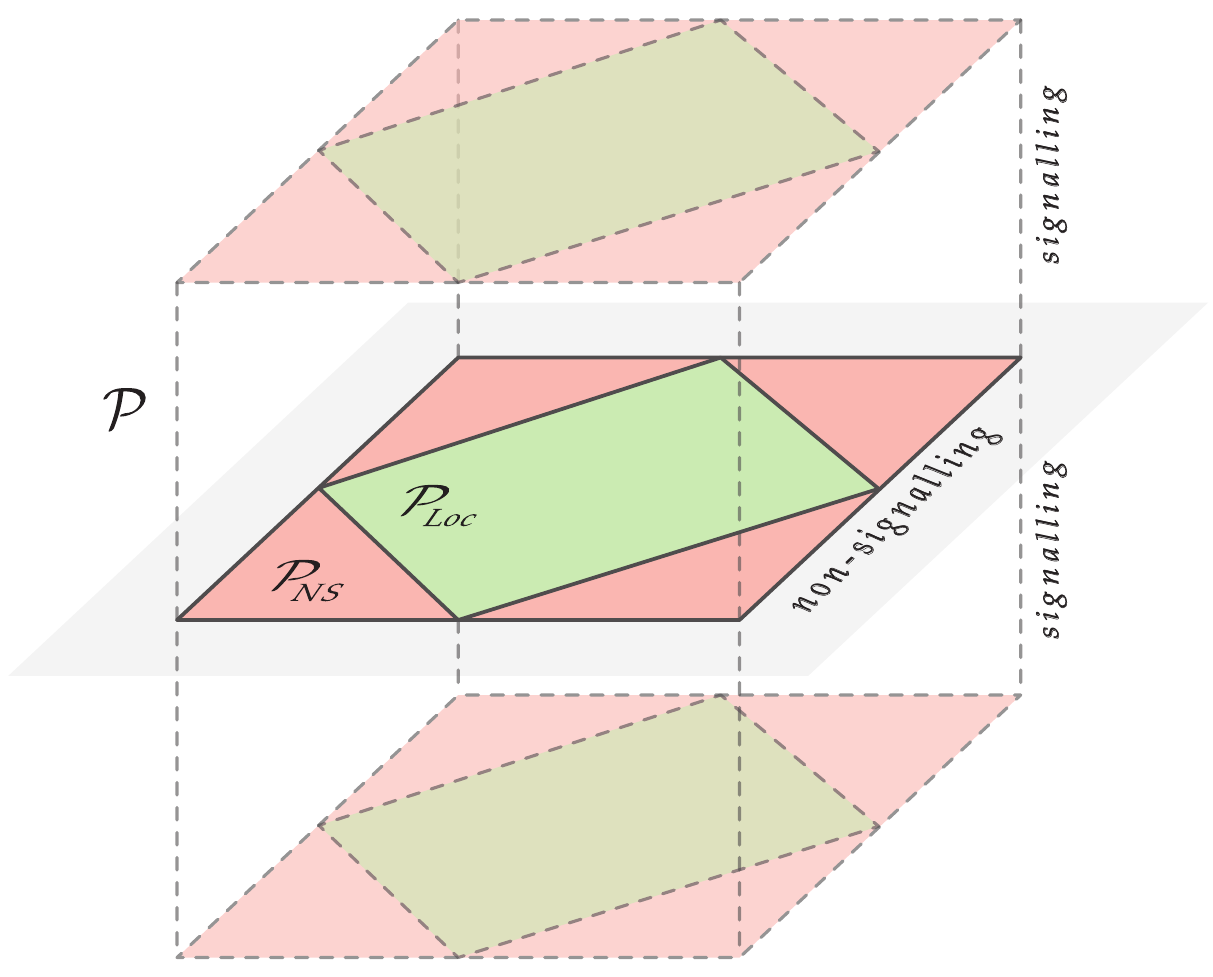}\caption
{\label{Fig-Polytopes}{\bf\textsf{Geometry of behaviors.}} Schematic illustration of the nested structure of convex polytopes embedded in $\mathbb{R}^{\scriptscriptstyle 16}$. Each relaxation of constraints enlarges the set: from the \textit{local polytope} $\mathcal{P}_{\scriptscriptstyle Loc}$ (8-dimensions) to the \textit{non-signaling polytope} $\mathcal{P}_{\scriptscriptstyle NS}$ (8-dimensions), and finally to the full \textit{correlation space (polytope)} $\mathcal{P}$ (12-dimensions).}
\end{figure}

\subsection{\textsf{Measuring locality \& non-signaling}}\vspace{-0.3cm}

When locality cannot be maintained---particularly when this is at odds with intuition---the natural question is not merely whether the observed correlations admit a local explanation, but rather how strongly non-local mechanisms must be invoked. This question is especially pertinent when there is suspicion that signaling might be responsible for apparent locality violations.

In the following discussion, we regard non-locality as an exceptional resource, invoked only when necessary to explain the observed correlations, whereas common causes,  as shown in Fig.~\ref{Bell-DAGs}, left, are assumed to be freely available as the conventional sources of correlation. This viewpoint naturally motivates treating non-locality as a (possibly small) correction to predominantly local content. Technically, for a given behavior of interest $\vec{P}$, we therefore consider all possible convex decompositions of the form
\begin{eqnarray}\label{decompL}
\vec{P} &=& p\cdot \vec{P}^{\scriptscriptstyle\,L} + (1-p)\cdot \vec{P}^{\scriptscriptstyle\,NL}\,,
\end{eqnarray}
where $\vec{P}^{\scriptscriptstyle\,L}\in\mathcal{P}_{\scriptscriptstyle\,Loc}$ is a local behavior, $\vec{P}^{\scriptscriptstyle\,NL}\in\mathcal{P}$ is any other (possibly non-local) behavior, and $0 \leqslant p \leqslant 1$. Our aim is to determine the maximal possible contribution of the local component.

\begin{definition}[Measure of locality]\label{Def-muL}\ \\
For a given observed behavior $\vec{P}$, the measure of its local component is defined as
\begin{eqnarray}\label{muL}
\mu_{\scriptscriptstyle  L}(\vec{P})&:=&\underset{\scriptscriptstyle decomp.\,(\ref{decompL})}{\textnormal{max}}\ p\,,
\end{eqnarray}
where the maximum is taken over all convex decompositions of the form in Eq.\,(\ref{decompL}).\end{definition}

\noindent Sometimes referred to as the \textit{local fraction} or \textit{local content} of a behavior $\vec{P}$, this notion was introduced in studies of quantum correlations~\cite{ElPoRo92,Ha91,BaKePi06,CoRe08,CoRe16,PoBrGi12} and later extended to arbitrary non-signaling statistics~\cite{BlPoYeGaBo21}.\footnote{The idea was first put forward in Ref.~\cite{Ha91}, where it was used to derive Bell inequalities. In parallel, Elitzur et al.~\cite{ElPoRo92} introduced it as a measure of locality, demonstrating that the local fraction vanishes in the limit of infinitely many measurement settings, with further refinements provided in Refs.~\cite{BaKePi06,CoRe08,CoRe16}. Remarkably, subsequent studies managed to compute the local fraction explicitly for the quantum statistics of any pure two-qubit state~\cite{PoBrGi12}. Importantly, the concept also applies beyond quantum mechanics: in the two-setting, two-outcome Bell scenario, Ref.~\cite{BlPoYeGaBo21} demonstrated that the local fraction is directly related to the violation of the Bell-CHSH inequalities under the non-signaling assumption, cf. \textbf{Theorem~\ref{Theorem-muL-e}}.} In this work, we dispense with prior assumptions to establish an unconstrained framework for quantifying the extent to which the observed statistics can be accounted for by a local explanation. 

A distinct advantage of the \textit{measure of locality} $\mu_{\scriptscriptstyle L}$ lies in its intuitive interpretation, captured by the following:
\begin{eqnarray}\nonumber
\hspace{0.45cm}\parbox{0.9\columnwidth}{\myuline{\textbf{Question}} (Locality)\vspace{0.05cm}\\ \textit{What is the maximal proportion of trials in a repeated Bell experiment that can be explained by a local model while still reproducing the observed correlations?}}&
\end{eqnarray}
Of course, in the remaining fraction of trials, some form of non-local mechanism must necessarily be involved. Our concern here is not with the particular mechanism, since the manifestation of non-locality may differ or even remain unknown across experimental contexts. What matters is that the question of \textit{"how often must non-locality minimally occur"} can still be posed in a clear and meaningful way.

To demonstrate its generality, this approach can also be adapted to quantify the extent of signaling present in a given behavior. Specifically, the signaling content of $\vec{P}$ can be characterized by considering convex decompositions of the form
\begin{eqnarray}\label{decompNS}
\vec{P}&=&p\cdot\vec{P}^{\scriptscriptstyle{\,NS}}+(1-p)\cdot\vec{P}^{\scriptscriptstyle{\,S}},
\end{eqnarray}
where $\vec{P}^{\scriptscriptstyle\,NS}\in\mathcal{P}_{\scriptscriptstyle NS}$ is a non-signaling behavior, $\vec{P}^{\scriptscriptstyle\,S}\in\mathcal{P}$ is any other (possibly signaling) behavior, and $0 \leqslant p \leqslant 1$. In direct analogy with the local case, we define:
\begin{definition}[Measure of non-signaling]\label{Def-muNS}\ \\
For a given observed behavior $\vec{P}$, we define the measure of its non-signaling component as
\begin{eqnarray}\label{muNS}
\mu_{\scriptscriptstyle  NS}(\vec{P})&:=&\underset{\scriptscriptstyle decomp.\,(\ref{decompNS})}{\textnormal{max}}\ p\,,
\end{eqnarray}
where the maximum is taken over all possible convex decompositions of the form in Eq.\,(\ref{decompNS}).
\end{definition}
\noindent Analogously, this may be termed the \textit{non-signaling fraction} or \textit{non-signaling content} of a behavior~$\vec{P}$. This case, too, admits a simple interpretation. If signaling is regarded as a costly resource, then the \textit{measure of non-signaling}~$\mu_{\scriptscriptstyle NS}$ answers the following question:\footnote{We note that, although both locality and signaling admit statistical definitions, via Eqs.\,(\ref{Bell-factorisation}) and (\ref{NS}), only locality allows for a deeper causal interpretation. This distinction leads to a subtle shift in emphasis: in the case of locality one may ask about the existence of a partial local model, whereas for non-locality one can only ask about the existence of a suitable statistical decomposition. Cf. footnote~\ref{footnote-1}.}
\begin{eqnarray}\nonumber
\hspace{0.35cm}\parbox{0.94\columnwidth}{\myuline{\textbf{Question}} (Non-signaling)\vspace{0.05cm}\\ \textit{What is the maximal proportion of trials in a repeated Bell experiment that can be accounted for by non-signaling statistics while reproducing the observed correlations?}}&
\end{eqnarray}

To recapitulate, the measures $\mu_{\scriptscriptstyle L}$ and $\mu_{\scriptscriptstyle NS}$ offer a unified operational framework for quantifying the indispensability of non-local and signaling resources in reproducing the observed behaviors. Within this framework, we derive explicit expressions for the general case of arbitrary, potentially signaling correlations in the two-setting, two-outcome Bell scenario.
\vspace{-0.3cm}

\subsection{\textsf{Main result}}\vspace{-0.3cm}

As a starting point, we recall a result from Ref.~\cite{BlPoYeGaBo21} concerning non-signaling statistics, presented here in a form adapted to our purposes. For clarity of exposition, all proofs and the specific vectors referenced in the following discussion are presented explicitly in the \textbf{Supplementary Information} (see also the \textbf{Methods} section).

\begin{theorem}\label{Theorem-muL-e}\ \\
Given a \myuline{non-signaling} behavior $\vec{P}$ in a two-setting, two-outcome Bell experiment, the \myuline{measure of locality} is linked to the Bell-CHSH expressions Eqs.\,(\ref{CHSH}) and (\ref{S1}-\ref{S4}), and takes the form
\begin{eqnarray}\hspace{-0.7cm}\mu_{\scriptscriptstyle L}(\vec{P})&=&\min\big\{1\,,\tfrac{1}{2}(4-S)\big\}\ =\  \min_{\scriptscriptstyle{i\text{\,=\,}1,...,8}}\big\{1\,,\,\vec{e}_i^{\scriptscriptstyle{\,T}}\!\cdot\vec{P}\,\big\}\,,\label{muL-e}
\end{eqnarray}
where the $\vec{e}_i$ are 8 specific vectors in $\mathbb{R}^{16}$, each with entries restricted to 0 or 1 (called the e-vectors).
\end{theorem}
\noindent This provides a clear interpretation of the degree of violation of the Bell-CHSH inequalities for non-signaling statistics. In the second equality, we recast the expression to underscore its resemblance to the subsequent analysis. The key observation is that the measure $\mu_{\scriptscriptstyle L}$ can be neatly expressed as the minimum over a family of eight simpler linear expressions, defined as projections of the behavior $\vec{P}$ onto specific directions (the \textit{e-vectors}) in the correlation space. In practice, each expression just selects and adds together certain terms of the behavior $\vec{P}$, as indicated by the non-zero entries of $\vec{e}_i$. These sums can be thought of as simple snapshots of $\vec{P}$ taken along particular directions in correlation space.

In a similar fashion, maximal signaling can be characterized by the minimum of simple linear expressions associated with specific directions (the \textit{f-vectors}) in the correlation space.

\begin{lemma}\label{Lemma-muNS-f}\ \\
The quantity of \myuline{maximal signaling} for a given behavior $\vec{P}$, as defined in Eqs.\,(\ref{NS}) and (\ref{NS-AB0}-\ref{NS-BA1}), can be written in the form
\begin{eqnarray}\label{NS-f}
\Delta&=&1-\min_{\scriptscriptstyle{j\text{\,=\,}1,...,8}}\ \vec{f}_j^{\scriptscriptstyle{\,T}}\!\cdot\vec{P}\,,
\end{eqnarray}
where the $\vec{f}_j$ are 8 specific vectors in $\mathbb{R}^{16}$, each with entries restricted to 0 or 1 (called the f-vectors).
\end{lemma}

Our aim is to move beyond non-signaling statistics in the Bell scenario and consider the full range of arbitrary behaviors $\vec{P}$. One might naively expect that \textbf{Theorem~\ref{Theorem-muL-e}} would suffice, with a modest adjustment to account for signaling in the spirit of \textbf{Lemma~\ref{Lemma-muNS-f}}. This intuition, however, proves to be misleading.

The main result of this work is a closed-form expression for the measure of locality in \textbf{Definition~\ref{Def-muL}}, which we present in the following theorem.

\begin{theorem}\label{theoremL}\ \\
For a given behavior $\vec{P}$ in a two-setting, two-outcome Bell experiment, the \myuline{measure of locality} takes the form
\begin{eqnarray}\label{muL-solution}
\mu_{\scriptscriptstyle  L}(\vec{P})&=&\min_{\vec{q}\,\in\,\mathbb{Q}}\ \vec{q}^{\scriptscriptstyle{\,T}}\!\cdot\vec{P}\,,
\end{eqnarray}
where $\mathbb{Q}$ denotes the set of 128 specific vectors in $\mathbb{R}^{16}$, each with entries restricted to 0 or 1.
\end{theorem}
\noindent Accordingly, the full solution follows the structure outlined above, appearing as the minimum of a family of simpler linear measures, each aligned with a specific direction in the correlation space. Remarkably, however, the solution reveals an unexpected degree of complexity, as all 128 vectors in the set $\mathbb{Q}$ prove essential. 

It is very instructive to compare the above result with the parallel measure of non-signaling in \textbf{Definition~\ref{Def-muNS}}, providing a clearer perspective on the underlying structures in the problem.

\begin{theorem}\label{theoremNS}\ \\
For a given behavior $\vec{P}$ in a two-setting, two-outcome Bell experiment, the \myuline{measure of non-signaling} takes the form
\begin{eqnarray}\label{muNS-solution}
\mu_{\scriptscriptstyle  NS}(\vec{P})&=&\min_{\vec{q}\,\in\,\mathbb{S}}\ \vec{q}^{\scriptscriptstyle{\,T}}\!\cdot\vec{P}\,,
\end{eqnarray}
where $\mathbb{S}$ denotes a set of 120 special vectors in $\mathbb{R}^{16}$, each with entries restricted to 0, 1 or 2.
\end{theorem}
\noindent Here again, the complexity of the solution persists, as all 120 vectors in the set $\mathbb{S}$ contribute to the result. This indicates that the challenge of finding the local fraction stems from the intricate internal structure of the non-signaling space itself. Geometrically, when examining the nested hierarchy in Eq.\,(\ref{PPR}), the difficulty of the problem lies in characterizing the embedding $\mathcal{P}_{\scriptscriptstyle NS} \subsetneq \mathcal{P}$ (\textbf{Theorem~\ref{theoremNS}}), which in turn extends to $\mathcal{P}_{\scriptscriptstyle Loc} \subsetneq \mathcal{P}$ (\textbf{Theorem~\ref{theoremL}}). By contrast, the embedding $\mathcal{P}_{\scriptscriptstyle Loc} \subsetneq \mathcal{P}_{\scriptscriptstyle NS}$ (\textbf{Theorem~\ref{Theorem-muL-e}}) is much simpler, illustrating how broader embeddings entail a richer underlying geometry.

A concise proof of \textbf{Theorem~\ref{theoremL}}, obtained via a linear-programming approach, is outlined in the \textbf{Methods} section. All proofs, including that of \textbf{Theorem~\ref{theoremNS}}, as well as the explicit listing of the vectors discussed, are provided in the \textbf{Supplementary Information}.
\vspace{-0.3cm}

\subsection{\textsf{Structure of the solutions}}\vspace{-0.3cm}
We now unfold the full structure of the solutions defined by the sets $\mathbb{Q}$ and $\mathbb{S}$, revealing the complexity of the problem, the underlying symmetry, and the connection to previous partial results.\vspace{0.2cm}

\myuline{\textbf{\textit{Classification}}}. The solution set $\mathbb{Q}$, associated with the \textit{measure of locality} $\mu_{\scriptscriptstyle  L}$ (\textbf{Theorem~\ref{theoremL}}), admits a natural partition into five categories (see Fig.~\ref{Fig-Solution-Summary}):
\begin{eqnarray}\nonumber
\hspace{-1cm}\!\mathbb{Q}&=&\{\vec{f}_1,...\,,\vec{f}_8\}\,\cup
\,\{\vec{g}_1,...\,,\vec{g}_{16}\}\,\cup\,\{\vec{t}_1,...\,,\vec{t}_{32}\}\,\cup\\\label{Q}
&&\{\vec{s}_1,...\,,\vec{s}_{64}\}\,\cup
\,\{\vec{e}_1,...\,,\vec{e}_8\}\,.
\end{eqnarray}
Similarly, for the \textit{non-signaling measure} $\mu_{\scriptscriptstyle  NS}$ (\textbf{Theorem~\ref{theoremNS}}), the corresponding solution set $\mathbb{S}$ splits into four categories (see Fig.~\ref{Fig-Solution-Summary}):
\begin{eqnarray}\nonumber
\hspace{-1cm}\!\mathbb{S}&=&\{\vec{f}_1,...\,,\vec{f}_8\}\,\cup
\,\{\vec{g}_1,...\,,\vec{g}_{16}\}\,\cup\,\{\vec{t}_1,...\,,\vec{t}_{32}\}\,\cup\\\label{S}
&&\{\vec{z}_1,...\,,\vec{z}_{64}\}\,.
\end{eqnarray}
The above grouping organizes the vectors according to the number of their \textit{non-zero entries}: the $\vec{f}$\,’s have 4, the $\vec{g}$\,’s have 5, the $\vec{t}$\,’s, $\vec{s}$\,’s, and $\vec{z}$\,’s each have 6, and the $\vec{e}$\,’s have 8 non-zero entries. 

Notably, the solution sets $\mathbb{Q}$ and $\mathbb{S}$ in Eqs.\,(\ref{Q}) and (\ref{S}) overlap on the $f$-, $g$-, and $t$-vectors. The difference between the $s$- and $z$-vectors is minor, as each $s$- vector differs from the corresponding $z$-vector only by replacing one of its 1’s with a 2. The essential distinction lies in the additional class of $e$-vectors, which appear exclusively in the solution set $\mathbb{Q}$. See Eq.\,(\ref{representatives}) for a representative of each class and \textbf{Suppl. Inf.} (Part A) for the complete list of vectors.\vspace{0.2cm}

\myuline{\textbf{\textit{Symmetries}}}. The problem considered in this paper is invariant under certain permutations of the indices $a$, $b$, $x$, and $y$ (or, equivalently, of the composite labels $ab|xy$). The elementary symmetry operations are as follows~\cite{Sc19,BrCaPiScWe14}:
\begin{eqnarray}\label{sym-1}
a&\ \leftrightarrow\ &1-a\,,\\\label{sym-2}
b&\ \leftrightarrow\ &1-b\,,\\\label{sym-3}
x&\ \leftrightarrow\ &1-x\,,\\\label{sym-4}
y&\ \leftrightarrow\ &1-y\,,\\\label{sym-5}
(x,a)&\ \leftrightarrow\ &(y,b)\,.
\end{eqnarray}
These transformations define 32 symmetries that reflect the fundamental invariances of the problem, namely, that both \textit{locality} and \textit{non-signaling} are preserved under outcome relabeling [Eqs.\,(\ref{sym-1})-(\ref{sym-2})], setting relabeling [Eqs.\,(\ref{sym-3})-(\ref{sym-4})], and party exchange (Alice-Bob symmetry) [Eq.\,(\ref{sym-5})].
Together, these operations constitute the automorphism group of both the \textit{local polytope} $\mathcal{P}_{\scriptscriptstyle Loc}$ and the \textit{non-signaling polytope} $\mathcal{P}_{\scriptscriptstyle NS}$.

Exploiting these symmetries allows the solution sets to be further simplified. Each class of \textit{f}-, \textit{g}-, and \textit{e}-vectors forms a single orbit under the group generated by Eqs.\,(\ref{sym-1})-(\ref{sym-5}). The class of \textit{t}-vectors consists of two orbits of size 16, while the \textit{s}- and \textit{z}-vector classes each consist of three orbits of sizes 16, 16, and 32.  
Corresponding representative vectors for each class are given below:
\begin{eqnarray}\nonumber
\!\!\!\!\text{$\vec{f}$\,'s:}&\!\!&\ \ \ \begin{array}{rcl}
\vec{f}_\alpha^{\scriptscriptstyle{\,T}} &=& (1,1,0,0,0,0,1,1,0,0,0,0,0,0,0,0)\,,
\end{array}\nonumber\\[5pt]\nonumber
\!\!\!\!\text{$\vec{g}$\,'s:}&\!\!&\ \ \ \begin{array}{rcl}
\vec{g}_\beta^{\scriptscriptstyle{\,T}} &=& (1, 0, 0, 0, 0, 0, 1, 1, 0, 1, 0, 1, 0, 0, 0, 0)\,, 
\end{array}\vspace{10pt}\nonumber\\[5pt]\nonumber
\!\!\!\!\text{$\vec{t}$\,'s:}&\!\!&\left\{\ \begin{array}{rcl}
\vec{t}_\beta^{\scriptscriptstyle{\,T}} &=& (1, 1, 0, 0, 0, 0, 1, 0, 1, 1, 0, 0, 0, 0, 0, 1)\,, \\[5pt]
\vec{t}_\gamma^{\scriptscriptstyle{\,T}} &=& (1, 1, 0, 0, 0, 0, 1, 0, 0, 0, 1, 1, 0, 1, 0, 0)\,, 
\end{array}\right.\nonumber\\[5pt]\nonumber
\!\!\!\!\text{$\vec{s}$\,'s:}&\!\!&\left\{\ \begin{array}{rcl}
\vec{s}_\beta^{\scriptscriptstyle{\,T}} &=& (1, 1, 1, 0, 0, 0, 1, 0, 0, 1, 0, 0, 0, 0, 0, 1)\,, \\[5pt]
\vec{s}_\gamma^{\scriptscriptstyle{\,T}} &=& (1, 1, 1, 0, 0, 0, 0, 1, 0, 0, 0, 1, 1, 0, 0, 0)\,, \\[5pt]
\vec{s}_\delta^{\scriptscriptstyle{\,T}} &=& (1, 1, 1, 0, 0, 0, 1, 0, 0, 0, 0, 1, 0, 1, 0, 0)\,,
\end{array}\right.\vspace{10pt}\nonumber\\[5pt]\nonumber
\!\!\!\!\text{$\vec{z}$\,'s:}&\!\!&\left\{\ \begin{array}{rcl}
\vec{z}_\beta^{\scriptscriptstyle{\,T}} &=& (2, 1, 1, 0, 0, 0, 1, 0, 0, 1, 0, 0, 0, 0, 0, 1)\,, \\[5pt]
\vec{z}_\gamma^{\scriptscriptstyle{\,T}} &=& (2, 1, 1, 0, 0, 0, 0, 1, 0, 0, 0, 1, 1, 0, 0, 0)\,, \\[5pt]
\vec{z}_\delta^{\scriptscriptstyle{\,T}} &=& (2, 1, 1, 0, 0, 0, 1, 0, 0, 0, 0, 1, 0, 1, 0, 0)\,,
\end{array}\right.\vspace{10pt}\nonumber\\[5pt]\nonumber
\!\!\!\!\text{$\vec{e}$\,'s:}&\!\!&\ \ \ \ \,\begin{array}{rcl}
\vec{e}_\alpha^{\scriptscriptstyle{\,T}} &=& (1, 0, 0, 1, 0, 1, 1, 0, 0, 1, 1, 0, 0, 1, 1, 0)\,, 
\end{array}\vspace{10pt}
\\\label{representatives}
\end{eqnarray}
where the labels $\alpha$, $\beta$, $\gamma$, and $\delta$ distinguish the orbits and denote their respective cardinalities of 8, 16, 16, and 32. 
\vspace{0.2cm}



\myuline{\textbf{\textit{Relation to previous results}}}. A closer inspection reveals that several components of the solution sets $\mathbb{Q}$ and $\mathbb{S}$ in Eqs.\,(\ref{Q}) and (\ref{S}) can be associated with well-known quantities.
Specifically, the minimum over the class of \textit{f-measures} coincides with the \textit{maximal signaling} $\Delta$ defined in Eq.\,(\ref{NS}); see \textbf{Lemma~\ref{Lemma-muNS-f}}. This follows from the fact that each individual \textit{f-measure} can be expressed in terms of the corresponding signaling signature $\Delta_{\scriptscriptstyle i}$ given in Eqs.~(\ref{NS-AB0})-(\ref{NS-BA1}); see Eqs.~(\ref{L1-f12})-(\ref{L1-f78}) in \textbf{Suppl. Inf.} Thus, these \textit{f-measures} provide an alternative characterization of the non-signaling polytope $\mathcal{P}_{\scriptscriptstyle NS}$.
Likewise, the minimum over the class of \textit{e-measures} is directly related to the maximal \textit{Bell-CHSH expression} $S$ defined in Eq.\,(\ref{CHSH}); see \textbf{Theorem~\ref{Theorem-muL-e}} and Ref.~\cite{BlPoYeGaBo21}. Moreover, each individual \textit{e-measure} can be expressed in terms of the corresponding $S_{\scriptscriptstyle i}$ given in Eqs.~(\ref{S1})-(\ref{S4}); see Eqs.~(\ref{E12-S4})-(\ref{E78-S1}) in \textbf{Suppl. Inf.} This reflects the characterization of the local polytope $\mathcal{P}_{\scriptscriptstyle Loc}$ as a subset of the larger non-signaling polytope $\mathcal{P}_{\scriptscriptstyle NS}$.

Importantly, all remaining types, namely the \textit{g}-, \textit{t}-, \textit{s}-, and \textit{z}-vectors, contribute non-trivially only in the presence of signaling. Thus, the full solution reveals a much richer structure than mere Bell-CHSH expressions and signaling signatures (see Fig.~\ref{Fig-Solution-Summary}).\vspace{0.2cm}

\begin{figure}[t]
\centering
\includegraphics[width=\columnwidth]{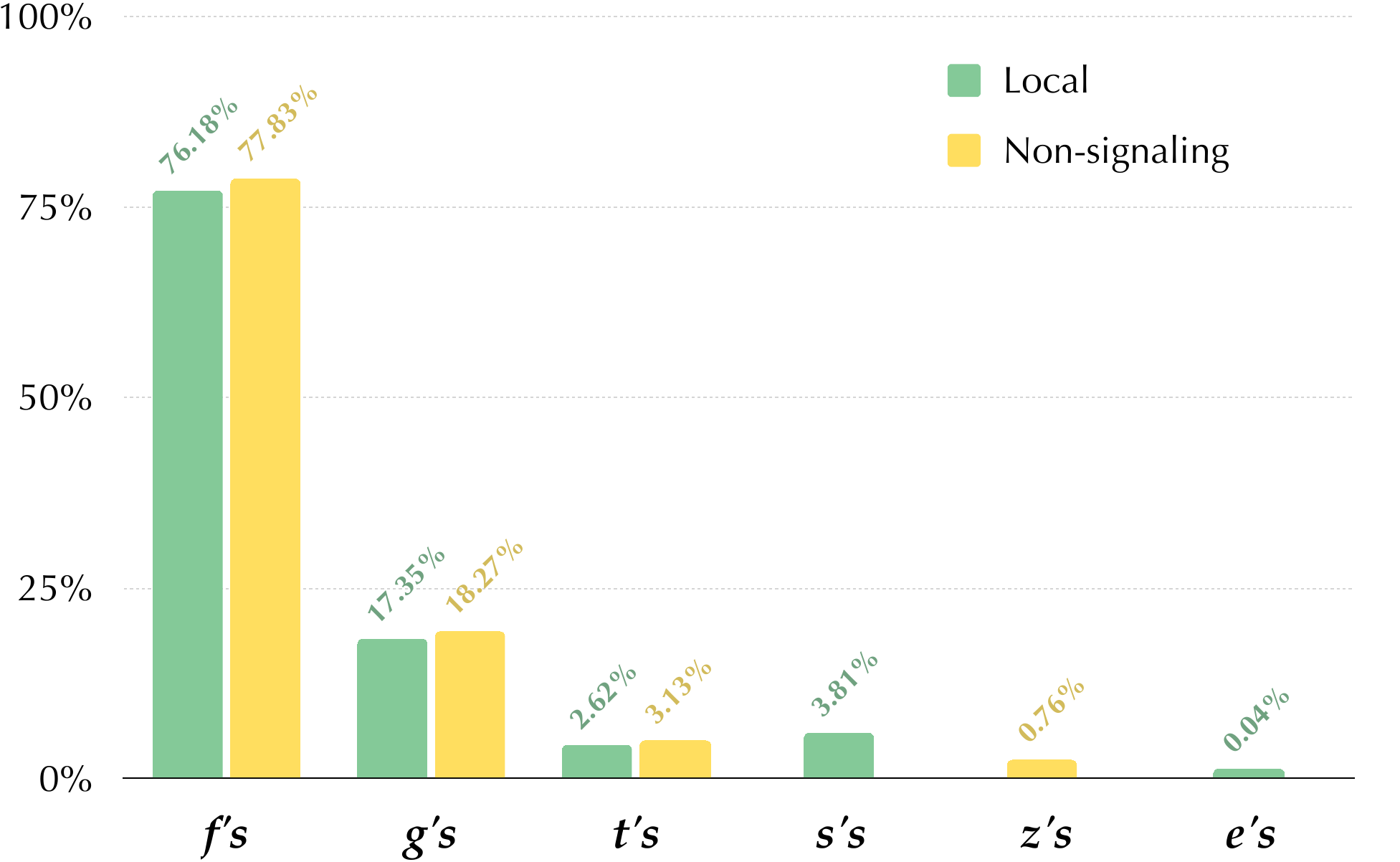}\caption
{\label{Fig-Weights}{\bf\textsf{Significance of the vector measures.}} Numerical results showing the proportion of randomly sampled behaviors for which the solutions in Eqs.\,(\ref{muL-solution}) and (\ref{muNS-solution}) are uniquely attained by each vector type within the respective subsets of $\mathbb{Q}$ (green) and $\mathbb{S}$ (yellow). Evidently, the \textit{f-measures} (maximal signaling) dominate for most behaviors, although a considerable fraction is accounted for by the remaining ones. Interestingly, for a random behavior, the \textit{e-measures} (Bell-CHSH expressions) are the least likely to contribute to the quantification of locality, while various forms of signaling corrections (\textit{f}-, \textit{g}-, \textit{t}-, \textit{s}-, and \textit{z-measures}) play a more prominent role.}
\end{figure}

\myuline{\textbf{\textit{Non-redundancy of the solutions}}}. As already noted, none of the vectors in the  solution sets, Eqs.~(\ref{Q}) and (\ref{S}), is superfluous. Specifically, for every vector in $\mathbb{Q}$ and $\mathbb{S}$ there exist behaviors for which the corresponding measure uniquely attains its minimum in Eqs.~(\ref{muL-solution}) and (\ref{muNS-solution}), respectively.
Beyond their non-redundancy, the relative importance of the different vector measures can be quantified by analyzing how often each one yields the minimum in Eqs.\,(\ref{muL-solution}) and (\ref{muNS-solution}) over randomly sampled behaviors~$\vec{P}$.
To this end, we numerically estimated the corresponding frequencies as shown in Fig.~\ref{Fig-Weights}. The results indicate that all vectors contribute to the solution, though with markedly different prevalence, thereby highlighting the internal structure of the solution space.\vspace{0.2cm}

\myuline{\textbf{\textit{Geometrical insight}}}. The fractional measure reflects the position of the observed behavior $\vec{P}$ relative to both the target set and the entire correlation space, i.e., its distance from the \textit{local polytope} $\mathcal{P}_{\scriptscriptstyle Loc}$ (or the \textit{non-signaling polytope} $\mathcal{P}_{\scriptscriptstyle NS}$) and, no less importantly, its proximity to the boundary of the full \textit{correlation polytope} $\mathcal{P}$. Thus, solving the problem requires insight into the geometry of how $\mathcal{P}_{\scriptscriptstyle Loc}$ (or $\mathcal{P}_{\scriptscriptstyle NS}$) is embedded within $\mathcal{P}$. Notably, the task can be cast as a linear optimization problem and strong duality then allows an explicit solution (see \textbf{Methods} section for details). The problem can also be tackled directly via an analytic construction (see \textbf{Suppl. Inf.}). Both approaches yield  the same two solution sets,  $\mathbb{Q}$ and $\mathbb{S}$, in Eqs.\,(\ref{Q}) and (\ref{S}). They fully encode, through simple linear expressions, the embeddings $\mathcal{P}_{\scriptscriptstyle Loc} \subsetneq \mathcal{P}$ and $\mathcal{P}_{\scriptscriptstyle NS} \subsetneq \mathcal{P}$, respectively, within the fractional measure framework (\textbf{Theorem~\ref{theoremL}} and \textbf{Theorem~\ref{theoremNS}}). This substantially extends Ref.~\cite{BlPoYeGaBo21}, which characterized the simpler case of the embedding $\mathcal{P}_{\scriptscriptstyle Loc} \subsetneq \mathcal{P}_{\scriptscriptstyle NS}$ (see \textbf{Theorem~\ref{Theorem-muL-e}}).

The geometry of $\mathcal{P}_{\scriptscriptstyle Loc}$ and $\mathcal{P}_{\scriptscriptstyle NS}$, in relation to the various solution vectors, is relatively straightforward; see Fig.~\ref{Fig-Polytopes}. For $\mathcal{P}_{\scriptscriptstyle NS}\subsetneq \mathcal{P}$, the eight $f$-vectors are normal vectors to the facets of this polytope, since the non-signaling condition can be re-expressed as $\vec{f}_j^{\scriptscriptstyle{\,T}}\!\cdot\vec{P}=1$ for all $j\text \,=\, 1,\ldots,8$ (see \textbf{Lemma~\ref{Lemma-muNS-f}}). Then, the embedding $\mathcal{P}_{\scriptscriptstyle Loc}\subsetneq \mathcal{P}_{\scriptscriptstyle NS}$ is characterized by the eight $e$-vectors, which are normal to the facets of $\mathcal{P}_{\scriptscriptstyle Loc}$ , which themselves are the $\mathcal{P}_{\scriptscriptstyle NS}$ facets, as truncated by the locality constraints $\vec{e}_i^{\,\scriptscriptstyle T}\!\cdot\vec{P}\ge 1$ for all $i=1,\ldots,8$ (see \textbf{Theorem~\ref{Theorem-muL-e}}). In other words, to define the non-signaling and local polytopes within $\mathcal{P}$, the \textit{f}- and \textit{e}-vectors are enough. Furthermore, $\mathcal{P}_{\scriptscriptstyle NS}$ has a particularly simple geometry relative to $\mathcal{P}_{\scriptscriptstyle Loc}$, with essentially one PR box per relevant local facet~\cite{PoRo94,Sc19,BrCaPiScWe14}. 
 The reason this does not suffice to compute the local (or non-signaling) fractions is that the answer depends on how a given behavior is situated between $\mathcal{P}_{\scriptscriptstyle Loc}$ (or $\mathcal{P}_{\scriptscriptstyle NS}$) and the full correlation space $\mathcal{P}$. Imagine extending a line segment through a behavior $\vec{P}$ until it intersects the boundaries of both $\mathcal{P}_{\scriptscriptstyle Loc}$ (or $\mathcal{P}_{\scriptscriptstyle NS}$) and $\mathcal{P}$. The fractional measure is obtained by choosing the direction along which $\vec{P}$ is closest to the target set $\mathcal{P}_{\scriptscriptstyle Loc}$ (or $\mathcal{P}_{\scriptscriptstyle NS}$) relative to its distance to the boundary of $\mathcal{P}$. Therefore, the geometric relationships are more intricate, which is why as many as 128 (or 120) vectors are needed to quantify this notion of proximity to the boundaries of  $\mathcal{P}_{\scriptscriptstyle Loc}$ (or $\mathcal{P}_{\scriptscriptstyle NS}$) within $\mathcal{P}$.

\section{\textsf{Discussion}}\vspace{-0.2cm}

The \textit{fractional measure} considered in this paper provides a principled way to quantify how strongly an observed behavior violates a given assumption. Beyond a binary pass/fail classification, this measure assigns a graded notion of “distance” from the corresponding constraint set by asking a simple question: \textit{what fraction of experimental runs can be accounted for by a model compatible with the assumption?} This approach is applicable to a wide range of problems, covering both causal notions and purely operational (statistical) ones, as long as the problem admits a convex-geometric characterization.

In this paper, we study the standard Bell scenario with two parties, two settings, and two outcomes. Our goal was to compare and explicitly compute for 
an arbitrary observed behavior $P_{\scriptscriptstyle ab|xy}$\,,\vspace{0.1cm}

\textit{\ (i)} the \textit{(fractional) measure of locality} $\mu_{\scriptscriptstyle L}$, and\vspace{0.1cm}

\textit{(ii)} the \textit{(fractional) measure of non-signaling} $\mu_{\scriptscriptstyle NS}$.\vspace{0.1cm}

\noindent Note, locality is inherently causal, while signaling is defined operationally from observed statistics. Yet both admit a convex-geometric formulation, as used in \textbf{Definition~\ref{Def-muL}} and \textbf{Definition~\ref{Def-muNS}}. The solution leverages the geometry of the local polytope $\mathcal{P}_{\scriptscriptstyle L}$ and the non-signaling polytope $\mathcal{P}_{\scriptscriptstyle NS}$ to quantify the closeness of an observed behavior $P_{\scriptscriptstyle ab|xy}$ to each of them, subject to probability constraints that define the full correlation space $\mathcal{P}$. Crucially, the measure depends not only on the target sets themselves but also on their embedding within the overall correlation space, 
thereby adding substantial geometric complexity to the solution.

We show that the solution can be written as the minimum over a family of simpler 
vector measures, each obtained by projecting the observed behavior onto a specific direction in the correlation space. In full generality, the required set of directions is comprised of
\vspace{0.1cm}

\textit{\ (i)} 128 vectors for the measure $\mu_{\scriptscriptstyle L}$ in \textbf{Theorem~\ref{theoremL}}, and\vspace{0.1cm}

\textit{(ii)} 120  vectors for the measure $\mu_{\scriptscriptstyle NS}$ in \textbf{Theorem~\ref{theoremNS}}.\vspace{0.1cm}

\noindent The number of vectors may seem surprising, but none of the 128/120 directions is redundant. While some of the associated vector measures coincide with familiar quantities, such as the \textit{Bell-CHSH expressions} (\textbf{Theorem~\ref{Theorem-muL-e}}) and the \textit{signaling signatures} (\textbf{Lemma~\ref{Lemma-muNS-f}}), 
the remaining ones contribute genuinely non-trivial terms that are essential for the full solution. 
These additional directions capture geometric features that only become relevant in the signaling regime.


In a nutshell, in the presence of signaling it is difficult to disentangle genuine non-locality from signaling-induced effects, since apparent non-locality may partly be driven by signaling. Operationally, signaling implies non-locality, but the  converse is not true; the two notions are conceptually distinct. These differences become transparent in our analysis, in which we provide complete solutions for both fractional measures, make the full structure explicit, and give a ready-to-use procedure for practical calculation.\footnote{For an alternative approach based on the contextuality-by-default framework, see Refs.~\cite{DzKuLa15,KuDzLa15}. We note, however, that our treatment is conceptually distinct, and the resulting conclusions are therefore markedly different.}
Notably, the full picture goes well beyond the Bell-CHSH expressions $S_{\scriptscriptstyle i}$ ($e$-\textit{measures}) in Eqs.~(\ref{S1})-(\ref{S4}), even when supplemented by the signaling signatures $\Delta_{\scriptscriptstyle i}$ ($f$-\textit{measures}) in Eqs.~(\ref{NS-AB0})-(\ref{NS-BA1}). A complete description further involves the $g$-, $t$-, and $s$-\textit{measures}. Interestingly, we also observe a similar pattern when assessing the degree of non-signaling itself, as a complete characterization involves far more than the signatures $\Delta_{\scriptscriptstyle i}$ ($f$-\textit{measures}) that define signaling of the behavior. Capturing the full structure requires the $g$-, $t$-, and $z$-\textit{measures} as well. 


Our focus has been on the local fraction, because of the fundamental role of locality in Bell-inspired experiments concerning quantum structure in nature. No less importantly, the analysis of the (fractional) measure of non-signaling might shed new light on communication complexity in signaling scenarios. The surprisingly rich geometric structure of the non-signaling fraction solution certainly helps explain the intricacy of the local fraction solution, since, as it turns out, a substantial part of the apparent violation of locality can be attributed to signaling, which accounts for most of the non-trivial components in the solution sets (as evident from comparing $\mathbb{Q}$ and $\mathbb{S}$).

More generally, the present proposal of the \textit{fractional measure} is applicable beyond locality and signaling problems, because of the way it lends itself to a natural resource theoretic interpretation: \textit{if violating a given assumption comes with a cost, then such violations constitute a resource that should be used sparingly.} In this sense, the \textit{fractional measure} $\mu$ quantifies the required (minimal) amount of this resource. Our solution can be viewed as a characterization of possible simulation strategies via the simpler vector measures, followed by selecting the class that attains the (minimal) fraction. Geometrically, the problem can be formulated as vertex enumeration of a suitably constructed dual polytope.

Our results are supported by a transparent methodological route: in the \textbf{Methods} section, we cast the problem as a linear program and highlight the role of strong duality in obtaining the solution. The framework easily extends to arbitrary numbers of parties, settings, and outcomes. More generally, the methodology applies to a broad class of constraints, that is, any set of constraints that admit a convex-geometric formulation, including (fractional) measures of the causal strength of a given arrow (or set of arrows) in a causal diagram.



A key significance and motivation of our result is that it eliminates reliance on the non-signaling assumption in Bell non-locality analyses, enabling a consistent treatment of general (possibly signaling) behaviors, thereby going beyond the scope of Ref.~\cite{BlPoYeGaBo21}. 
In doing so, we offer a principled framework for standard Bell non-locality tests, where signaling may arise from imperfections or noise, which can complicate experimental interpretation~\cite{AdKh17,HeKaBlDrReVeSc16,SmKlCaBo25}. Possible applications also extend to treatment of for finite-sample effects~\cite{Gi14c,ElWe16,ReRoMaGi17,Gi23} , whose handling in standard analyses is often simplified by imposing non-signaling. More importantly, our result enables a systematic framework for studying experimental implementations in which disturbance cannot be excluded and signaling is therefore expected, opening the door to investigations beyond the standard, arguably overly restrictive, non-signaling Bell framework. This becomes increasingly relevant in scenarios without 
spacelike separation~\cite{EmLaNo14,KoBr13,Fr10,BrKoMaPaPr15,RiGiChCoWhFe16,ChCaAgDiAoGiSc18,AgPoPoMiGaSuPo22}, including modern implementations~\cite{RyBiBaBe25}, and motivates extensions of device-independent frameworks~\cite{BrCaPiScWe14,Sc19}. Finally, both signaling and disturbance are commonplace when using Bell-type causal diagrams in other sciences~\cite{CeDz18,BrFeHoDeObGiMo23,GaPoBlYeWo23,PeMa18,Pe09,SpGlSc00,AnPi09,RuIm15,HeRo20}.

In sum, historically, the interpretation of apparent violations of locality had to rely on elaborate experimental or statistical safeguards. We offer a technique to disentangle violations of locality from 
violations of non-signaling, based primarily on the assumption that such resources are costly. The present analysis revealed a rich structure in how locality or non-signaling can be violated and showed, importantly, that such violations are not all equal, some are more costly than others. 






\section{\textsf{Methods}}\vspace{-0.3cm}

Here, we provide a sketch of the proof of \textbf{Theorem~\ref{theoremL}}, based on a linear programming approach. See the \textbf{Supplementary Information} for details and proofs of the remaining results.

\vspace{-0.3cm}



\subsection{\textsf{Notation and Preliminaries}}\vspace{-0.3cm}

\textit{Behaviors.} Our analysis concerns non-negative vectors $\vec{P} = (P_{\scriptscriptstyle ab|xy}) \in \mathbb{R}^{16}$, subject to the normalization conditions $\sum_{\scriptscriptstyle a,b} P_{\scriptscriptstyle ab|xy} = 1$ for all $x,y \in \{0,1\}$. These vectors, referred to as \textit{behaviors} in the two-setting, two-outcome Bell scenario, represent conditional probabilities of obtaining outcomes $a,b \in \{0,1\}$ given the choice of measurement settings $x,y \in \{0,1\}$. Accordingly, the sixteen-component column vector 
\begin{eqnarray}\nonumber
\vec{P}&=&(P_{\scriptscriptstyle{00|00}}\,,P_{\scriptscriptstyle{01|00}}\,,P_{\scriptscriptstyle{10|00}},P_{\scriptscriptstyle{11|00}}\,,P_{\scriptscriptstyle{00|01}}\,,P_{\scriptscriptstyle{01|01}}\,,P_{\scriptscriptstyle{10|01}}\,,P_{\scriptscriptstyle{11|01}}\,,\\\nonumber
&&\ \ P_{\scriptscriptstyle{00|10}}\,,P_{\scriptscriptstyle{01|10}}\,,P_{\scriptscriptstyle{10|10}}\,,P_{\scriptscriptstyle{11|10}}\,,P_{\scriptscriptstyle{00|11}}\,,P_{\scriptscriptstyle{01|11}}\,,P_{\scriptscriptstyle{10|11}}\,,P_{\scriptscriptstyle{11|11}})^{\scriptscriptstyle{T}}\\\label{Pab|xy}\vspace{-0.5cm}
\end{eqnarray}
encodes the distribution of outcomes for the four possible combinations of measurement settings in the Bell experiment.

\textit{Local polytope.} The essence of Bell’s theorem is that not all admissible behaviors can be explained in a local manner (Fig.\,\ref{Bell-DAGs}, left). As shown by A.~Fine~\cite{Fi82a}, the set of local behaviors forms the \textit{local polytope} $\mathcal{P}_{\scriptscriptstyle Loc}$, spanned by \textit{local deterministic strategies}. These strategies can be labeled by $\text{16 = 4} \cdot \text{4}$ pairs of functions $i = (f,g)$, where $f,g : \{0,1\} \rightarrow \{0,1\}$, and are explicitly given by
\begin{eqnarray}\label{LDS}
(\vec{L}_{\scriptscriptstyle{fg}})_{\scriptscriptstyle{ab}|\scriptscriptstyle{xy}}&=&\delta_{a\text{\,=\,}f(x)}\!\cdot\delta_{b\text{\,=\,}g(y)}\,.
\end{eqnarray}

We arrange these strategies into a matrix, denoted by $L\equiv(\vec{L}_{\scriptscriptstyle{1}},...\,,\vec{L}_{\scriptscriptstyle{16}})$, whose columns correspond to the sixteen \textit{local deterministic strategies} $\vec{L}_i$. The explicit enumeration of all strategies is provided in the \textbf{Suppl. Inf.} (Part B). Consequently, any local behavior $\vec{P}^{\scriptscriptstyle\,L}$ can be expressed as a convex decomposition of the form\begin{eqnarray}\label{Loc-det-decomp}
\vec{P}^{\scriptscriptstyle{\,L}}&=&\sum_{\scriptscriptstyle{i=1}}^{\scriptscriptstyle{16}}\, \vec{L}_i\,{q_i}\ =\ L\,\vec{q}\ ,
\end{eqnarray}
where $\vec{q}$ is a vector of non-negative weights satisfying the normalisation condition $\sum_{\scriptscriptstyle{i=1}}^{\scriptscriptstyle{16}} q_i=1$. 


\vspace{-0.3cm}

\subsection{\textsf{Formulation as a Linear Program}}\vspace{-0.3cm}

We show that determining the \textit{local fraction} $\mu_{\scriptscriptstyle{L}}(\vec{P})$, as defined in Eq.\,(\ref{muL}), is equivalent to solving the following linear program (LP):
\begin{eqnarray}\label{LP-1}
\text{maximize}&&{\vec{1}}^{\scriptscriptstyle{T}}\!\cdot\vec{p}\,,\\\label{LP-2}
\text{subject to}&&{L}\,\vec{p}\,\leqslant\,\vec{P}\,,\\\label{LP-3}
&&\ \ \ \,\vec{p}\geqslant\,0\,,
\end{eqnarray}
where $\vec{p}\in\mathbb{R}^{16}$, and $L$ is the matrix comprised of sixteen \textit{local deterministic strategies} discussed above. We denote $\vec{1}:=(1,1\,,...\,,1)^{\scriptscriptstyle{T}}\in\mathbb{R}^{16}$ as a column vector with ones.

\begin{proof}[Justification]

Let us rewrite the local part of the decomposition in Eq.\,(\ref{decompL}) using the convex expansion introduced in Eq.\,(\ref{Loc-det-decomp}), i.e.,
\begin{eqnarray}\label{condition-1}
\vec{P}&=&q\cdot \underbrace{\sum_{\scriptscriptstyle i=1}^{\scriptscriptstyle16}\ \vec{L}_i\,q_i}_{\vec{P}^{\scriptscriptstyle{\,L}}}\ +\ (1-q)\cdot\vec{P}^{\scriptscriptstyle{\,NL}}\ .
\end{eqnarray}
This implies the existence of coefficients $p_i$, defined by $p_i := q q_i$, such that
\begin{eqnarray}\label{condition-2}
\vec{P}\ -\ \sum_{\scriptscriptstyle i=1}^{\scriptscriptstyle 16}\ \vec{L}_i\,p_i\ \geqslant\ 0&\ \ \text{and}\ \ &p_i\geqslant0\ .
\end{eqnarray}

We find that the converse also holds: the existence of coefficients $p_i$ satisfying Eq.\,(\ref{condition-2}) implies the decomposition in Eq.\,(\ref{condition-1}). This can be seen by rewriting
\begin{eqnarray}
\vec{P}&=&\sum_{\scriptscriptstyle i=1}^{\scriptscriptstyle16}\ \vec{L}_i\,p_i\ +\ \vec{P}\ -\ \sum_{\scriptscriptstyle i=1}^{\scriptscriptstyle16}\ \vec{L}_i\,p_i\\\nonumber
&=&q\cdot\underbrace{\sum_{\scriptscriptstyle i=1}^{\scriptscriptstyle 16}\ \vec{L}_i\,q_i}_{\vec{P}^{\scriptscriptstyle{\,L}}}\ +\ (1-q)\cdot\underbrace{\tfrac{1}{1-q}\,\Big(\vec{P}\ -\ \sum_{\scriptscriptstyle i=1}^{\scriptscriptstyle16}\ \vec{L}_i\,p_i\Big)}_{\vec{P}^{\scriptscriptstyle{\,NL}}}\,,
\end{eqnarray}
with $q := \sum_i p_i$ and $q_i := \nicefrac{p_i}{q}$. It is straightforward to verify that the vectors defined in this way, $\vec{P}^{\scriptscriptstyle{\,L}}$ and $\vec{P}^{\scriptscriptstyle{\,NL}}$, are properly normalized behaviors, since both $\vec{P}$ and the $\vec{L}_i$'s are normalized behaviors by construction.

This allows us to formulate the task of finding $\mu_{\scriptscriptstyle L}(\vec{P})$ in Eq.\,(\ref{muL}) as the following linear program:
\begin{eqnarray}
\text{maximize}&&{\sum}_{\scriptscriptstyle i=1}^{\scriptscriptstyle 16}\ p_i\\
\text{subject to}&&\vec{P}-{\sum}_{\scriptscriptstyle i=1}^{\scriptscriptstyle 16}\ \vec{L}_i\ p_i\geqslant0\ ,\\
&&\qquad\qquad\quad\ \   p_i\geqslant0\ ,
\end{eqnarray}
which, in compact form, corresponds to the LP given in Eqs.\,(\ref{LP-1})-(\ref{LP-3}).
\end{proof}
\vspace{-0.3cm}

\subsection{\textsf{Dual Linear Program \& Solution}}\vspace{-0.3cm}
We remark that the LP in Eqs.\,(\ref{LP-1})-(\ref{LP-3}) is not straightforward to handle, as its feasible region depends on the chosen behavior $\vec{P}$. This difficulty can be overcome by considering the \textit{dual linear program} (DLP), which takes a more convenient form
\begin{eqnarray}\label{DLP-1}
\text{minimize}&&{\vec{P}}^{\scriptscriptstyle{\,T}}\!\cdot\vec{q}\,,\\\label{DLP-2}
\text{subject to}&&{L}^{\scriptscriptstyle{T}}\,\vec{q}\,\geqslant\,\vec{{1}}\,,\\\label{DLP-3}
&&\quad\ \vec{q}\,\geqslant\,0\,.
\end{eqnarray}

This reformulation offers two main advantages:

\noindent\textit{(i)} unlike in the primal LP, the chosen behavior $\vec{P}$ now appears \textit{only} in the objective function, Eq.\,(\ref{DLP-1}), and\\
\noindent\textit{(ii)} the constraints now define a \textit{fixed} polyhedron in $\mathbb{R}^{16}$, specified by Eqs.\,(\ref{DLP-2})-(\ref{DLP-3}), which is \textit{the same} for all behaviors $\vec{P}$.

We can now solve the problem. By the \textit{Strong Duality Theorem}~\cite{BoVa04}, the solutions of the primal LP and its dual DLP coincide. This is ensured because the primal LP in Eqs.\,(\ref{LP-1})-(\ref{LP-3}) has a solution, as it maximizes a linear function over a compact domain.

Hence, the \textit{local fraction} $\mu_{\scriptscriptstyle L}(\vec{P})$ can equivalently be obtained from the DLP in Eqs.\,(\ref{DLP-1})-(\ref{DLP-3}). Since the optimum of a linear program is achieved at one of the vertices of its feasible region, we obtain
\begin{eqnarray}\label{mu-solution}
\mu_{\scriptscriptstyle L}(\vec{P})&=&\min_{\vec{q}_i}\ \vec{q}_i^{\scriptscriptstyle{\,T}}\!\cdot\vec{P}\ ,
\end{eqnarray}
where $\{\vec{q}_i\}$ is the set of vertices of a polyhedron defined by $\text{32 = 16\,+\,16}$ inequalities:
\begin{eqnarray}\label{DLP-poly-1}
&&{L}^{\scriptscriptstyle{T}}\,\vec{q}\,\geqslant\,\vec{{1}}\,,\\\label{DLP-poly-2}
&&\quad\ \vec{q}\,\geqslant\,0\,.
\end{eqnarray}
The vertices can be explicitly enumerated using standard methods of convex analysis~\cite{BoVa04}. We find that the relevant set of vertices, denoted by $\mathbb{Q} \equiv {\vec{q}_i}$, comprises 128 binary vectors (with entries 0 or 1), as stated in \textbf{Theorem~\ref{theoremL}}. A complete list of these vertices, together with additional discussion, is provided in the \textbf{Suppl. Inf.}.

\vspace{0.4cm}
\noindent{\small\bf\textsf{{Acknowledgments}}}\\
We acknowledge funding support for EP from the European Office of Aerospace Research and Development (EOARD) grant FA8655-23-1-7220, for CG from LOEWE grant LOEWE/5/A004/519/06/00.006(0004)E36, and for TN from EPSRC grant 2919480.

\vspace{0.3cm}
\noindent{\small\bf\textsf{Author contributions}}\\
MB developed the analytic proofs; TN carried out the computational work; PB developed the LP proofs and wrote the initial draft of the manuscript; PB, EP, and CG conceptualized the project; all authors helped with the mathematical work and revisions of the manuscript.

\vspace{0.3cm}
\noindent{\small\bf\textsf{Competing interests}}\\
The authors declare no competing interests.
\vspace{0.2cm}



\bibliography{CombQuant}

@article{ElWe16,
	author = {Elkouss, D. and Wehner, S.},
	journal = {npj Quantum Inf.},
	pages = {16026},
	title = {(Nearly) optimal {P} values for all {B}ell inequalities},
	url = {https://doi.org/10.1038/npjqi.2016.26},
	volume = {2},
	year = {2016}}

@article{Gi23,
	author = {Gill, R. D.},
	journal = {Appl. Math.},
	pages = {446},
	title = {Optimal Statistical Analyses of Bell Experiments},
	url = {https://doi.org/10.3390/appliedmath3020023},
	volume = {3},
	year = {2023}}

@article{RyBiBaBe25,
	author = {Rybotycki, T. and Bialecki, T. and Batle, J. and Bednorz, Adam},
	journal = {Adv. Quantum Technol.},
	pages = {2400661},
	title = {{V}iolation of {N}o-{S}ignaling on a {P}ublic {Q}uantum {C}omputer},
	url = {https://doi.org/10.1002/qute.202400661},
	volume = {8},
	year = {2025}}

@article{SmKlCaBo25,
	author = {Smania, M. and Kleinmann, M. and Cabello, A. and Bourennane, M.},
	journal = {Quantum},
	pages = {1760},
	title = {How to avoid (apparent) signaling in {B}ell tests},
	url = {https://doi.org/10.22331/q-2025-06-04-1760},
	volume = {9},
	year = {2025}}

@article{BrFeHoDeObGiMo23,
	author = {Bruza, P. D. and Fell, L. and Hoyte, P. and Dehdashti, S. and Obeid, A. and Gibson, A. and Moreira, C.},
	journal = {Cogn. Psychol.},
	pages = {101529},
	title = {Contextuality and context-sensitivity in probabilistic models of cognition},
	url = {https://doi.org/10.1016/j.cogpsych.2022.101529},
	volume = {140},
	year = {2023}}

@unpublished{Ch97,
	author = {Christof, T.},
	note = {{A}viable at: https://porta.zib.de},
	title = {{PORTA} software system ({PO}lyhedron {R}epresentation {T}ransformation {A}lgorithm)},
	url = {https://porta.zib.de},
	year = {1997}}

@article{BrKoMaPaPr15,
	author = {Brierley, S. and Kosowski, A. and Markiewicz, M. and Paterek, T. and Przysiezna, A.},
	journal = {Phys. Rev. Lett.},
	pages = {120404},
	title = {Nonclassicality of {T}emporal {C}orrelations},
	url = {https://doi.org/10.1103/PhysRevLett.115.120404},
	volume = {115},
	year = {2015}}

@article{ViRaCa25,
	author = {Vieira, C. and Ramanathan, R. and Cabello, A.},
	journal = {Nat. Commun.},
	pages = {4390},
	title = {Test of the physical significance of {B}ell non-locality},
	url = {https://doi.org/10.1038/s41467-025-59247-7},
	volume = {16},
	year = {2025}}

@article{FrCl72,
	author = {Freedman, S. J. and Clauser, J. F.},
	journal = {Phys. Rev. Lett.},
	pages = {938},
	title = {Experimental {T}est of {L}ocal {H}idden-{V}ariable {T}heories},
	url = {https://doi.org/10.1103/PhysRevLett.28.938},
	volume = {28},
	year = {1972}}

@article{AdKh17,
	author = {Adenier, G. and Khrennikov, A. Y.},
	journal = {Fortschr. Phys.},
	pages = {1600096},
	title = {Test of the no-signaling principle in the {H}ensen loophole-free {CHSH} experiment},
	url = {https://doi.org/10.1002/prop.201600096},
	volume = {65},
	year = {2017}}

@article{HeKaBlDrReVeSc16,
	author = {Hensen, B. and Kalb, N. and Blok, M. S. and Dr{\'e}au, A. E. and Reiserer, A. and Vermeulen, R. F. L. and Schouten, R. N. and Markham, M. and Twitchen, D. J. and Goodenough, K. and Elkouss, D. and Wehner, S. and Taminiau, T. H. and Hanson, R.},
	journal = {Sci. Rep.},
	pages = {30289},
	title = {Loophole-free {B}ell test using electron spins in diamond: second experiment and additional analysis},
	url = {https://doi.org/10.1038/srep30289},
	volume = {6},
	year = {2016}}

@book{BoVa04,
	author = {Boyd, S. and Vandenberghe, L.},
	publisher = {Cambridge University Press},
	title = {Convex Optimization},
	year = {2004}}

@article{BlGa24,
	author = {Blasiak, P. and Gallus, C.},
	journal = {Phil. Trans. R. Soc. A},
	pages = {20230005},
	title = {Comparing the cost of violating causal assumptions in {B}ell experiments: locality, free choice and arrow-of-time},
	url = {https://doi.org/10.1098/rsta.2023.0005},
	volume = {382},
	year = {2024}}

@inproceedings{WiCa17,
	author = {Wiseman, H. M. and Cavalcanti, E. G.},
	booktitle = {Quantum {$[$}Un{$]$}Speakables II: Half a Century of Bell's Theorem},
	editor = {Bertlmann, R. and Zeilinger, A.},
	pages = {119--142},
	publisher = {Springer},
	title = {Causarum {I}nvestigatio and the {T}wo {B}ell's {T}heorems of {J}ohn {B}ell},
	url = {https://doi.org/10.1007/978-3-319-38987-5_6},
	year = {2017}}

@article{GaPoBlYeWo23,
	author = {Gallus, C. and Pothos, E. M. and Blasiak, P. and Yearsley, J. M. and Wojciechowski, B. W.},
	journal = {Sci. Rep.},
	pages = {4394},
	title = {Bell correlations outside physics},
	url = {https://doi.org/10.1038/s41598-023-31441-x},
	volume = {13},
	year = {2023}}

@article{AgPoPoMiGaSuPo22,
	author = {Agresti, I. and Poderini, D. and Polacchi, B. and Miklin, N. and Gachechiladze, M. and Suprano, A. and Polino, E. and Milani, G. and Carvacho, G. and Chaves, R. and Sciarrino, F.},
	journal = {Sci. Adv.},
	pages = {eabm1515},
	title = {Experimental test of quantum causal influences},
	url = {https://doi.org/10.1126/sciadv.abm1515},
	volume = {8},
	year = {2022}}

@book{AnPi09,
	author = {Angrist, J. D. and Pischke, J.-S.},
	publisher = {Princeton University Press},
	title = {Mostly Harmless Econometrics: An Empiricist's Companion},
	year = {2009}}

@book{RuIm15,
	author = {Rubin, D. B. and Imbens, G. W.},
	publisher = {Cambridge University Press},
	title = {Causal Inference for Statistics, Social, and Biomedical Sciences: An Introduction},
	year = {2015}}

@book{HeRo20,
	author = {Hern{\'a}n, M. A. and Robins, J. M.},
	publisher = {Chapman \& Hall/CRC},
	title = {Causal Inference: What If},
	year = {2020}}

@article{BlPoYeGaBo21,
	author = {Blasiak, P. and Pothos, E. M. and Yearsley, J. M. and Gallus, C. and Borsuk, E.},
	journal = {Proc. Natl. Acad. Sci. USA (PNAS)},
	note = {\\\href{https://www.eurekalert.org/news-releases/671079}{\textbf{EurekAlert! (2021):} \textit{We know the cost of free choice and locality - in physics and not only}}},
	pages = {e2020569118},
	title = {Violations of locality and free choice are equivalent resources in {B}ell experiments},
	url = {https://doi.org/10.1073/pnas.2020569118},
	volume = {118},
	year = {2021}}

@article{ChCaAgDiAoGiSc18,
	author = {Chaves, R. and Carvacho, G. and Agresti, I. and Di Giulio, V. and Aolita, L. and Giacomini, S. and Sciarrino, F.},
	journal = {Nature Phys.},
	pages = {291},
	title = {Quantum violation of an instrumental test},
	url = {https://doi.org/10.1038/s41567-017-0008-5},
	volume = {14},
	year = {2018}}

@article{Ha91,
	author = {Hardy, L.},
	journal = {Phys. Lett. A},
	pages = {21},
	title = {A new way to obtain {B}ell inequalities},
	url = {https://doi.org/10.1016/0375-9601(91)90537-I},
	volume = {161},
	year = {1991}}

@article{GiVeWeHaHoPhSt15,
	author = {Giustina, M. and Versteegh, M. A. M. and Wengerowsky, S. and Handsteiner, J. and Hochrainer, A. and Phelan, K. and Steinlechner, F. and Kofler, J. and Larsson, J.-A. and Abell{\'a}n, C. and Amaya, W. and Pruneri, V. and Mitchell, M. W. and Beyer, J. and Gerrits, T. and Lita, A. E. and Shalm, L. K. and Nam, S. W. and Scheidl, T. and Ursin, R. and Wittmann, B. and Zeilinger, A.},
	journal = {Phys. Rev. Lett.},
	pages = {250401},
	title = {Significant-{L}oophole-{F}ree {T}est of {B}ell's {T}heorem with {E}ntangled {P}hotons},
	url = {https://doi.org/10.1103/PhysRevLett.115.250401},
	volume = {115},
	year = {2015}}

@article{ShMeChBiWaStGe15,
	author = {Shalm, L. K. and Meyer-Scott, E. and Christensen, B. G. and Bierhorst, P. and Wayne, M. A. and Stevens, M. J. and Gerrits, T. and Glancy, S. and Hamel, D. R. and Allman, M. S. and Coakley, K. J. and Dyer, S. D. and Hodge, C. and Lita, A. E. and Verma, V. B. and Lambrocco, C. and Tortorici, E. and Migdall, A. L. and Zhang, Y. and Kumor, D. R. and Farr, W. H. and Marsili, F. and Shaw, M. D. and Stern, J. A. and Abell{\'a}n, C. and Amaya, W. and Pruneri, V. and Jennewein, T. and Mitchell, M. W. and Kwiat, P. G. and Bienfang, J. C. and Mirin, R. P. and Knill, E. and Nam, S. W.},
	journal = {Phys. Rev. Lett.},
	pages = {250402},
	title = {Strong {L}oophole-{F}ree {T}est of {L}ocal {R}ealism},
	url = {https://doi.org/10.1103/PhysRevLett.115.250402},
	volume = {115},
	year = {2015}}

@article{RaHaHoGaFrLeLi18,
	author = {Rauch, D. and Handsteiner, J. and Hochrainer, A. and Gallicchio, J. and Friedman, A. S. and Leung, C. and Liu, B. and Bulla, L. and Ecker, S. and Steinlechner, F. and Ursin, R. and Hu, B. and Leon, D. and Benn, C. and Ghedina, A. and Cecconi, M. and Guth, A. H. and Kaiser, D. I. and Scheidl, T. and Zeilinger, A.},
	journal = {Phys. Rev. Lett.},
	pages = {080403},
	title = {Cosmic {B}ell {T}est {U}sing {R}andom {M}easurement {S}ettings from {H}igh-{R}edshift {Q}uasars},
	url = {https://doi.org/10.1103/PhysRevLett.121.080403},
	volume = {121},
	year = {2018}}

@article{BIGBellCollaboration18,
	author = {{The BIG Bell test Collaboration}},
	journal = {Nature},
	pages = {212},
	title = {Challenging local realism with human choices},
	url = {https://doi.org/10.1038/s41586-018-0085-3},
	volume = {557},
	year = {2018}}

@article{Gi14c,
	author = {Gill, R. D.},
	journal = {Statist. Sci.},
	pages = {512},
	title = {Statistics, {C}ausality and {B}ell's {T}heorem},
	url = {https://doi.org/10.1214/14-STS490},
	volume = {29},
	year = {2014}}

@incollection{CoRe16,
	author = {Colbeck, R. and Renner, R.},
	booktitle = {Quantum Theory: Informational Foundations and Foils},
	editor = {Chiribella, G. and Spekkens, R. W.},
	pages = {497--528},
	publisher = {Springer},
	title = {The {C}ompleteness of {Q}uantum {T}heory for {P}redicting {M}easurement {O}utcomes},
	url = {https://doi.org/10.1007/978-94-017-7303-4_15},
	year = {2016}}

@article{CoRe08,
	author = {Colbeck, R. and Renner, R.},
	journal = {Phys. Rev. Lett.},
	pages = {050403},
	title = {Hidden {V}ariable {M}odels for {Q}uantum {T}heory {C}annot {H}ave {A}ny {L}ocal {P}art},
	url = {https://doi.org/10.1103/PhysRevLett.101.050403},
	volume = {101},
	year = {2008}}

@article{BaKePi06,
	author = {Barrett, J. and Kent, A. and Pironio, S.},
	journal = {Phys. Rev. Lett.},
	pages = {170409},
	title = {Maximally {N}onlocal and {M}onogamous {Q}uantum {C}orrelations},
	url = {https://doi.org/10.1103/PhysRevLett.97.170409},
	volume = {97},
	year = {2006}}

@article{ReRoMaGi17,
	author = {Renou, M. O. and Rosset, D. and Martin, A. and Gisin, N.},
	journal = {J. Phys. A: Math. Theor.},
	pages = {255301},
	title = {On the inequivalence of the {CH} and {CHSH} inequalities due to finite statistics},
	url = {http://dx.doi.org/10.1088/1751-8121/aa6f78},
	volume = {50},
	year = {2017}}

@article{PoBrGi12,
	author = {Portmann, S. and Branciard, C. and Gisin, N.},
	journal = {Phys. Rev. A},
	pages = {012104},
	title = {Local content of all pure two-qubit states},
	url = {https://doi.org/10.1103/PhysRevA.86.012104},
	volume = {86},
	year = {2012}}

@article{ElPoRo92,
	author = {Elitzur, A. C. and Popescu, S. and Rohrlich, D.},
	journal = {Phys. Lett. A},
	pages = {25},
	title = {Quantum nonlocality for each pair in an ensemble},
	url = {https://doi.org/10.1016/0375-9601(92)90952-I},
	volume = {162},
	year = {1992}}

@book{Sc19,
	author = {Scarani, V.},
	publisher = {Oxford University Press},
	title = {Bell Nonlocality},
	year = {2019}}

@article{Ts80,
	author = {Tsirelson, B. S.},
	journal = {Lett. Math. Phys.},
	pages = {93--100},
	title = {Quantum generalizations of {B}ell's inequality},
	url = {https://doi.org/10.1007/BF00417500},
	volume = {4},
	year = {1980}}

@article{CeDz18,
	author = {Cervantes, V. H. and Dzhafarov, E. N.},
	journal = {Decision},
	pages = {193},
	title = {Snow queen is evil and beautiful: {E}xperimental evidence for probabilistic contextuality in human choices.},
	url = {http://dx.doi.org/10.1037/dec0000095},
	volume = {5},
	year = {2018}}

@book{PeMa18,
	author = {Pearl, J. and Mackenzie, D.},
	publisher = {Basic Books},
	title = {The Book of Why: The New Science of Cause and Effect},
	year = {2018}}

@article{KuDzLa15,
	author = {Kujala, J. V. and Dzhafarov, E. N. and Larsson, J.-A.},
	journal = {Phys. Rev. Lett.},
	pages = {150401},
	title = {Necessary and {S}ufficient {C}onditions for an {E}xtended {N}oncontextuality in a {B}road {C}lass of {Q}uantum {M}echanical {S}ystems},
	url = {https://doi.org/10.1103/PhysRevLett.115.150401},
	volume = {115},
	year = {2015}}

@article{DzKuLa15,
	author = {Dzhafarov, E. N. and Kujala, J. V. and Larsson, J.-A.},
	journal = {Found. Phys.},
	number = {7},
	pages = {762--782},
	title = {Contextuality in {T}hree {T}ypes of {Q}uantum-{M}echanical {S}ystems},
	url = {http://dx.doi.org/10.1007/s10701-015-9882-9},
	volume = {45},
	year = {2015}}

@article{RiGiChCoWhFe16,
	author = {Ringbauer, M. and Giarmatzi, C. and Chaves, R. and Costa, F. and White, A. G. and Fedrizzi, A.},
	journal = {Sci. Adv.},
	pages = {e1600162},
	title = {Experimental test of nonlocal causality},
	url = {https://doi.org/10.1126/sciadv.1600162},
	volume = {2},
	year = {2016}}

@article{HeBeDrReKaBlRu15,
	author = {Hensen, B. and Bernien, H. and Dreau, A. E. and Reiserer, A. and Kalb, N. and Blok, M. S. and Ruitenberg, J. and Vermeulen, R. F. L. and Schouten, R. N. and Abellan, C. and Amaya, W. and Pruneri, V. and Mitchell, M. W. and Markham, M. and Twitchen, D. J. and Elkouss, D. and Wehner, S. and Taminiau, T. H. and Hanson, R.},
	journal = {Nature},
	pages = {682},
	title = {Loophole-free {B}ell inequality violation using electron spins separated by 1.3 kilometres},
	url = {http://dx.doi.org/10.1038/nature15759},
	volume = {526},
	year = {2015}}

@article{KoBr13,
	author = {Kofler, J. and Brukner, C.},
	journal = {Phys. Rev. A},
	pages = {052115},
	title = {Condition for macroscopic realism beyond the {L}eggett-{G}arg inequalities},
	url = {http://dx.doi.org/10.1103/PhysRevA.87.052115},
	volume = {87},
	year = {2013}}

@article{BrCaPiScWe14,
	author = {Brunner, N. and Cavalcanti, D. and Pironio, S. and Scarani, V. and Wehner, S.},
	journal = {Rev. Mod. Phys.},
	pages = {419},
	title = {Bell nonlocality},
	url = {http://dx.doi.org/10.1103/RevModPhys.86.419},
	volume = {86},
	year = {2014}}

@article{Fr10,
	author = {Tobias Fritz},
	isbn = {1367--2630},
	journal = {New J. Phys.},
	pages = {083055},
	title = {Quantum correlations in the temporal {C}lauser--{H}orne--{S}himony--{H}olt ({CHSH}) scenario},
	volume = {12},
	year = {2010}}

@article{EmLaNo14,
	author = {Emary, C. and Lambert, N. and Nori, F.},
	journal = {Rep. Prog. Phys.},
	pages = {016001},
	title = {Leggett--{G}arg inequalities},
	url = {http://dx.doi.org/10.1088/0034-4885/77/1/016001},
	volume = {77},
	year = {2014}}

@article{Be64,
	author = {Bell, J. S.},
	journal = {Physics},
	pages = {195},
	title = {On the {E}instein-{P}odolsky-{R}osen paradox},
	volume = {1},
	year = {1964}}

@article{ClHoShHo69,
	author = {Clauser, J. F. and Horne, M. A. and Shimony, A. and Holt, R. A.},
	journal = {Phys. Rev. Lett.},
	pages = {880--884},
	title = {Proposed {E}xperiment to {T}est {L}ocal {H}idden-{V}ariable {T}heories},
	url = {https://doi.org/10.1103/PhysRevLett.23.880},
	volume = {23},
	year = {1969}}

@article{Fi82a,
	author = {Fine, A.},
	journal = {Phys. Rev. Lett.},
	pages = {291--295},
	title = {Hidden {V}ariables, {J}oint {P}robability, and the {B}ell {I}nequalities},
	url = {https://doi.org/10.1103/PhysRevLett.48.291},
	volume = {48},
	year = {1982}}

@article{AsDaRo82,
	author = {Aspect, A. and Dalibard, J. and Roger, G.},
	journal = {Phys. Rev. Lett.},
	pages = {1804},
	title = {Experimental {T}est of {B}ell's {I}nequalities {U}sing {T}ime-{V}arying {A}nalyzers},
	url = {https://doi.org/10.1103/PhysRevLett.49.1804},
	volume = {49},
	year = {1982}}

@book{Pe09,
	author = {Pearl, J.},
	booktitle = {Causality: Models, Reasoning, and Inference},
	edition = {2nd},
	publisher = {Cambridge University Press},
	title = {Causality: Models, Reasoning, and Inference},
	year = {2009}}

@book{SpGlSc00,
	author = {Spirtes, P. and Glymour, C. and Scheines, R.},
	publisher = {The MIT Press},
	title = {Causation, Prediction, and Search},
	year = {2000}}

@article{PoRo94,
	author = {Popescu, S. and Rohrlich, D.},
	journal = {Found. Phys.},
	pages = {379},
	title = {Quantum {N}onlocality as an {A}xiom},
	url = {http://dx.doi.org/10.1007/BF02058098},
	volume = {24},
	year = {1994}}

@book{Be93,
	author = {Bell, J. S.},
	booktitle = {Speakable and unspeakable in quantum mechanics},
	publisher = {Cambridge University Press},
	title = {Speakable and unspeakable in quantum mechanics},
	year = {1987}}

\newpage\ \newpage


\setcounter{page}{1} 

\onecolumngrid
\begin{center}
{\large\text{\hypertarget{Supplementary-Information}{\myuline{\textit{Supplementary Information}}}}\vspace{0.4cm}\\\Large\textbf{
Measuring Bell non-locality in the presence of signaling }}\vspace{0.4cm}\\
Mark Broom$^{\text{1}}$, Talel Naccache$^{\text{1}}$,
Emmanuel M. Pothos$^{\text{1}}$,
Christoph Gallus$^{\text{2}}$, and
Pawel Blasiak$^{\text{3}}$\vspace{0.2cm}\\
\small\emph{$^{\text{1}}$City St George's, University of London, London EC1V 0HB, United Kingdom\\
$^{\text{2}}$Technische Hochschule Mittelhessen, D-35390 Gießen, Germany\\
$^{\text{3}}$Institute of Nuclear Physics, Polish Academy of Sciences, 31342 Krak\'ow, Poland}

\end{center}

\vspace{0.2cm}

\renewcommand{\theequation}{S\arabic{equation}}  
\renewcommand{\thefigure}{S\arabic{figure}}  

\setcounter{equation}{0}
\setcounter{figure}{0}

This Supplement provides detailed proofs of all statements presented in the main manuscript, as well as additional analytical insights into the structure of the solution. In \textbf{Part A}, we provide pointers to the full listings of the solution sets $\mathbb{Q}$ and $\mathbb{S}$ (see \textbf{Supplementary Data}). In \textbf{Part B}, we expand on the necessary technical details and give full proofs of \textbf{Theorems~\ref{Theorem-muL-e}}, \textbf{\ref{theoremL}}, \textbf{\ref{theoremNS}}, and \textbf{Lemma~\ref{Lemma-muNS-f}}. While the solution in \textbf{Part B} is obtained via a linear-programming approach, \textbf{Part C} provides an independent analytical treatment of key elements of the proof, based on explicit parametrization. Finally, \textbf{Part D} offers some additional insight into the  structure of the solution.

\vspace{0.6cm}

\textbf{\textsf{\underline{Table of contents}:}}
\begin{itemize}
\item[]{\textsf{\textbf{\underline{Part A}:} Solution sets $\mathbb{Q}$ and $\mathbb{S}$}}
\item[]{\textsf{\textbf{\underline{Part B}:} Linear programming approach}\vspace{-0.1cm}
\begin{itemize}
\item[]{B.1. \textsf{Technical background}}
\item[]{B.2. \textsf{Proof of \textbf{Theorem~\ref{theoremL}}}}
\item[]{B.3. \textsf{Proof of \textbf{Theorem~\ref{theoremNS}}}}
\item[]{B.4. \textsf{Proof of \textbf{Theorem~\ref{Theorem-muL-e}} and \textbf{Lemma~\ref{Lemma-muNS-f}}}}
\item[]{B.5. \textsf{Prevalence of the solution vectors}}
\end{itemize}}
\item[]{\textsf{\textbf{\underline{Part C}:} Analytical approach}\vspace{-0.1cm}
\begin{itemize}
\item[]{C.1. \textsf{The vector spaces}}
\item[]{C.2. \textsf{Vector space parametrizations}}
\item[]{C.3. \textsf{Maximum signaling}}
\item[]{C.4. \textsf{Vector measures}}
\item[]{C.5. \textsf{Finding the measure of non-signaling}}
\item[]{C.6. \textsf{Finding the measure of locality}}
\end{itemize}}
\item[]{\textsf{\textbf{\underline{Part D}:} Insights into solutions}}
\end{itemize}

\newpage

\noindent\textsf{\Large \textbf{\underline{Part A}}: Solution sets $\mathbb{Q}$ and $\mathbb{S}$}\vspace{0.5cm}

The main results of the paper, presented in \textbf{Theorem~\ref{theoremL}} and \textbf{Theorem~\ref{theoremNS}}, are based on the solution sets $\mathbb{Q}$ and $\mathbb{S}$, which contain 128 and 120 vectors, respectively. For clarity, we defined these sets in Eqs.~(\ref{Q}) and (\ref{S}) and described their structure without listing their elements explicitly, i.e.,
\begin{eqnarray}\label{Q-Suppl}
\mathbb{Q}&=&\{\vec{f}_1,...\,,\vec{f}_8\}\,\cup
\,\{\vec{g}_1,...\,,\vec{g}_{16}\}\,\cup\,\{\vec{t}_1,...\,,\vec{t}_{32}\}\,\cup\,\{\vec{s}_1,...\,,\vec{s}_{64}\}\,\cup
\,\{\vec{e}_1,...\,,\vec{e}_8\}\qquad\text{(see \textbf{Theorem~\ref{theoremL}})}\,,
\end{eqnarray}
and
\begin{eqnarray}\label{S-Suppl}
\mathbb{S}&=&\{\vec{f}_1,...\,,\vec{f}_8\}\,\cup
\,\{\vec{g}_1,...\,,\vec{g}_{16}\}\,\cup\,\{\vec{t}_1,...\,,\vec{t}_{32}\}\,\cup\,\{\vec{z}_1,...\,,\vec{z}_{64}\}\qquad\text{(see \textbf{Theorem~\ref{theoremNS}})}\,.
\end{eqnarray}
For completeness, we state them explicitly in two \textbf{Supplementary Data} files:
\begin{itemize}
\item[\textit{(i)}]\href{https://drive.google.com/file/d/1mDnEi6WXuOuXqPvZyFADXTj_KelZwqEO/view?usp=sharing}{\texttt{SolutionVectorsQ.txt}}\qquad(solution set $\mathbb{Q}$ in \textbf{Theorem~\ref{theoremL}})\,, and

\item[\textit{(ii)}]\href{https://drive.google.com/file/d/1cNTpB33FgRGV0ctaG4pJal9XaePZrNWj/view?usp=sharing}{\texttt{SolutionVectorsS.txt}}\qquad(solution set $\mathbb{S}$ in \textbf{Theorem~\ref{theoremNS}})\,.
\end{itemize}
These files are ready for direct use in numerical computations.

\vspace{0.3cm}

In the main text, we also discussed symmetries of both sets given in Eqs.~(\ref{sym-1})-(\ref{sym-5}). Those symmetries further partition the respective subsets of $\vec{f}$\,’s, $\vec{g}$\,’s, $\vec{t}$\,’s, $\vec{s}/\vec{z}$\,’s and $\vec{e}$\,’s into equivalence classes (orbits), with representatives given in Eq.~(\ref{representatives}). We also provide this classification in the accompanying \textbf{Supplementary Data} files.


\newpage

\noindent\textsf{\Large \textbf{\underline{Part B}}: Linear programming approach}\vspace{0.5cm}

\vspace{0cm}
\noindent\textbf{\textsf{\text{B.1. {Technical background}}}}\vspace{0.2cm}

We briefly review the technical notions and definitions required for the proofs presented in the manuscript.\vspace{0.3cm}

\noindent\textit{\textsf{\textit{B.1.1. {Local polytope $\mathcal{P}_{\scriptscriptstyle Loc}$}}}}\vspace{0.2cm}

The \textit{local polytope} $\mathcal{P}_{\scriptscriptstyle Loc}$ is the convex hull of sixteen \textit{local deterministic strategies}~\cite{Fi82a,BrCaPiScWe14,Sc19}. Each of these strategies can be identified with a pair of functions $(f,g)$, where $f,g:\{0,1\}\to\{0,1\}$, yielding a total of $\text{16 = 4}\cdot \text{4}$ distinct combinations. Using the notation of Eq.\,(\ref{Pab|xy}), these strategies are explicitly given by Eq.\,(\ref{LDS}), i.e.,
\begin{eqnarray}\label{LDS-Supplement}
(\vec{L}_{\scriptscriptstyle{fg}})_{\scriptscriptstyle{ab}|\scriptscriptstyle{xy}}&=&\delta_{a\text{\,=\,}f(x)}\!\cdot\delta_{b\text{\,=\,}g(y)}\,.
\end{eqnarray}
Identifying the sixteen function pairs $(f,g)$ with the labels $1,\dots,16$, we arrange the corresponding strategies as the columns of a matrix $L$, whose explicit form is given below for completeness:
\begin{eqnarray}\label{Local-matrix}
L&:=&(\vec{L}_{\scriptscriptstyle1}\,,\vec{L}_{\scriptscriptstyle2}\,,...\,,\vec{L}_{\scriptscriptstyle16})\ =\ 
{\color{gray}
\begin{array}{l}
\\
\scalebox{0.65}{
\text{${P}_{\scriptscriptstyle00|00}$}}\\[-2pt]
\scalebox{0.65}{
\text{${P}_{\scriptscriptstyle01|00}$}}\\[-2pt]
\scalebox{0.65}{
\text{${P}_{\scriptscriptstyle10|00}$}}\\[-2pt]
\scalebox{0.65}{
\text{${P}_{\scriptscriptstyle11|00}$}}\\[1pt]
\scalebox{0.65}{
\text{${P}_{\scriptscriptstyle00|01}$}}\\[-2pt]
\scalebox{0.65}{
\text{${P}_{\scriptscriptstyle01|01}$}}\\[-2pt]
\scalebox{0.65}{
\text{${P}_{\scriptscriptstyle10|01}$}}\\[-2pt]
\scalebox{0.65}{
\text{${P}_{\scriptscriptstyle11|01}$}}\\[1pt]
\scalebox{0.65}{
\text{${P}_{\scriptscriptstyle00|10}$}}\\[-2pt]
\scalebox{0.65}{
\text{${P}_{\scriptscriptstyle01|10}$}}\\[-2pt]
\scalebox{0.65}{
\text{${P}_{\scriptscriptstyle10|10}$}}\\[-2pt]
\scalebox{0.65}{
\text{${P}_{\scriptscriptstyle11|10}$}}\\[1pt]
\scalebox{0.65}{
\text{${P}_{\scriptscriptstyle00|11}$}}\\[-2pt]
\scalebox{0.65}{
\text{${P}_{\scriptscriptstyle01|11}$}}\\[-2pt]
\scalebox{0.65}{
\text{${P}_{\scriptscriptstyle10|11}$}}\\[-2pt]
\scalebox{0.65}{
\text{${P}_{\scriptscriptstyle11|11}$}}
\end{array}
}
\begin{array}{l}
{\color{gray}\ \ \ \ \scalebox{0.65}{\text{$\vec{L}_{\scriptstyle1}\ \ \,\vec{L}_{\scriptstyle2}$\ \ $\vec{L}_{\scriptstyle3}\ \ \vec{L}_{\scriptstyle4}$\ \ \,$\vec{L}_{\scriptstyle5}\,\ \ \vec{L}_{\scriptstyle6}$\ \ $\vec{L}_{\scriptstyle7}\ \ \vec{L}_{\scriptstyle8}$\ \ $\vec{L}_{\scriptstyle9}\ \ \vec{L}_{\scriptstyle10}$\ $\vec{L}_{\scriptstyle11}\ \vec{L}_{\scriptstyle12}$\ $\vec{L}_{\scriptstyle13}\ \vec{L}_{\scriptstyle14}$\ $\vec{L}_{\scriptstyle15}\ \vec{L}_{\scriptstyle16}$ }}
}\\[2.5pt]
\left(
\begin{array}{cccccccccccccccc}
 1 & 0 & 1 & 0 & 0 & 0 & 0 & 0 & 1 & 0 & 1 & 0 & 0 & 0 & 0 & 0 \\[-2pt]
 0 & 1 & 0 & 1 & 0 & 0 & 0 & 0 & 0 & 1 & 0 & 1 & 0 & 0 & 0 & 0 \\[-2pt]
 0 & 0 & 0 & 0 & 1 & 0 & 1 & 0 & 0 & 0 & 0 & 0 & 1 & 0 & 1 & 0 \\[-2pt]
 0 & 0 & 0 & 0 & 0 & 1 & 0 & 1 & 0 & 0 & 0 & 0 & 0 & 1 & 0 & 1 \\\hline
 1 & 0 & 0 & 1 & 0 & 0 & 0 & 0 & 1 & 0 & 0 & 1 & 0 & 0 & 0 & 0 \\[-2pt]
 0 & 1 & 1 & 0 & 0 & 0 & 0 & 0 & 0 & 1 & 1 & 0 & 0 & 0 & 0 & 0 \\[-2pt]
 0 & 0 & 0 & 0 & 1 & 0 & 0 & 1 & 0 & 0 & 0 & 0 & 1 & 0 & 0 & 1 \\[-2pt]
 0 & 0 & 0 & 0 & 0 & 1 & 1 & 0 & 0 & 0 & 0 & 0 & 0 & 1 & 1 & 0 \\\hline
 1 & 0 & 1 & 0 & 0 & 0 & 0 & 0 & 0 & 0 & 0 & 0 & 1 & 0 & 1 & 0 \\[-2pt]
 0 & 1 & 0 & 1 & 0 & 0 & 0 & 0 & 0 & 0 & 0 & 0 & 0 & 1 & 0 & 1 \\[-2pt]
 0 & 0 & 0 & 0 & 1 & 0 & 1 & 0 & 1 & 0 & 1 & 0 & 0 & 0 & 0 & 0 \\[-2pt]
 0 & 0 & 0 & 0 & 0 & 1 & 0 & 1 & 0 & 1 & 0 & 1 & 0 & 0 & 0 & 0 \\\hline
 1 & 0 & 0 & 1 & 0 & 0 & 0 & 0 & 0 & 0 & 0 & 0 & 1 & 0 & 0 & 1 \\[-2pt]
 0 & 1 & 1 & 0 & 0 & 0 & 0 & 0 & 0 & 0 & 0 & 0 & 0 & 1 & 1 & 0 \\[-2pt]
 0 & 0 & 0 & 0 & 1 & 0 & 0 & 1 & 1 & 0 & 0 & 1 & 0 & 0 & 0 & 0 \\[-2pt]
 0 & 0 & 0 & 0 & 0 & 1 & 1 & 0 & 0 & 1 & 1 & 0 & 0 & 0 & 0 & 0 
\end{array}
\right)
\end{array}.
\end{eqnarray}

Crucially, any local behaviour $\vec{P}^{\scriptscriptstyle \,L}$ admits a convex decomposition of the form given in Eq.\,(\ref{Loc-det-decomp}), i.e.,
\begin{eqnarray}\label{Loc-det-decomp-Supplement}
\vec{P}^{\scriptscriptstyle{\,L}}&=&\sum_{\scriptscriptstyle{i=1}}^{\scriptscriptstyle{16}}\, \vec{L}_i\,{q_i}\ =\ L\,\vec{q}\ ,
\end{eqnarray}
where $\vec{q}$ is a vector of non-negative weights satisfying the normalization condition $\sum_{\scriptscriptstyle{i=1}}^{\scriptscriptstyle{16}} q_i=1$.

One readily verifies that the set of \textit{local deterministic strategies} $\{\vec{L}_{\scriptscriptstyle1}\,,\vec{L}_{\scriptscriptstyle2}\,,...\,,\vec{L}_{\scriptscriptstyle16}\}$, which define the \textit{local polytope} $\mathcal{P}_{\scriptscriptstyle Loc}$, is invariant under the symmetry transformations given in Eqs.\,(\ref{sym-1})-(\ref{sym-5}).

\vspace{0.3cm}
\noindent\textit{\textsf{\textit{B.1.2. {Non-signaling polytope $\mathcal{P}_{\scriptscriptstyle NS}$}}}}\vspace{0.2cm}

The \textit{non-signaling polytope} $\mathcal{P}_{\scriptscriptstyle NS}$ is spanned by the sixteen \textit{local deterministic strategies} introduced above, together with the eight so-called \textit{PR boxes}~\cite{PoRo94,BrCaPiScWe14,Sc19}. These boxes correspond to non-deterministic, non-signaling behaviors that saturate the algebraic bound of the CHSH inequalities ($S = 4$). A canonical representative is
\begin{eqnarray}\label{PR-Supplement}
(\vec{R}_{\scriptscriptstyle{1}})_{\scriptscriptstyle{ab}|\scriptscriptstyle{xy}}&:=&\tfrac{1}{2}\,\delta_{a \oplus b\,\text{=}\, x\cdot y}\,,
\end{eqnarray}
where $\oplus$ denotes addition modulo~2. The remaining boxes are obtained by relabeling the measurement settings and outcomes.
Altogether, these $\text{24 = 16\,+\,8}$ column vectors can be collected into a matrix $N=(\vec{N}_{\scriptscriptstyle1}\,,...\,,\vec{N}_{\scriptscriptstyle24})$, whose explicit form is given below:

\begin{eqnarray}\label{NS-matrix}
\!\!\!\!\!\!\!\!\!N&:=&(\vec{L}_{\scriptscriptstyle1}\,,...\,,\vec{L}_{\scriptscriptstyle16}\,,\vec{R}_{\scriptscriptstyle1}\,,...\,,\vec{R}_{\scriptscriptstyle8})\ =\ 
{\color{gray}
\begin{array}{l}
\\
\scalebox{0.65}{
\text{${P}_{\scriptscriptstyle00|00}$}}\\[-2pt]
\scalebox{0.65}{
\text{${P}_{\scriptscriptstyle01|00}$}}\\[-2pt]
\scalebox{0.65}{
\text{${P}_{\scriptscriptstyle10|00}$}}\\[-2pt]
\scalebox{0.65}{
\text{${P}_{\scriptscriptstyle11|00}$}}\\[1pt]
\scalebox{0.65}{
\text{${P}_{\scriptscriptstyle00|01}$}}\\[-2pt]
\scalebox{0.65}{
\text{${P}_{\scriptscriptstyle01|01}$}}\\[-2pt]
\scalebox{0.65}{
\text{${P}_{\scriptscriptstyle10|01}$}}\\[-2pt]
\scalebox{0.65}{
\text{${P}_{\scriptscriptstyle11|01}$}}\\[1pt]
\scalebox{0.65}{
\text{${P}_{\scriptscriptstyle00|10}$}}\\[-2pt]
\scalebox{0.65}{
\text{${P}_{\scriptscriptstyle01|10}$}}\\[-2pt]
\scalebox{0.65}{
\text{${P}_{\scriptscriptstyle10|10}$}}\\[-2pt]
\scalebox{0.65}{
\text{${P}_{\scriptscriptstyle11|10}$}}\\[1pt]
\scalebox{0.65}{
\text{${P}_{\scriptscriptstyle00|11}$}}\\[-2pt]
\scalebox{0.65}{
\text{${P}_{\scriptscriptstyle01|11}$}}\\[-2pt]
\scalebox{0.65}{
\text{${P}_{\scriptscriptstyle10|11}$}}\\[-2pt]
\scalebox{0.65}{
\text{${P}_{\scriptscriptstyle11|11}$}}
\end{array}
}
\begin{array}{l}
{\color{gray}\ \ \ \ \scalebox{0.65}{\text{$\vec{L}_{\scriptstyle1}\ \ \,\vec{L}_{\scriptstyle2}$\ \ $\vec{L}_{\scriptstyle3}\ \ \vec{L}_{\scriptstyle4}$\ \ \,$\vec{L}_{\scriptstyle5}\,\ \ \vec{L}_{\scriptstyle6}$\ \ $\vec{L}_{\scriptstyle7}\ \ \vec{L}_{\scriptstyle8}$\ \ $\vec{L}_{\scriptstyle9}\ \ \vec{L}_{\scriptstyle10}$\ $\vec{L}_{\scriptstyle11}\ \vec{L}_{\scriptstyle12}$\ $\vec{L}_{\scriptstyle13}\ \vec{L}_{\scriptstyle14}$\ $\vec{L}_{\scriptstyle15}\ \vec{L}_{\scriptstyle16}$\ \,$\vec{R}_{\scriptstyle1}\ \ \ \vec{R}_{\scriptstyle2}$\ \ \ $\vec{R}_{\scriptstyle3}\ \ \ \,\vec{R}_{\scriptstyle4}$\ \ \ \,$\vec{R}_{\scriptstyle5}\ \ \ \vec{R}_{\scriptstyle6}$\ \ \ \,$\vec{R}_{\scriptstyle7}\ \ \ \vec{R}_{\scriptstyle8}$}}
}\\[2.5pt]
\left(
\begin{array}{cccccccccccccccc|cccccccc}
 1 & 0 & 1 & 0 & 0 & 0 & 0 & 0 & 1 & 0 & 1 & 0 & 0 & 0 & 0 & 0 &\text{\textonehalf}&\text{\textonehalf}&\text{\textonehalf}&0&0&0&0&\text{\textonehalf}\\[-2pt]
 0 & 1 & 0 & 1 & 0 & 0 & 0 & 0 & 0 & 1 & 0 & 1 & 0 & 0 & 0 & 0 &0&0&0&\text{\textonehalf}&\text{\textonehalf}&\text{\textonehalf}&\text{\textonehalf}&0\\[-2pt]
 0 & 0 & 0 & 0 & 1 & 0 & 1 & 0 & 0 & 0 & 0 & 0 & 1 & 0 & 1 & 0&0&0&0&\text{\textonehalf}&\text{\textonehalf}&\text{\textonehalf}&\text{\textonehalf}&0\\[-2pt]
 0 & 0 & 0 & 0 & 0 & 1 & 0 & 1 & 0 & 0 & 0 & 0 & 0 & 1 & 0 & 1&\text{\textonehalf}&\text{\textonehalf}&\text{\textonehalf}&0&0&0&0&\text{\textonehalf}\\\hline
 1 & 0 & 0 & 1 & 0 & 0 & 0 & 0 & 1 & 0 & 0 & 1 & 0 & 0 & 0 & 0&\text{\textonehalf}&\text{\textonehalf}&0&\text{\textonehalf}&0&0&\text{\textonehalf}&0\\[-2pt]
 0 & 1 & 1 & 0 & 0 & 0 & 0 & 0 & 0 & 1 & 1 & 0 & 0 & 0 & 0 & 0&0&0&\text{\textonehalf}&0&\text{\textonehalf}&\text{\textonehalf}&0&\text{\textonehalf}\\[-2pt]
 0 & 0 & 0 & 0 & 1 & 0 & 0 & 1 & 0 & 0 & 0 & 0 & 1 & 0 & 0 & 1&0&0&\text{\textonehalf}&0&\text{\textonehalf}&\text{\textonehalf}&0&\text{\textonehalf}\\[-2pt]
 0 & 0 & 0 & 0 & 0 & 1 & 1 & 0 & 0 & 0 & 0 & 0 & 0 & 1 & 1 & 0&\text{\textonehalf}&\text{\textonehalf}&0&\text{\textonehalf}&0&0&\text{\textonehalf}&0\\\hline
 1 & 0 & 1 & 0 & 0 & 0 & 0 & 0 & 0 & 0 & 0 & 0 & 1 & 0 & 1 & 0&\text{\textonehalf}&0&\text{\textonehalf}&\text{\textonehalf}&0&\text{\textonehalf}&0&0\\[-2pt]
 0 & 1 & 0 & 1 & 0 & 0 & 0 & 0 & 0 & 0 & 0 & 0 & 0 & 1 & 0 & 1&0&\text{\textonehalf}&0&0&\text{\textonehalf}&0&\text{\textonehalf}&\text{\textonehalf}\\[-2pt]
 0 & 0 & 0 & 0 & 1 & 0 & 1 & 0 & 1 & 0 & 1 & 0 & 0 & 0 & 0 & 0&0&\text{\textonehalf}&0&0&\text{\textonehalf}&0&\text{\textonehalf}&\text{\textonehalf}\\[-2pt]
 0 & 0 & 0 & 0 & 0 & 1 & 0 & 1 & 0 & 1 & 0 & 1 & 0 & 0 & 0 & 0&\text{\textonehalf}&0&\text{\textonehalf}&\text{\textonehalf}&0&\text{\textonehalf}&0&0\\\hline
 1 & 0 & 0 & 1 & 0 & 0 & 0 & 0 & 0 & 0 & 0 & 0 & 1 & 0 & 0 & 0 &0&\text{\textonehalf}&\text{\textonehalf}&\text{\textonehalf}&\text{\textonehalf}&0&0&0\\[-2pt]
 0 & 1 & 1 & 0 & 0 & 0 & 0 & 0 & 0 & 0 & 0 & 0 & 0 & 1 & 1 & 0&\text{\textonehalf}&0&0&0&0&\text{\textonehalf}&\text{\textonehalf}&\text{\textonehalf}\\[-2pt]
 0 & 0 & 0 & 0 & 1 & 0 & 0 & 1 & 1 & 0 & 0 & 1 & 0 & 0 & 0 & 0&\text{\textonehalf}&0&0&0&0&\text{\textonehalf}&\text{\textonehalf}&\text{\textonehalf}\\[-2pt]
 0 & 0 & 0 & 0 & 0 & 1 & 1 & 0 & 0 & 1 & 1 & 0 & 0 & 0 & 0 & 0&0&\text{\textonehalf}&\text{\textonehalf}&\text{\textonehalf}&\text{\textonehalf}&0&0&0
\end{array}
\right)
\end{array}.
\end{eqnarray}

This means that any non-signaling behavior $\vec{P}^{\scriptscriptstyle{\,NS}}$ can be expressed as a convex mixture of the extremal points $\vec{N}_i$ given above, i.e.,
\begin{eqnarray}\label{LS-decomp-Supplement}
\vec{P}^{\scriptscriptstyle{\,NS}}&=&\sum_{\scriptscriptstyle{i=1}}^{\scriptscriptstyle{24}}\, \vec{N}_i\,{q_i}\ =\ N\,\vec{q}\ ,
\end{eqnarray}
where $\vec{q}$ is a vector of non-negative weights satisfying the normalization condition $\sum_{\scriptscriptstyle{i=1}}^{\scriptscriptstyle{24}} q_i=1$.

Again, it can be checked that the set of vertices $(\vec{L}_{\scriptscriptstyle1}\,,...\,,\vec{L}_{\scriptscriptstyle16}\,,\vec{R}_{\scriptscriptstyle1}\,,...\,,\vec{R}_{\scriptscriptstyle8})$, defining the \textit{non-signaling polytope} $\mathcal{P}_{\scriptscriptstyle NS}$, remains invariant under the symmetry transformations given in Eqs.\,(\ref{sym-1})-(\ref{sym-5}).





\vspace{0.3cm}

\noindent\textbf{\textsf{\text{B.2. {\textbf{Proof of Theorem~\ref{theoremL}}}}}}\vspace{0.2cm}

As shown in the \textbf{Methods} section, evaluating the \textit{local fraction} $\mu_{\scriptscriptstyle L}$ reduces to finding the vertices of the polyhedron defined by Eqs.\,(\ref{DLP-poly-1})-(\ref{DLP-poly-2}). The \textit{local fraction} $\mu_{\scriptscriptstyle L}(\vec{P})$ is then obtained as the minimum of the scalar products between $\vec{P}$ and the vertex directions; see Eq.\,(\ref{mu-solution}).

The task of vertex enumeration is effectively performed using computer algorithms. In this work, we used \href{https://porta.zib.de}{\textsc{PORTA}} software system (POlyhedron Representation Transformation Algorithm)~\cite{Ch97}, which implements Fourier-Motzkin elimination to convert between the $V$- and $H$-representations of polyhedra.

It follows that the polyhedron defined by Eqs.\,(\ref{DLP-poly-1})-(\ref{DLP-poly-2}) is the convex hull of 132 vertices together with a cone. The latter corresponds to the non-negative orthant, $\vec{q} \geqslant 0$, which is irrelevant for the solution. Among the 132 vertices, four are uninteresting for our purposes, namely $(1111000000000000)^T$, $(0000111100000000)^T$, $(0000000011110000)^T$, and $(0000000000001111)^T$, since for these vectors the expression in Eq.\,(\ref{mu-solution}) evaluates to 1. This follows directly from the normalisation condition in Eq.\,(\ref{Pab|xy}), which requires that $\sum_{\scriptscriptstyle a,b} P_{\scriptscriptstyle ab|xy} = 1$ for all $x, y \in {0,1}$. We are thus left with 128 non-trivial vertices contributing to the minimum in Eq.\,(\ref{mu-solution}); these define the solution set $\mathbb{Q}$ stated in \textbf{Theorem~\ref{theoremL}}, thereby completing its proof. The explicit list of all relevant vertices is provided in~\textbf{Part A} (\texttt{SolutionVectorsA.txt}).

\vspace{0.3cm}

\noindent\textbf{\textsf{B.3. {\textbf{Proof of Theorem~\ref{theoremNS}}}}}\vspace{0.2cm}

The idea of the proof follows the same route as that of \textbf{Theorem~\ref{theoremL}}, presented in the \textbf{Methods} section of the paper and completed in \textbf{Section~B.2} above. This highlights the generality of our method of computing \textit{fractional measures}.

Accordingly, in the first step, we show that the \textit{non-signaling fraction} $\mu_{\scriptscriptstyle{NS}}(\vec{P})$, as defined in Eq.\,(\ref{muNS}), is equivalent to solving the following linear program (LP):
\begin{eqnarray}\label{LP-NS-1}
\text{maximize}&&{\vec{1}}^{\scriptscriptstyle{T}}\!\cdot\vec{p}\,,\\\label{LP-NS-2}
\text{subject to}&&{N}\,\vec{p}\,\leqslant\,\vec{P}\,,\\\label{LP-NS-3}
&&\ \ \ \ \vec{p}\geqslant\,0\,.
\end{eqnarray}
Here $\vec{p}\in\mathbb{R}^{24}$, and $N \in\mathbb{R}^{16\times24}$ is the matrix consisting of the 24 vertices of the \textit{non-signaling polytope} defined in Eq.\,(\ref{NS-matrix}). The vector $\vec{1}:=(1,1,\ldots,1)^{\scriptscriptstyle{T}}\in\mathbb{R}^{24}$ denotes a column vector with all entries equal to one.

\begin{proof}[\myuline{Justification}]

The key observation is that, using Eq.\,(\ref{LS-decomp-Supplement}), the decomposition in Eq.\,(\ref{decompNS}) can be written as follows:
\begin{eqnarray}\label{condition-NS-1}
\vec{P}&=&q\cdot \underbrace{\sum_{\scriptscriptstyle i=1}^{\scriptscriptstyle24}\ \vec{N}_i\,q_i}_{\vec{P}^{\scriptscriptstyle{\,NS}}}\ +\ (1-q)\cdot\vec{P}^{\scriptscriptstyle{\,S}}\ .
\end{eqnarray}
Note, that we assume $0 \leqslant q < 1$, because the case  $q=1$ corresponds to the trivial situation where $\vec{P}$ is a non-signaling behavior. Consequently, there exist coefficients $p_i$, defined as $p_i := q q_i$, such that
\begin{eqnarray}\label{condition-NS-2}
\vec{P}\ -\ \sum_{\scriptscriptstyle i=1}^{\scriptscriptstyle 24}\ \vec{N}_i\,p_i\ \geqslant\ 0&\ \ \text{and}\ \ &p_i\geqslant0\ .
\end{eqnarray}

Conversely, the existence of coefficients $p_i$ fulfilling Eq.\,(\ref{condition-NS-2}) leads to the decomposition in Eq.\,(\ref{condition-NS-1}), as can be seen by rewriting
\begin{eqnarray}
\vec{P}&=&\sum_{\scriptscriptstyle i=1}^{\scriptscriptstyle24}\ \vec{N}_i\,p_i\ +\ \vec{P}\ -\ \sum_{\scriptscriptstyle i=1}^{\scriptscriptstyle24}\ \vec{N}_i\,p_i\ =\ q\cdot\underbrace{\sum_{\scriptscriptstyle i=1}^{\scriptscriptstyle 24}\ \vec{N}_i\,q_i}_{\vec{P}^{\scriptscriptstyle{\,NS}}}\ +\ (1-q)\cdot\underbrace{\tfrac{1}{1-q}\,\Big(\vec{P}\ -\ \sum_{\scriptscriptstyle i=1}^{\scriptscriptstyle24}\ \vec{N}_i\,p_i\Big)}_{\vec{P}^{\scriptscriptstyle{\,S}}}\,,
\end{eqnarray}
with $q := \sum_i p_i$ and $q_i := \nicefrac{p_i}{q}$. It remains to observe that the vectors defined in this way, $\vec{P}^{\scriptscriptstyle{\,NS}}$ and $\vec{P}^{\scriptscriptstyle{\,N}}$, are properly normalized behaviors, since both $\vec{P}$ and the $\vec{N}_i$’s are normalized behaviors by construction.

Therefore, finding $\mu_{\scriptscriptstyle NS}(\vec{P})$ in Eq.\,(\ref{muNS}) is equivalent to solving the following linear program:
\begin{eqnarray}
\text{maximize}&&{\sum}_{\scriptscriptstyle i=1}^{\scriptscriptstyle 24}\ p_i\\
\text{subject to}&&\vec{P}-{\sum}_{\scriptscriptstyle i=1}^{\scriptscriptstyle 24}\ \vec{N}_i\ p_i\geqslant0\ \ \  \text{ (16 inequalities) },\\
&&\qquad\ p_1,...,p_{24} \ \geqslant0\ \ \  \text{ (24 inequalities) },
\end{eqnarray}
which, in compact form, is given by the LP in Eqs.\,(\ref{LP-NS-1})-(\ref{LP-NS-3}).

\end{proof}

In the second step, we invoke the \textit{Strong Duality Theorem}, which guarantees the equivalence between the primal LP and its \textit{dual linear program} (DLP), given by
\begin{eqnarray}\label{DLP-NS-1}
\text{minimize}&&{\vec{P}}^{\scriptscriptstyle{\,T}}\!\cdot\vec{q}\,\\\label{DLP-NS-2}
\text{subject to}&&{N}^{\scriptscriptstyle{T}}\,\vec{q}\,\geqslant\,\vec{{1}}\ \ \  \text{ (24 inequalities) },\\\label{DLP-NS-3}
&&\quad\ \,\vec{q}\,\geqslant\,0 \ \ \   \text{ (16 inequalities) }\,,
\end{eqnarray}
with $\vec{q}\in\mathbb{R}^{16}$. The validity of this equivalence is justified by the existence of the solution to the primal LP in Eqs.\,(\ref{LP-NS-1})-(\ref{LP-NS-3}), since it maximizes a linear function over a compact set. Again, we observe that, if the problem is formulated in this way:

\noindent\textit{(i)} the chosen behavior $\vec{P}$ now appears \textit{only} in the objective function, Eq.\,(\ref{DLP-NS-1}), and

\noindent\textit{(ii)} the polyhedron in $\mathbb{R}^{16}$ defined by Eqs.\,(\ref{DLP-NS-2})-(\ref{DLP-NS-3}) is \textit{the same} for all behaviors $\vec{P}$.

To conclude, since the solution to he DLP in Eqs.\,(\ref{DLP-NS-1})-(\ref{DLP-NS-3}) is attained at the vertices of its feasible region, we have shown that the \textit{non-signaling fraction} $\mu_{\scriptscriptstyle NS}(\vec{P})$ can be computed as
\begin{eqnarray}\label{mu-NS-solution}
\mu_{\scriptscriptstyle NS}(\vec{P})&=&\min_{\vec{q}_i}\ \vec{q}_i^{\scriptscriptstyle{\,T}}\!\cdot\vec{P}\ ,
\end{eqnarray}
where $\{\vec{q}_i\}$ is the set of vertices of a polyhedron defined by $\text{40 = 24\,+\,16}$ inequalities:
\begin{eqnarray}\label{DLP-NS-poly-1}
&&{N}^{\scriptscriptstyle{T}}\,\vec{q}\,\geqslant\,\vec{{1}}\,,\\\label{DLP-NS-poly-2}
&&\quad\ \vec{q}\,\geqslant\,0\,.
\end{eqnarray}

What remains is to enumerate the relevant set of vertices, which we denote in this case by $\mathbb{S}\equiv\{\vec{q}_i\}$. Again, the problem can be efficiently solved using the \href{https://porta.zib.de}{\textsc{PORTA}} software system (POlyhedron Representation Transformation Algorithm)~\cite{Ch97}.
As a result, we find that the polyhedron defined by Eqs.\,(\ref{DLP-NS-poly-1})-(\ref{DLP-NS-poly-2}) is a convex hull of 124 vertices together with a cone. A closer inspection shows that the cone corresponds to the non-negative orthant, $\vec{q}\geqslant0$, and is therefore irrelevant for the solution. Among the 124 vertices, we may again neglect four uninteresting ones, namely $(1111000000000000)^{\scriptscriptstyle T}$, $(0000111100000000)^{\scriptscriptstyle T}$, $(0000000011110000)^{\scriptscriptstyle T}$, and $(0000000000001111)^{\scriptscriptstyle T}$, since for these vectors the expression in Eq.\,(\ref{mu-NS-solution}) evaluates to 1. This follows from the normalization of the behavior $\vec{P}$ in Eq.\,(\ref{Pab|xy}), which satisfies $\sum_{\scriptscriptstyle a,b} P_{\scriptscriptstyle ab|xy} = 1$ for all $x, y \in {0,1}$.
We are therefore left with 120 non-trivial vertices contributing to the minimum in Eq.\,(\ref{muNS-solution}), defining the solution set $\mathbb{S}$ stated in \textbf{Theorem~\ref{theoremNS}}, thereby completing its proof. For the explicit list of all relevant vertices, see \textbf{Part~A} (\texttt{SolutionVectorsS.txt}).
\vspace{0.3cm}

\noindent\textbf{\textsf{B.4. {\textbf{Proof of Theorem~\ref{Theorem-muL-e} and \textbf{Lemma~\ref{Lemma-muNS-f}}}}}}

\begin{proof}[\myuline{Proof of \textbf{Lemma~\ref{Lemma-muNS-f}}}]\ \\
It is straightforward to verify, using the explicit list of $f$-vectors in \textbf{Part A} (or \textbf{Part D}) and the definitions in Eqs.~(\ref{NS-AB0})-(\ref{NS-BA1}), that for a given behaviour $\vec{P}$ we have
\begin{eqnarray}\label{L1-f12}
\vec{f}^{\scriptscriptstyle{\,T}}_{\scriptscriptstyle  1}\!\cdot\vec{P} &=&1+\Delta_{\scriptscriptstyle  1}\qquad\&\qquad\vec{f}^{\scriptscriptstyle{\,T}}_{\scriptscriptstyle  2}\!\cdot\vec{P} \ =\ 1-\Delta_{\scriptscriptstyle  1}\ ,\\\label{L1-f34}
\vec{f}^{\scriptscriptstyle{\,T}}_{\scriptscriptstyle  3}\!\cdot\vec{P} &=&1+\Delta_{\scriptscriptstyle  2}\qquad\&\qquad\vec{f}^{\scriptscriptstyle{\,T}}_{\scriptscriptstyle  4}\!\cdot\vec{P} \ =\ 1-\Delta_{\scriptscriptstyle  2}\ ,\\\label{L1-f56}
\vec{f}^{\scriptscriptstyle{\,T}}_{\scriptscriptstyle  5}\!\cdot\vec{P} &=&1+\Delta_{\scriptscriptstyle  3}\qquad\&\qquad\vec{f}^{\scriptscriptstyle{\,T}}_{\scriptscriptstyle  6}\!\cdot\vec{P} \ =\ 1-\Delta_{\scriptscriptstyle  3}\ ,\\\label{L1-f78}
\vec{f}^{\scriptscriptstyle{\,T}}_{\scriptscriptstyle  7}\!\cdot\vec{P} &=&1+\Delta_{\scriptscriptstyle  4}\qquad\&\qquad\vec{f}^{\scriptscriptstyle{\,T}}_{\scriptscriptstyle  8}\!\cdot\vec{P} \ =\ 1-\Delta_{\scriptscriptstyle  4}\ .
\end{eqnarray}
This allows us to write
\begin{eqnarray}
\min_{\scriptscriptstyle{i\text{\,=\,}1,...,8}}\ \vec{f}_i^{\scriptscriptstyle{\,T}}\!\cdot\vec{P}&=&1+\min_{\scriptscriptstyle{i\text{\,=\,}1,...,4}}\ \pm\Delta_i\ =\ 1-\max_{\scriptscriptstyle{i\text{\,=\,}1,...,4}}\ |\Delta_i|\ \stackrel{(\ref{NS})}{=}\ 1-\Delta\ .
\end{eqnarray}
which, by definition in Eq.\,(\ref{NS}), proves \textbf{Lemma~\ref{Lemma-muNS-f}}.

\end{proof}

\begin{proof}[\myuline{Proof of \textbf{Theorem~\ref{Theorem-muL-e}}}]\ \\
For any behavior $\vec{P}$, the explicit list of $e$-vectors in \textbf{Part A} (or \textbf{Part D}), together with Eqs.~(\ref{S1})-(\ref{S4}), implies the following identities:
\begin{eqnarray}\label{E12-S4}
\vec{e}^{\scriptscriptstyle{\,T}}_{\scriptscriptstyle  1}\!\cdot\vec{P} &=&2-\tfrac{1}{2}S_{\scriptscriptstyle  4}\qquad\&\qquad\vec{e}^{\scriptscriptstyle{\,T}}_{\scriptscriptstyle  2}\!\cdot\vec{P} \ =\ 2+\tfrac{1}{2}S_{\scriptscriptstyle  4}\ ,\\\label{E34-S3}
\vec{e}^{\scriptscriptstyle{\,T}}_{\scriptscriptstyle  3}\!\cdot\vec{P} &=&2-\tfrac{1}{2}S_{\scriptscriptstyle  3}\qquad\&\qquad\vec{e}^{\scriptscriptstyle{\,T}}_{\scriptscriptstyle  4}\!\cdot\vec{P} \ =\ 2+\tfrac{1}{2}S_{\scriptscriptstyle  3}\ ,\\\label{E56-S2}
\vec{e}^{\scriptscriptstyle{\,T}}_{\scriptscriptstyle  5}\!\cdot\vec{P} &=&2-\tfrac{1}{2}S_{\scriptscriptstyle  2}\qquad\&\qquad\vec{e}^{\scriptscriptstyle{\,T}}_{\scriptscriptstyle  6}\!\cdot\vec{P} \ =\ 2+\tfrac{1}{2}S_{\scriptscriptstyle  2}\ ,\\\label{E78-S1}
\vec{e}^{\scriptscriptstyle{\,T}}_{\scriptscriptstyle  7}\!\cdot\vec{P} &=&2-\tfrac{1}{2}S_{\scriptscriptstyle  1}\qquad\&\qquad\vec{e}^{\scriptscriptstyle{\,T}}_{\scriptscriptstyle  8}\!\cdot\vec{P} \ =\ 2+\tfrac{1}{2}S_{\scriptscriptstyle  1}\ .
\end{eqnarray}
This, in turn allows us to write
\begin{eqnarray}\label{eP=S}
\min_{\scriptscriptstyle{i\text{\,=\,}1,...,8}}\ \vec{e}_i^{\scriptscriptstyle{\,T}}\!\cdot\vec{P}&=&2+\tfrac{1}{2}\min_{\scriptscriptstyle{i\text{\,=\,}1,...,4}}\ \pm S_i\ =\ 2-\tfrac{1}{2}\max_{\scriptscriptstyle{i\text{\,=\,}1,...,4}}\ |S_i|\ \stackrel{(\ref{CHSH})}{=}\ 2-\tfrac{1}{2}S\ . 
\end{eqnarray}
Now, under the assumption of a \textit{non-signaling} behavior $\vec{P}$, we have from the result derived in Ref.~\cite{BlPoYeGaBo21} that
\begin{eqnarray}\label{eq-2}
\mu_{\scriptscriptstyle L}(\vec{P})&\stackrel{\text{\cite{BlPoYeGaBo21}}}{=}&\min\big\{1\,,\tfrac{1}{2}(4-S)\big\}\,\stackrel{(\ref{eP=S})}{=}\,\min\big\{1\,,\min_{\scriptscriptstyle{i\text{\,=\,}1,...,8}}\ \vec{e}_i^{\scriptscriptstyle{\,T}}\!\cdot\vec{P}\,\big\}\,,
\end{eqnarray}
which completes the proof of \textbf{Theorem~\ref{Theorem-muL-e}}.

\end{proof}

As an aside, we note that the difficult part of the above derivation, Eq.\,(\ref{eq-2}), hinges on the result in Ref.~\cite{BlPoYeGaBo21}. However, we could have formulated this argument by referring instead to the stronger result, \textbf{Theorem \ref{theoremL}}, proved in this paper, which holds for arbitrary (including signaling) behaviors. We therefore now give a second, independent proof of \textbf{Theorem~\ref{Theorem-muL-e}}.

\begin{proof}[\myuline{Proof of \textbf{Theorem~\ref{Theorem-muL-e}} from \textbf{Theorem~\ref{theoremL}}}]\ \\
For completeness, let us justify the result in Ref.~\cite{BlPoYeGaBo21} using our \textbf{Theorem \ref{theoremL}}. In view of the property established in Eq.~(\ref{eP=S}), it suffices to show that for \textit{non-signaling} behaviors $\vec{P}$ ($\Delta=0$), the minimization over $\mathbb{Q}$ in Eq.~(\ref{muL-solution}) can be restricted to the $\vec{e}_i$’s. That is, we have
\begin{eqnarray}\label{minQ-e}
\min_{\scriptscriptstyle{\vec{q}\,\in\,\mathbb{Q}}}\,\vec{q}^{\scriptscriptstyle{\,T}}\!\cdot\vec{P}&=&\min_{\scriptscriptstyle{i\text{\,=\,}1,...,8}}\big\{1\,,\,\vec{e}_i^{\scriptscriptstyle{\,T}}\!\cdot\vec{P}\,\big\}\,, \quad\qquad\text{for \textit{non-signaling} behaviors $\vec{P}$ ($\Delta=0$)}\,.
\end{eqnarray}

We start by observing that, in the \textit{non-signaling} case, \textbf{Lemma~\ref{Lemma-muNS-f}} implies that the \textit{f-measures} are equal to 1.
Therefore, to establish Eq.~(\ref{minQ-e}), it suffices to show that for \textit{non-signaling} behaviors $\vec{P}$, the $g$-, $t$-, and $s$-measures are always bounded below by $\min_{\scriptscriptstyle{i\text{\,=\,}1,...,8}}\big\{1\,,\,\vec{e}_i^{\scriptscriptstyle{\,T}}\!\cdot\vec{P}\,\big\}$. 

We first demonstrate this for the representative vectors $\vec{g}_{\beta}$, $\vec{t}_{\beta}$ and $\vec{t}_{\gamma}$ in Eq.\,(\ref{representatives}), for which we can verify
\begin{eqnarray}\label{ineq-1}
\vec{g}_{\alpha}^{\scriptscriptstyle{\,T}}\!\cdot\vec{P}&=&\Delta_{\scriptscriptstyle  1}+\Delta_{\scriptscriptstyle  4}+P_{\scriptscriptstyle01|01}+1\ \geqslant1\,,\\\label{ineq-2}
\vec{t}_{\beta}^{\scriptscriptstyle{\,T}}\!\cdot\vec{P}&=&\Delta_{\scriptscriptstyle  1}+\Delta_{\scriptscriptstyle  2}+\Delta_{\scriptscriptstyle  4}+P_{\scriptscriptstyle01|01}+P_{\scriptscriptstyle00|11}+1\ \geqslant1\,,\\\label{ineq-3}
\vec{t}_{\gamma}^{\scriptscriptstyle{\,T}}\!\cdot\vec{P}&=&\Delta_{\scriptscriptstyle  1}+\Delta_{\scriptscriptstyle  2}-\Delta_{\scriptscriptstyle  4}+P_{\scriptscriptstyle01|01}+P_{\scriptscriptstyle10|11}+1\ \geqslant1\,,
\end{eqnarray}
where $\Delta_i$ denote signaling signatures defined in Eqs.\,(\ref{NS-AB0})-(\ref{NS-BA1}). The last inequality in each of the three cases holds for \textit{non-signaling} behaviours, since in that case $\Delta_{\scriptscriptstyle  i}=0$ for all $i=1,...\,,8$.

For the representatives $\vec{s}_{\beta}$, $\vec{s}_{\gamma}$, and $\vec{s}_{\delta}$ in Eq.\,(\ref{representatives}), we can write
\begin{eqnarray}
\label{ineq-4}
\vec{s}_{\beta}^{\scriptscriptstyle{\,T}}\!\cdot\vec{P}&=&\tfrac{1}{2}\big(\,\vec{e}_7^{\scriptscriptstyle{\,T}}\!\cdot\vec{P}+\Delta_{\scriptscriptstyle  1}+\Delta_{\scriptscriptstyle  2}+\Delta_{\scriptscriptstyle  3}+\Delta_{\scriptscriptstyle  4}+1\,\big) \stackrel{\scriptscriptstyle{\Delta\text{\,=\,}0}}{=}\ \tfrac{1}{2}\big(\,\vec{e}_7^{\scriptscriptstyle{\,T}}\!\cdot\vec{P}+1\,\big)\ \geqslant\ \min\big\{1\,,\,\vec{e}_7^{\scriptscriptstyle{\,T}}\!\cdot\vec{P}\,\big\}\ \geqslant\ \min_{\scriptscriptstyle{i\text{\,=\,}1,...,8}}\big\{1\,,\,\vec{e}_i^{\scriptscriptstyle{\,T}}\!\cdot\vec{P}\,\big\}\,,\\\label{ineq-5}
\vec{s}_{\gamma}^{\scriptscriptstyle{\,T}}\!\cdot\vec{P}&=&\tfrac{1}{2}\big(\,\vec{e}_2^{\scriptscriptstyle{\,T}}\!\cdot\vec{P}+\Delta_{\scriptscriptstyle  1}-\Delta_{\scriptscriptstyle  2}+\Delta_{\scriptscriptstyle  3}-\Delta_{\scriptscriptstyle  4}+1\,\big) \stackrel{\scriptscriptstyle{\Delta\text{\,=\,}0}}{=}\ \tfrac{1}{2}\big(\,\vec{e}_2^{\scriptscriptstyle{\,T}}\!\cdot\vec{P}+1\,\big)\ \geqslant\ \min\big\{1\,,\,\vec{e}_2^{\scriptscriptstyle{\,T}}\!\cdot\vec{P}\,\big\}\ \geqslant\ \min_{\scriptscriptstyle{i\text{\,=\,}1,...,8}}\big\{1\,,\,\vec{e}_i^{\scriptscriptstyle{\,T}}\!\cdot\vec{P}\,\big\}\,,
\\\label{ineq-6}
\vec{s}_{\delta}^{\scriptscriptstyle{\,T}}\!\cdot\vec{P}&=&\tfrac{1}{2}\big(\,\vec{e}_5^{\scriptscriptstyle{\,T}}\!\cdot\vec{P}+\Delta_{\scriptscriptstyle  1}-\Delta_{\scriptscriptstyle  2}+\Delta_{\scriptscriptstyle  3}+\Delta_{\scriptscriptstyle  4}+1\,\big) \stackrel{\scriptscriptstyle{\Delta\text{\,=\,}0}}{=}\ \tfrac{1}{2}\big(\,\vec{e}_5^{\scriptscriptstyle{\,T}}\!\cdot\vec{P}+1\,\big)\ \geqslant\ \min\big\{1\,,\,\vec{e}_5^{\scriptscriptstyle{\,T}}\!\cdot\vec{P}\,\big\}\ \geqslant\ \min_{\scriptscriptstyle{i\text{\,=\,}1,...,8}}\big\{1\,,\,\vec{e}_i^{\scriptscriptstyle{\,T}}\!\cdot\vec{P}\,\big\}\,,
\end{eqnarray}
where the last equality in each of the three cases uses the \textit{non-signaling} assumption. The penultimate inequality in each case follows from the elementary bound $x+1\,\geqslant\,2\,\min\,\{1\,,\,x\}$.

Now, we recall that these representative vectors $\vec{g}_{\beta}$, $\vec{t}_{\beta}$, $\vec{t}_{\gamma}$, $\vec{s}_{\beta}$, $\vec{s}_{\gamma}$, and $\vec{s}_{\delta}$ generate all remaining \textit{g}-, \textit{t}-, and \textit{s}-type vectors under the symmetry group defined by Eqs.\,(\ref{sym-1})-(\ref{sym-5}). Since \textit{non-signaling} is preserved under the symmetry transformations, the inequalities in Eqs.\,(\ref{ineq-1})-(\ref{ineq-6}) extend directly to the full orbits\footnote{As a technical remark, recall that the inverse of a permutation matrix is its transpose. Indeed, any permutation can be written as a product of transpositions, and the claim holds for each transposition. Clearly, this applies to the symmetries in Eqs.~(\ref{sym-1})-(\ref{sym-5}) since they act as permutations on $\vec{P}$.\\ \\ \\ \\ \\}, namely,
\begin{eqnarray}
\vec{g}_{\scriptscriptstyle  i}^{\scriptscriptstyle{\,T}}\!\cdot\vec{P}&\geqslant&1\,,\qquad\qquad\qquad\qquad\qquad\ \ \text {for \ $i=1,...\,,16$}\,,\\
\vec{t}_{\scriptscriptstyle  j}^{\scriptscriptstyle{\,T}}\!\cdot\vec{P}&\geqslant&1\,,\qquad\qquad\qquad\qquad\qquad\ \ \text {for \ $j=1,...\,,32$}\,,\\
\vec{s}_{\scriptscriptstyle  k}^{\scriptscriptstyle{\,T}}\!\cdot\vec{P}&\geqslant&\min_{\scriptscriptstyle{i\text{\,=\,}1,...,8}}\big\{1\,,\,\vec{e}_i^{\scriptscriptstyle{\,T}}\!\cdot\vec{P}\,\big\}\,,\qquad\qquad\text {for \ $k=1,...\,,64$}\,,
\end{eqnarray}
for all \textit{non-signaling} $\vec{P}$.
We have thus established Eq.\,(\ref{minQ-e}), thereby completing the proof.

\end{proof}

\vspace{0.3cm}

\noindent\textbf{\textsf{\text{B.5. Prevalence of the solution vectors}}}\vspace{0.3cm}

\noindent\textit{\textsf{\textit{B.5.1. {Non-redundancy of the solution sets $\mathbb{Q}$ and $\mathbb{S}$}}}}\vspace{0.2cm}

It is natural to ask whether all 128 and 120 measures in the respective solution sets $\mathbb{Q}$ and $\mathbb{S}$ are indeed essential, that is, whether any of the vectors can be discarded without loss of generality. We find that each vector plays a necessary role, since for every one of them there exists a behavior $\vec{P}$ for which the minimum in Eqs.\,(\ref{muL-solution}) or (\ref{muNS-solution}) is attained uniquely with that vector.

Consider the solution set $\mathbb{Q}$ in Eq.\,(\ref{muL-solution})/(\ref{Q}), which consists of 128 vectors, each with entries equal to 0 or 1. For every vector $\vec{q} \in \mathbb{Q}$, construct a corresponding behavior $\vec{P}^{\scriptscriptstyle (q)}$ by replacing all 0’s with 1’s and all 1’s with 0’s in $\vec{q}$, and uniformly normalizing the result such that $\sum_{\scriptscriptstyle ab}\vec{P}^{\scriptscriptstyle (q)}_{\scriptscriptstyle ab|xy}=1$ for each $x,y=0,1$. This construction ensures that for each such behavior
$\vec{q}^{\scriptscriptstyle\,T}\!\cdot \vec{P}^{\scriptscriptstyle (q)} = 0$,
whereas for any other vector $\vec{q}' \neq \vec{q}$ one has
$\vec{q}'^{\scriptscriptstyle\,T}\!\cdot \vec{P}^{\scriptscriptstyle (q)} > 0$ (since some entry satisfies $q_i=0$ and $q’_i=1$).
Hence, the minimum in Eq.\,(\ref{muL-solution}) is attained uniquely for the vector $\vec{q}$, demonstrating that every element of the solution set $\mathbb{Q}$ is necessary for the completeness of the solution structure.

For the solution set $\mathbb{S}$ in Eqs.\,(\ref{muNS-solution})/(\ref{S}), we obtain a similar structure, except that in the \textit{z}-vectors one of the entries equal to 1 is replaced with a 2. Applying the same construction as above (with 2’s also replaced by 0’s), one can verify that for each vector $\vec{q} \in \mathbb{S}$ there exists a behavior $\vec{P}^{\scriptscriptstyle (q)}$ for which the minimum in Eq.\,(\ref{muNS-solution}) is attained uniquely with that vector.
Hence, all 120 vectors in the solution set $\mathbb{S}$ are likewise essential.

\vspace{0.3cm}
\noindent\textit{\textsf{\textit{B.5.2. {Details of sampling method in Fig.~\ref{Fig-Weights}}}}}\vspace{0.2cm}

The relevance of particular vectors within the solution sets $\mathbb{Q}$ and $\mathbb{S}$ can be analyzed by estimating, over randomly chosen behaviors, the frequency with which each measure uniquely attains the minimum in Eqs.~(\ref{muL-solution}) and (\ref{muNS-solution}), respectively. We numerically evaluated those frequencies and collected the results for each vector class, as presented in Fig.~\ref{Fig-Weights}.

In our analysis, we generated 500,000 random behaviors $\vec{P}\in\mathcal{P}$, obtaining the following results (note, the locality and non-signaling calculations are based on the same set of vectors):

\begin{table}[hbt!]
\begin{minipage}[t]{60mm}
\textbf{Locality}\vspace{0.5pt}\\
\begin{tabular}{||c|c|c||}
\hline
\ Class\ \ &\ No. of cases\ \ &\ \%\ \ \\
\hline\hline
$f$'s&380,881&76.18\\
$g$'s&86,744 &17.35\\
$t$'s&13,096  & 2.62\\
$s$'s&19,069 &3.81\\
$e$'s&210 &0.04\\
\hline
\end{tabular}
\end{minipage}
\begin{minipage}[t]{60mm}
\textbf{Non-signaling}\vspace{0.5pt}\\
\begin{tabular}{||c|c|c||}
\hline
\ Class\ \ &\ No. of cases\ \ &\ \%\ \ \\
\hline\hline
$f$'s&389,161&77.83\\
$g$'s&91,363 &18.27\\
$t$'s&15,664  & 3.13\\
$z$'s&3,812 &0.76\\
\hline
\end{tabular}
\end{minipage}
\end{table}

The generated sixteen-component vectors $\vec{P}$ were normalised so that $\sum_{ab} P_{\scriptscriptstyle ab|xy}=1$ for each $x,y$ (see Eq.~(\ref{Pab|xy})). To implement this constraint, we sampled four independent uniform Dirichlet distributions, \texttt{Dir}(1,1,1,1), one for each quartet $\{P_{\scriptscriptstyle ab|00}\}$, $\{P_{\scriptscriptstyle ab|01}\}$, $\{P_{\scriptscriptstyle ab|10}\}$, and $\{P_{\scriptscriptstyle ab|11}\}$. In practice, we used the function $\texttt{'np.random.dirichlet'}$ from the \texttt{NumPy} package in \texttt{Python}.

The percentages, in general, follow expectation. Note that, regarding the percentages of $z$- versus $s$-solutions, the $z$- and $s$-vectors are in one-to-one correspondence and have the same form, except that one of the 1’s in an $s$-vector becomes a 2 in the corresponding $z$-vector. As a result, a $z$-measure can evaluate to a larger value than its $s$-counterpart, making $z$-measures less likely to attain the minimum and hence the minimizing class tends to migrate from the $s$- to the $f$-, $g$-, $t$-, or $z$-\textit{measures}. Clearly, in the case of the $e$-measures, the minimizer must shift to other classes altogether. For the generated dataset, the following table gives the exact counts of behaviors for which the minimizing vector measure changes class when switching from $\mathbb{Q}$ (measuring locality) to $\mathbb{S}$ (measuring non-signaling).

\begin{table}[hbt!]
\begin{minipage}[t]{60mm}
\textbf{Migration between classes\\when switching form $\mathbb{Q}$ to $\mathbb{S}$}\vspace{0.5pt}\\
\begin{tabular}{||c|c||}
\hline
\ Class change\ \ &\ No. of cases\ \ \\
\hline\hline
$s$'s to $f$'s&8,131\\
$s$'s to $g$'s&4,576\\
$s$'s to $t$'s&2,557\\
$s$'s to $z$'s&3,805\\
$e$'s to $f$'s&149\\
$e$'s to $g$'s&43\\
$e$'s to $t$'s&11\\
$e$'s to $z$'s&7\\
\hline
\end{tabular}
\end{minipage}
\end{table}

\vspace{1cm}


\newpage

\newpage

\noindent\textsf{\Large \textbf{\underline{Part C}}: Analytical approach}\vspace{0.5cm}


\noindent\textbf{\textsf{\text{C.1. {The vector spaces}}}}\vspace{0.2cm}

In this section, we provide an alternative method for understanding \textbf{Theorems~\ref{theoremL}} and \textbf{\ref{theoremNS}}, based on an explicit construction of the solutions. This would be relevant to readers who wish to know how one could go about identifying the solutions one by one and to explicitly understand the form of the local/non-signalling component of any vector (which might inform extensions of the present work in which the linear programming approach is inapplicable). Even though this method does not constitute a complete proof (because of occasional corrections needed at various points), it offers insight into the structure of the solution space. Note, relative to the main text, we alter the order of the derivation of the various results here, since the construction approach we employ is easier to illustrate in the different order we adopt.

Consider a scenario as described at the start of the \textbf{Results} section in the main paper.  We have four possible states of the space, $i=1,\ldots,4$ and each state yields four possible outcomes $j=1,\ldots,4$. Denote by $v_{ij}$ the probability of outcome $j$ given that we are in state $i$. We imagine a $4 \times 4$ matrix and call the $i$ index the row number and the $j$ index the column number. Here, the conditional probabilities $P_{ab|xy}$ are denoted $v_{ij}$, with $i=1, 2, 3$ and $4$ corresponding to ($x=0, y=0$), ($x=0, y=1$), ($x=1, y=0$) and ($x=1, y=1$) respectively and $j=1, 2, 3$ and $4$ corresponding to ($a=0, b=0$), ($a=0, b=1$), ($a=1, b=0$) and ($a=1, b=1$) respectively.

Now consider the vector space $\mathcal{P}$ denoted by all possible 16-element vectors ${\bf v}=(v_{ij})_{i,j=1,\ldots,4}$. It is clear that all 
\begin{equation}\label{eq:generalgreat}
v_{ij} \geq 0
\end{equation}
 and
\begin{equation}\label{eq:generalsum}
\sum_{j} v_{ij}=1 \hspace{0.5cm} i=1,\ldots,4.
\end{equation}
These are the only restrictions governing the general correlation space $\mathcal{P}$.

We consider two subspaces of $\mathcal{P}$. The first of these is the non-signaling space $\mathcal{P}_{\scriptscriptstyle NS}$. Here vectors, denoted in an analogous way by ${\bf u}=(u_{ij})_{i,j=1,\ldots,4}$, satisfy the above conditions and in addition the four non-signaling conditions from equations (\ref{NS-AB0}-\ref{NS-BA1}) from the main paper, which we write below as
\begin{eqnarray}\label{eq:nonsignaling}
u_{11}+u_{12}=u_{21}+u_{22}, \hspace{0.2cm} u_{31}+u_{32}=u_{41}+u_{42}, \nonumber \\
u_{11}+u_{13}=u_{31}+u_{33}, \hspace{0.2cm} u_{21}+u_{23}=u_{41}+u_{43}.
\end{eqnarray}

The second subspace is the local space $\mathcal{P}_{\scriptscriptstyle{Loc}}$. Here the vectors satisfy the general conditions and the non-signaling conditions, and in addition they satisfy the four locality conditions summarized in inequality (\ref{CHSH}), using the Bell-CHSH expressions from equations (\ref{S1}-\ref{S4}) of the main paper, which we write below as:
\begin{eqnarray}\label{eq:locality}
0 \leq u_{12}+u_{13}+u_{22}+u_{23}+u_{32}+u_{33}-u_{42}-u_{43} \leq 2, \nonumber \\
0 \leq u_{12}+u_{13}+u_{22}+u_{23}-u_{32}-u_{33}+u_{42}+u_{43} \leq 2, \nonumber \\
0 \leq u_{12}+u_{13}-u_{22}-u_{23}+u_{32}+u_{33}+u_{42}+u_{43} \leq 2, \nonumber \\
0 \leq -u_{12}-u_{13}+u_{22}+u_{23}+u_{32}+u_{33}+u_{42}+u_{43} \leq 2.
\end{eqnarray}

The {\it closeness} from vector {\bf v} to vector {\bf u} is defined as the maximal value of $p \in [0,1]$ for which
\begin{equation}\label{eq:pdistance}
v_{ij}=pu_{ij}+(1-p) w_{ij} \hspace{0.5cm} i,j=1,\ldots,4
\end{equation}
for some behavior ${\bf w}=(w_{ij})_{i,j=1,\ldots,4} \in \mathcal{P}$,  $p \in [0, 1]$ 

Note that since $\sum_{j} v_{ij}=1$ and $\sum_{j} u_{ij}=1$ we automatically have that $\sum_{j} w_{ij}=1$ for any vectors satisfying (\ref{eq:pdistance}). We thus just need to find the largest value of $p$ for which all $w_{ij}$ terms are non-negative. Rearranging (\ref{eq:pdistance}) gives 
\begin{equation}
w_{ij}=\frac{v_{ij}-pu_{ij}}{1-p} \geq 0 \Rightarrow
\end{equation}
\begin{equation}\label{eq:plessvoveru}
p \leq \frac{v_{ij}}{u_{ij}}
\end{equation}
for all positive $u_{ij}$. For $u_{ij}=0$ (or even taking negative values, as we consider later in the extended space $\mathcal{P}_{\scriptscriptstyle e}$) the $w_{ij}$ terms in (\ref{eq:pdistance}) are clearly non-negative.

The largest such value of $p$, the \textit{closeness} from {\bf v} to {\bf u}, is thus given by
\begin{equation}\label{eq:pvuasmin}
p({\bf v}, {\bf u})= \min_{i,j} \left( \frac{v_{ij}}{u_{ij}} \right).
\end{equation}

Similarly we can define the \textit{closeness} of a vector to a vector space by maximizing the closeness from the focal vector over all elements of the target space. For example, the closeness of {\bf v} to $\mathcal{P}_{\scriptscriptstyle NS}$ is
\begin{equation}
p({\bf v}, \mathcal{P}_{\scriptscriptstyle NS})=\max_{{\bf u} \in \mathcal{P}_{\scriptscriptstyle NS}} p({\bf v}, {\bf u}) .
\end{equation}
The closeness of a vector to the non-signaling space (or the local space) is then the non-signaling (or local) fraction as defined in the main paper.

The non-signaling fraction $\mu_{\scriptscriptstyle  NS}$ (see \textbf{Definition~\ref{Def-muNS}} of the main paper) is then simply 
\begin{equation}
\mu_{\scriptscriptstyle  NS}=p({\bf v}, \mathcal{P}_{\scriptscriptstyle NS}).
\end{equation}
Similarly the local fraction (see \textbf{Definition~\ref{Def-muL}} of the main paper) is 
\begin{equation}
\mu_{\scriptscriptstyle  L}=p({\bf v}, \mathcal{P}_{\scriptscriptstyle Loc}).
\end{equation}
Note that the closeness of a vector to a space that the vector is an element of, for example $p({\bf v}, \mathcal{P})$ for any vector $v$ (or $p({\bf v}, \mathcal{P}_{\scriptscriptstyle NS})$ if ${\bf v}$ is local), is 1.

\vspace{0.3cm}
\noindent\textbf{\textsf{\text{C.2. {Vector space parametrizations}}}}
\vspace{0.2cm}

The space $\mathcal{P}$ is easy to represent generally, this being all vectors that are non-negative and satisfy the four summations from (\ref{eq:generalsum}). This gives a 12 dimensional space (three free entries in each row, since the fourth is determined from the others to ensure a sum of one).

The spaces $\mathcal{P}_{\scriptscriptstyle NS}$ and $\mathcal{P}_{\scriptscriptstyle Loc}$ are more complicated and we find it useful to use the following parametrization.

For $\mathcal{P}_{\scriptscriptstyle NS}$ we have the four additional constraints from (\ref{eq:nonsignaling}). These are equality constraints, so that they reduce the dimension of the space by a further 4 dimensions to 8. Our parametrization is then;
\begin{eqnarray}\label{eq:parametrization}
u_{11}=1-\lambda_{2}-\lambda_{4}-\lambda_{8}, \hspace{0.2cm} u_{12}=\lambda_{2}+\lambda_{4}, \hspace{0.2cm} u_{13}=\lambda_{8}-\lambda_{6}, \hspace{0.2cm} u_{14}=\lambda_{6}; \nonumber \\
u_{21}=1-\lambda_{1}-\lambda_{3}-\lambda_{8}, \hspace{0.2cm} u_{22}=\lambda_{1}+\lambda_{3}, \hspace{0.2cm} u_{23}=\lambda_{8}-\lambda_{5}, \hspace{0.2cm} u_{24}=\lambda_{5}; \nonumber \\
u_{31}=1-\lambda_{2}-\lambda_{6}-\lambda_{7}, \hspace{0.2cm} u_{32}=\lambda_{2}+\lambda_{6}, \hspace{0.2cm} u_{33}=\lambda_{7}-\lambda_{4}, \hspace{0.2cm} u_{34}=\lambda_{4}; \nonumber \\
u_{41}=1-\lambda_{1}-\lambda_{5}-\lambda_{7}, \hspace{0.2cm} u_{42}=\lambda_{1}+\lambda_{5}, \hspace{0.2cm} u_{43}=\lambda_{7}-\lambda_{3}, \hspace{0.2cm} u_{44}=\lambda_{3}.
\end{eqnarray}
We see that, with this parametrization, all the conditions from (\ref{eq:generalsum}) and (\ref{eq:nonsignaling}) are automatically satisfied, and the only additional constraints that define the space are the non-negativeness constraints,
\begin{equation}\label{eq:ugreat}
u_{ij} \geq 0.
\end{equation}
We note that it is also convenient to define a superspace of $\mathcal{P}_{\scriptscriptstyle NS}$, denoted by $\mathcal{P}_{\scriptscriptstyle e}$, which is the above space, but without the constraints (\ref{eq:ugreat}), allowing the $u_{ij}$ vectors to break the standard probabilistic restrictions. This is useful when considering maximum signaling and its relationship to the spatial arguments developed in the main paper and below (see Lemma 1).

For $\mathcal{P}_{\scriptscriptstyle Loc}$ we have, as well, the four constraints from (\ref{eq:nonsignaling}). These are inequality constraints, so while they reduce the size of the space, they do not reduce its dimension. We simply then have the same parametrization from (\ref{eq:parametrization}) with the constraints (\ref{eq:locality}) rewritten in terms of our parameters as
\begin{eqnarray}\label{eq:parametrizationlocal}
0 \leq \lambda_{1}+\lambda_{4}+\lambda_{8}-\lambda_{6} \leq 1, \nonumber \\
0 \leq \lambda_{2}+\lambda_{3}+\lambda_{8}-\lambda_{5} \leq 1, \nonumber \\
0 \leq \lambda_{1}+\lambda_{6}+\lambda_{7}-\lambda_{4} \leq 1, \nonumber \\
0 \leq \lambda_{2}+\lambda_{5}+\lambda_{7}-\lambda_{3} \leq 1,
\end{eqnarray}
as well as the (\ref{eq:ugreat}) constraints.

\vspace{0.3cm}
\noindent\textbf{\textsf{\text{C.3. {Maximum signaling}}}}
\vspace{0.2cm}

By combining the non-signaling equations (\ref{eq:nonsignaling}) with the row summation equations, we obtain precisely eight collections of vector elements whose summations are guaranteed to be equal to 1 when we have non-signaling (these are easy to verify using the parametrization (\ref{eq:parametrization})). These are:
\begin{eqnarray}\label{eq:eightfours}
f_{1}=v_{11}+v_{12}+v_{23}+v_{24}, \nonumber \\
f_{2}=v_{13}+v_{14}+v_{21}+v_{22}, \nonumber \\
f_{3}=v_{31}+v_{32}+v_{43}+v_{44}, \nonumber \\
f_{4}=v_{33}+v_{34}+v_{41}+v_{42}, \nonumber \\
f_{5}=v_{11}+v_{13}+v_{32}+v_{34}, \nonumber \\
f_{6}=v_{12}+v_{14}+v_{31}+v_{33}, \nonumber \\
f_{7}=v_{21}+v_{23}+v_{42}+v_{44}, \nonumber \\
f_{8}=v_{22}+v_{24}+v_{41}+v_{43}.
\end{eqnarray}
These vector sums are the eight $f$-measures, as introduced in the main paper. Recall that the corresponding operators are referred to as $f$- vectors etc. 

We note that from the above 
\begin{equation}
f_{1}+f_{2}=f_{3}+f_{4}=f_{5}+f_{6}=f_{7}+f_{8}=2.
\end{equation}

It is easy to see that there is one general form and that each can be obtained as a simple renumbering of $f_{1}$. We demonstrate this in the final section of this \textbf{Supplementary Information} (Part D).

Thus, we can consider any single four element measure, without loss of generality, to show Lemma 1 below. \\

When we have non-signaling, maximum signaling ($\Delta$ from the main paper, see equation (\ref{NS})) is 0. When we have signaling, it is greater than 0 and is given from the first statement from the following theorem. The second statement then shows the equivalence of maximal signaling to one minus the closeness to the extended non-signaling space $\mathcal{P}_{\scriptscriptstyle e}$.\vspace{0.1cm}

\noindent\textbf{Lemma 1a} (Maximal signaling)\textbf{.} \\
$\Delta = 1-\min \{f_{1}, f_{2}, f_{3}, f_{4}, f_{5}, f_{6}, f_{7}, f_{8} \}$
\vspace{0.3cm}

\noindent\textit{Proof.}\ \\
From equations (\ref{NS-AB0}-\ref{NS-BA1}) in the main paper we have that:
\begin{eqnarray}\label{eq:Deltasandfs}
\Delta_{1}=v_{11}+v_{12}-v_{21}-v_{22}=f_{1}-1=1-f_{2}, \nonumber \\
\Delta_{2}=v_{31}+v_{32}-v_{41}-v_{42}=f_{3}-1=1-f_{4}, \nonumber \\
\Delta_{3}=v_{11}+v_{13}-v_{31}-v_{33}=f_{5}-1=1-f_{6}, \nonumber \\
\Delta_{4}=v_{21}+v_{23}-v_{41}-v_{43}=f_{7}-1=1-f_{8}.
\end{eqnarray}
Precisely one of each pair $f_{1}, f_{2}$; $f_{3}, f_{4}$; $f_{5}, f_{6}$ and $f_{7},f_{8}$ is less than 1 (unless both are equal to 1). Thus, the modulus of $\Delta_{1}$ is 1 minus the smaller of $f_{1}$ and $f_{2}$, and similarly for the other $\Delta_{i}$ terms. As $\Delta$ is the largest of the moduli of the four $\Delta_{i}$ terms, Lemma 1a follows.
\vspace{0.3cm}

\noindent\textbf{Lemma 1b} (Maximal signaling and the extended non-signaling space)\textbf{.} \\
$p({\bf v}, \mathcal{P}_{\scriptscriptstyle e}) = \Delta$.
\vspace{0.3cm}


\noindent\textit{Proof.}\ \\
a) First, we show that $p({\bf v}, \mathcal{P}_{\scriptscriptstyle e}) \leq f_{i}$. We give the calculations for our representative measure $f_{1}$ but they are the same for all of the measures. 
We know that from (\ref{eq:plessvoveru}) that for the closeness $p$ to any element of $\mathcal{P}_{\scriptscriptstyle e}$, we have
\begin{equation}
p \leq \frac{v_{11}}{1-\lambda_{2}-\lambda_{4}-\lambda_{8}}, p \leq \frac{v_{12}}{\lambda_{2}+\lambda_{4}}, p \leq \frac{v_{23}}{\lambda_{8}-\lambda_{5}}, p \leq \frac{v_{24}}{\lambda_{5}},
\end{equation}
whenever the right-hand side denominator is positive, with no corresponding condition when it is not. These rearrange to 
\begin{equation}\label{eq:lambdapineqs}
1-\lambda_{2}-\lambda_{4}-\lambda_{8} \leq \frac{v_{11}}{p}, \lambda_{2}+\lambda_{4} \leq \frac{v_{12}}{p}, \lambda_{8}-\lambda_{5} \leq \frac{v_{23}}{p}, \lambda_{5} \leq \frac{v_{24}}{p}.
\end{equation}
We note that the above conditions clearly hold when the right-hand side terms are negative or zero.

Summing the four terms in (\ref{eq:lambdapineqs}) gives
\begin{equation}
1 \leq \frac{v_{11}+v_{12}+v_{23}+v_{24}}{p} \Rightarrow p \leq f_{1}.
\end{equation}
The closeness measure to any vector (in main text we use the term behavior) in the space is thus bounded above by $f_{1}$, so the closeness to the space itself is similarly bounded above by $f_{1}$, i.e. $p({\bf v}, \mathcal{P}_{\scriptscriptstyle e}) \leq f_{1}$ and more generally we have that 
\begin{equation}
p({\bf v}, \mathcal{P}_{\scriptscriptstyle e}) \leq \min \{f_{1}, f_{2}, f_{3}, f_{4}, f_{5}, f_{6}, f_{7}, f_{8} \}.
\end{equation}

b) Now suppose, without loss of generality, that: \\
$\min \{f_{1}, f_{2}, f_{3}, f_{4}, f_{5}, f_{6}, f_{7}, f_{8} \}=f_{1}$. We now find a vector {\bf u} for which the closeness value $p({\bf v}, {\bf u})=f_{1}$ can be achieved. This would then imply that $p({\bf v}, \mathcal{P}_{\scriptscriptstyle e}) \geq \min \{f_{1}, f_{2}, f_{3}, f_{4}, f_{5}, f_{6}, f_{7}, f_{8} \}$, so that combining with a) above would complete the proof.

Solving for $v_{ij}=u_{ij}p$ for the four component parts of the 
f-measure $f_{1}$ we obtain 
\begin{equation}\label{eq:fvectorlambdasfirst}
\lambda_{5}=\frac{v_{24}}{p}, \lambda_{8}=\frac{v_{23}+v_{24}}{p}, \lambda_{2}+\lambda_{4}=\frac{v_{12}}{p}
\end{equation}
and 
\begin{equation}
p=v_{11}+v_{12}+v_{23}+v_{24}=f_{1}.
\end{equation}
Thus, for these four components the inequality  (\ref{eq:plessvoveru}) is satisfied with equality. To show that this yields a vector that satisfies $p({\bf v}, {\bf u})=f_{1}$, we need to show that the other 12 terms also all satisfy (\ref{eq:plessvoveru}) (with equality or not).

To do this, we need to select the remaining values of $\lambda_{i}$ appropriately, and there are eight distinct cases, depending upon which of $f_{3}$ or $f_{4}$, $f_{5}$ or $f_{6}$, $f_{7}$ or $f_{8}$ is smaller than one (recall that this will be the case for exactly one in each case, unless they are equal to 1, as the three pairs each sum to 2).

The construction we select is as follows: \\
Pick either row 3 or row 4. Select one element from this row to be satisfied with inequality, with the other three satisfied with equality. Then, from the other of rows 3 and 4 pick either the first and third or second and fourth elements to be satisfied with inequality and the other two with equality. This actually yields 16 possible choices, but we only need to select 8 of these to obtain all the relevant cases (two of the above choices yield each case, whichever is picked).

We demonstrate this for four of the eight cases, which then shows the pattern for how the solutions work.

\begin{table}[hbt!]
\begin{center}
\begin{tabular}{||c|cccc||}
\hline
 & $v_{i1}$ & $v_{i2}$ & $v_{i3}$ & $v_{i4}$ \\
\hline
Row 3 & $>$ & $=$ & $>$ & $=$ \\
Row 4 & $>$ & $=$ & $=$ & $=$ \\
\hline
\end{tabular}
\caption{Case A: This case yields an appropriate construction when $f_{3}, f_{5}$ and $f_{7}$ are less than 1. For case A we obtain the following parameter choices: 
$\lambda_{1}=(v_{42}-v_{24})/p, \lambda_{3}=v_{44}/p,  \lambda_{4}=v_{34}/p, \lambda_{6}=(v_{32}+v_{34}-v_{12})/p, \lambda_{7}=(v_{43}+v_{44})/p.$}
\label{tab:CaseAtheorem 1}
\end{center}
\end{table}

\begin{table}[hbt!]
\begin{center}
\begin{tabular}{||c|cccc||}
\hline
 & $v_{i1}$ & $v_{i2}$ & $v_{i3}$ & $v_{i4}$ \\
\hline
Row 3 & $>$ & $=$ & $=$ & $=$ \\
Row 4 & $>$ & $=$ & $>$ & $=$ \\
\hline
\end{tabular}
\caption{Case B: This case yields an appropriate construction when $f_{4}, f_{5}$ and $f_{7}$ are less than 1. For case B we obtain the following parameter choices: 
$\lambda_{1}=(v_{42}-v_{24})/p, \lambda_{3}=v_{44}/p,  \lambda_{4}=v_{34}/p, \lambda_{6}=(v_{32}+v_{34}-v_{12})/p, \lambda_{7}=(v_{33}+v_{34})/p.$}
\label{tab:CaseBtheorem 1}
\end{center}
\end{table}

\begin{table}[hbt!]
\begin{center}
\begin{tabular}{||c|cccc||}
\hline
 & $v_{i1}$ & $v_{i2}$ & $v_{i3}$ & $v_{i4}$ \\
\hline
Row 3 & $=$ & $>$ & $=$ & $>$ \\
Row 4 & $=$ & $>$ & $=$ & $=$ \\
\hline
\end{tabular}
\caption{Case C: This case yields an appropriate construction when $f_{3}, f_{6}$ and $f_{8}$ are less than 1. For case C we obtain the following parameter choices: 
$\lambda_{1}=(v_{11}+v_{12}+v_{23}-v_{41}-v_{43}-v_{44})/p, \lambda_{3}=v_{44}/p,  \lambda_{4}=(v_{43}+v_{44}-v_{33})/p, \lambda_{6}=(v_{11}+v_{23}+v_{24}-v_{31}-v_{33})/p, \lambda_{7}=(v_{43}+v_{44})/p.$}
\label{tab:CaseCtheorem 1}
\end{center}
\end{table}

\begin{table}[hbt!]
\begin{center}
\begin{tabular}{||c|cccc||}
\hline
 & $v_{i1}$ & $v_{i2}$ & $v_{i3}$ & $v_{i4}$ \\
\hline
Row 3 & $=$ & $>$ & $=$ & $>$ \\
Row 4 & $>$ & $=$ & $=$ & $=$ \\
\hline
\end{tabular}
\caption{Case D: This case yields an appropriate construction when $f_{3}, f_{6}$ and $f_{7}$ are less than 1.For case D we obtain the following parameter choices: 
$\lambda_{1}=(v_{42}-v_{24})/p, \lambda_{3}=v_{44}/p,  \lambda_{4}=(v_{43}+v_{44}-v_{33})/p, \lambda_{6}=(v_{11}+v_{23}+v_{24}-v_{31}-v_{33})/p, \lambda_{7}=(v_{43}+v_{44})/p.$}
\label{tab:CaseDtheorem 1}
\end{center}
\end{table}

We see from Tables \ref{tab:CaseAtheorem 1}-\ref{tab:CaseDtheorem 1} how the cases relate to each other. B is a swap of the signs in rows 3 and 4 from A and this keeps the $f_{5}/f_{6}$ and $f_{7}/f_{8}$ relationships the same, but swaps the pair $f_{3}/f_{4}$. From A again, D changes the pair of inequalities in row 3, but keeps row 4 the same (note this is the row now with the single inequality). This swaps the $f_{5}/f_{6}$ relationship, but keeps the others the same. C moves the inequality element in row 4 one place back from D, keeping all else the same. This swaps the $f_{7}/f_{8}$ relationship, whilst keeping the others the same.

From Table \ref{tab:fvectorsTheorem1}, we see how much the targets $v_{ij}$ are above the corresponding $u_{ij}p$ term for each of the four cases above. 

\begin{table}[hbt!]
\begin{center}
\begin{tabular}{||c|c|c|c|c|c||}
\hline
Vector $v_{ij}$ & Local parametrization $u_{ij}$ & Case A & Case B & Case C & Case D \\
\hline
$v_{11}$ & $1-\lambda_{2}-\lambda_{4}-\lambda_{8}$ & 0 & 0 & 0 & 0 \\
$v_{12}$ & $\lambda_{2}+\lambda_{4}$ & 0 & 0 & 0 & 0 \\
$v_{13}$ & $\lambda_{8}-\lambda_{6}$ & $f_{5}-f_{1}$ & $f_{5}-f_{1}$ & $1-f_{6}$ & $1-f_{6}$ \\
$v_{14}$ & $\lambda_{6}$ & $1-f_{5}$ & $1-f_{5}$ & $f_{6}-f_{1}$ & $f_{6}-f_{1}$ \\
$v_{21}$ & $1-\lambda_{1}-\lambda_{3}-\lambda_{8}$ & $f_{7}-f_{1}$ & $f_{7}-f_{1}$ & $1-f_{8}$ & $f_{7}-f_{1}$ \\
$v_{22}$ & $\lambda_{1}+\lambda_{3}$ & $1-f_{7}$ & $1-f_{7}$ & $f_{8}-f_{1}$ & $1-f_{7}$ \\
$v_{23}$ & $\lambda_{8}-\lambda_{5}$ & 0 & 0 & 0 & 0 \\
$v_{24}$ & $\lambda_{5}$ & 0 & 0 & 0 & 0 \\
$v_{31}$ & $1-\lambda_{2}-\lambda_{6}-\lambda_{7}$ & $f_{3}-f_{1}$ & $1-f_{1}$ & 0 & 0 \\
$v_{32}$ & $\lambda_{2}+\lambda_{6}$ & 0 & 0 & $f_{3}-f_{1}$ & $f_{3}-f_{1}$ \\
$v_{33}$ & $\lambda_{7}-\lambda_{4}$ & $1-f_{3}$ & 0 & 0 & 0 \\
$v_{34}$ & $\lambda_{4}$ & 0 & 0 & $1-f_{3}$ & $1-f_{3}$ \\
$v_{41}$ & $1-\lambda_{1}-\lambda_{5}-\lambda_{7}$ & $1-f_{1}$ & $f_{4}-f_{1}$ & 0 & $1-f_{1}$ \\
$v_{42}$ & $\lambda_{1}+\lambda_{5}$ & 0 & 0 & $1-f_{1}$ & 0 \\
$v_{43}$ & $\lambda_{7}-\lambda_{3}$ & 0 & $1-f_{4}$ & 0 & 0 \\
$v_{44}$ & $\lambda_{3}$ & 0 & 0 & 0 & 0 \\
\hline
\end{tabular}
\caption{The deviations above the target, for each of the four cases A, B, C and D, for the constructions from Tables \ref{tab:CaseAtheorem 1}-\ref{tab:CaseDtheorem 1}.}
\label{tab:fvectorsTheorem1}
\end{center}
\end{table}

\vspace{0.3cm}
\noindent\textbf{\textsf{\text{C.4. {Vector measures}}}}
\vspace{0.2cm}

For Lemma 1, the construction clearly works, as the space $\mathcal{P}_{\scriptscriptstyle e}$ does not impose additional bounds due to the probability constraints and any choices of the $\lambda_{i}$ parameters are acceptable. For the spaces $\mathcal{P}_{\scriptscriptstyle NS}$ and $\mathcal{P}_{\scriptscriptstyle Loc}$ there are additional constraints and the closeness measures to these two spaces are harder to characterize as a result. Note that we have:
\begin{equation}
\mathcal{P}_{\scriptscriptstyle Loc} \subsetneq \mathcal{P}_{\scriptscriptstyle NS} \subsetneq \mathcal{P}_{\scriptscriptstyle e}.
\end{equation}
As a result, instead of eight closeness measures all of a single form (\ref{eq:eightfours}), for $\mathcal{P}_{\scriptscriptstyle NS}$ we need 120 measures (see \textbf{Theorem~\ref{theoremL}} in the main paper), which can be broken down into four distinct forms, two of which break down further into sub-classes, with the different measures within each set relating to a simple renumbering as before. For $\mathcal{P}_{\scriptscriptstyle Loc}$ we need 128 measures (see \textbf{Theorem~\ref{theoremL}} in the main paper), involving two new distinct forms as well as most of those from Theorem 3 for $\mathcal{P}_{\scriptscriptstyle NS}$. In total, we have eleven different classes, which we shall label $F, E, G, S_{1}, S_{2}, S_{3}, T_{1}, T_{2}, Z_{1}, Z_{2}, Z_{3}$. Within each class, all elements are equivalent using the symmetries (\ref{sym-1}-\ref{sym-5}) from the main paper.

We note that in the main paper, the numbering is kept as simple as possible, so for example the $f$-measures are numbered $f_{1},\ldots,f_{8}$, the $s$-measures are numbered $s_{1},\ldots,s_{64}$ etc. In the following, as we construct the different measures, sometimes it is convenient to use a different numbering. Thus, whilst the $f$-measures and $e$-measures use a single subscript identical to the main paper, the $g$-measures and $t$-measures use two subscripts and the $s$-measures and $z$-measures use three. In every case, the simple numerical ordering in the paper can be translated to the one below, by relabeling the measures in increasing numerical order in the natural way, so that $s_{111}$ is $s_{1}$, $s_{112}$ is $s_{2}$ etc.

These closeness measures are as described below: 

(i) First, $F$, the set of eight measures from (\ref{eq:eightfours}) above. The representative measure that we consider is $f_{1}$, as before.

(ii) The second set $E$ are the following eight measures listed in (\ref{eq:eighteights}): 
\begin{eqnarray}\label{eq:eighteights}
e_{1}=v_{11}+v_{14}+v_{22}+v_{23}+v_{32}+v_{33}+v_{42}+v_{43}, \nonumber \\
e_{2}=v_{12}+v_{13}+v_{21}+v_{24}+v_{31}+v_{34}+v_{41}+v_{44}, \nonumber \\
e_{3}=v_{12}+v_{13}+v_{21}+v_{24}+v_{32}+v_{33}+v_{42}+v_{43}, \nonumber \\
e_{4}=v_{11}+v_{14}+v_{22}+v_{23}+v_{31}+v_{34}+v_{41}+v_{44}, \nonumber \\
e_{5}=v_{12}+v_{13}+v_{22}+v_{23}+v_{31}+v_{34}+v_{42}+v_{43}, \nonumber \\
e_{6}=v_{11}+v_{14}+v_{21}+v_{24}+v_{32}+v_{33}+v_{41}+v_{44}, \nonumber \\
e_{7}=v_{12}+v_{13}+v_{22}+v_{23}+v_{32}+v_{33}+v_{41}+v_{44}, \nonumber \\
e_{8}=v_{11}+v_{14}+v_{21}+v_{24}+v_{31}+v_{34}+v_{42}+v_{43}.
\end{eqnarray}

(iii) Third, there is a set $G$ of 16 measures with 5 elements. We describe these as follows. Pick one of the 16 $v_{ij}$ terms. This will feature in precisely two of the $f_{i}$ measures (one from $f_{1}$-$f_{4}$, one from $f_{5}$-$f_{8}$). We shall include the original measure plus the two elements within the other row from each of these two measures (two in each). \\
For example, choose the representative vector element $v_{11}$. This is in $f_{1}$ and $f_{5}$. Recall that $f_{1}=v_{11}+v_{12}+v_{23}+v_{24}, f_{5}=v_{11}+v_{13}+v_{32}+v_{34}$. \\
We thus select $v_{11}$, $v_{23}$, $v_{24}$, $v_{32}$ and $v_{34}$. This gives us the measure: \\
$g_{11}=v_{11}+v_{23}+v_{24}+v_{32}+v_{34}$. \\
Similarly, $g_{ij}$ is the measure following the same process starting with $v_{ij}$. All measures can be achieved as a simple renumbering of $g_{11}$.

(vi) Fourth there is another set of measures with 6 elements, this time 32 of them. We describe these as follows: \\
Pick a pair of $f_{i}$ measures with no elements in common. There are eight such pairs (each of $f_{1}$ and $f_{2}$ with each of $f_{3}$ and $f_{4}$, and each of $f_{5}$ and $f_{6}$ with each of $f_{7}$ and $f_{8}$). Pick a row from the first measure and its non-partner row from the second measure (thus, we must have one of rows 1 and 4 and one of rows 2 and 3 included). The four elements thus selected are either from only two columns or from all four columns. \\
The first sub-case $T_{1}$ is when there are only two columns, where we pick one element from each remaining row, which cannot be in the same column. This gives 4x2x2=16 cases. \\
For example, choose measures $f_{1}$ and $f_{3}$, $f_{1}=v_{11}+v_{12}+v_{23}+v_{24}$ and $f_{3}=v_{31}+v_{32}+v_{43}+v_{44}$. Our four choices are then: \\
$t_{11}=v_{11}+v_{12}+v_{23}+v_{31}+v_{32}+v_{44}$, \\
$t_{12}=v_{11}+v_{12}+v_{24}+v_{31}+v_{32}+v_{43}$, \\
$t_{13}=v_{11}+v_{23}+v_{24}+v_{32}+v_{43}+v_{44}$, \\
$t_{14}=v_{12}+v_{23}+v_{24}+v_{31}+v_{43}+v_{44}$. \\
In the numbering above, the first index represents the pair of the $f_{i}$ terms involved, listed in the natural numerical order 1 through 8 (the pairs being 1,3; 1,4; 2,3; 2,4; 5,7; 5,8; 6,7; 6,8), and the second index the corresponding selected rows and final two entries, again listed in the natural order, as shown above for the $t_{1j}$ vector sums. \\
The second sub-case $T_{2}$ is when we have all four columns, where we pick one element from each remaining row, which must be in partner column pairs (i.e. columns 1 and 4, or columns 2 and 3). This again gives 4x2x2=16 cases. \\
For example, choose vector sums $f_{1}$ and $f_{4}$, $f_{1}=v_{11}+v_{12}+v_{23}+v_{24}$ and $f_{4}=v_{33}+v_{33}+v_{41}+v_{42}$. Our four choices are then: \\
$t_{21}=v_{11}+v_{12}+v_{23}+v_{33}+v_{34}+v_{42}$, \\
$t_{22}=v_{11}+v_{12}+v_{24}+v_{33}+v_{34}+v_{41}$, \\
$t_{23}=v_{11}+v_{23}+v_{24}+v_{34}+v_{41}+v_{42}$, \\
$t_{24}=v_{12}+v_{23}+v_{24}+v_{33}+v_{41}+v_{42}$. \\
Each measure of this form is again a simple renumbering of $t_{11}$ or $t_{21}$.

(v) Fifth, there is a set of 64 measures with six elements, which can be described as follows: \\
Pick one of the 16 $v_{ij}$ terms. As described above, this will feature in precisely two of the $f_{i}$ measures. We shall include the original term plus the two elements within the same row from each of these two measures (one in each). Then pick one of the remaining elements from each of the other rows from the two measures (two times two possibilities). Now select the element from the unused row that makes the total number of elements in columns 1 and 4 equal, and the total number in columns 2 and 3 equal. There is one such possibility. This gives 16x4x1=64 measures. \\
For example, choose $v_{11}$, which is in $f_{1}$ and $f_{5}$. Since $f_{1}=v_{11}+v_{12}+v_{23}+v_{24}, f_{5}=v_{11}+v_{13}+v_{32}+v_{34}$, we thus select $v_{11}$, $v_{12}$ and $v_{13}$ from row 1. This gives four combinations from the other two included rows: \\
$v_{23}, v_{32}$ and to equalize columns 1 and 4 from the last row we need $v_{44}$, \\
$v_{23}, v_{34}$ and we need $v_{42}$, \\
$v_{24}, v_{32}$ and we need $v_{43}$, \\
$v_{24}, v_{34}$ and we need $v_{41}$. \\
We shall denote the above measures as $s_{111}$, $s_{112}$, $s_{113}$ and $s_{114}$ respectively (and in general the first two indices will represent the indices of the initially selected element). For example $s_{111}=v_{11}+v_{12}+v_{13}+v_{23}+v_{32}+v_{44}$. \\
Out of each set of four above, one is such that neither of the selected first two elements are in the unused column, two are such that exactly one is in the unused column, and one is such that both are in the unused column. These give us the three distinct cases, which we denote by $S_{1}, S_{2}$ and $S_{3}$ respectively; thus these three sub-cases have 16, 32 and 16 elements respectively. The 64 measures can be obtained by a simple renumbering of $s_{111}, s_{112}$ or $s_{114}$.

(vi) Finally, there is a set of 64 measures with six elements, closely related to the fifth set above. Here for each vector, there is simply a coefficient 2 in the leading original $v_{ij}$ term. Thus in all cases $z_{ijk}=s_{ijk}+v_{ij}$.

As for the $f_{i}$ measures, there are some relationships between the other measures, for similar reasons. While we will not list all here, it is worth noting that:
\begin{equation}
e_{1}+e_{2}=e_{3}+e_{4}=e_{5}+e_{6}=e_{7}+e_{8}=4.
\end{equation}

\vspace{0.3cm}
\noindent\textbf{\textsf{\text{C.5. {Finding the non-signaling fraction}}}}
\vspace{0.2cm}

As well as \textbf{Lemma~\ref{Lemma-muNS-f}}, the main paper contains three Theorems. \textbf{Theorem~\ref{Theorem-muL-e}} is a result from earlier work, and so our main two results are \textbf{Theorem~\ref{theoremL}} from \textbf{Theorem~\ref{theoremNS}}. Whilst from the narrative in the main paper it makes sense to introduce them in this order, for the sections below in terms of construction of the solutions, it makes sense to address them in the reverse order. Thus, here, we start with the non-signaling space and \textbf{Theorem~\ref{T3}}.

The non-signaling space is characterized by the equality constraints (\ref{eq:nonsignaling}) as well as the constraints related to the laws of probability (\ref{eq:ugreat}). This then requires seven of our 11 classes, $F, G, T_{1}, T_{2}, Z_{1}, Z_{2}, Z_{3}$.
\setcounter{theorem}{2}
\begin{theorem}\label{T3}\ \\
$p({\bf v}, \mathcal{P}_{\scriptscriptstyle NS})$ $= \min \{f_{1}, \ldots , f_{8}, g_{1}, \ldots g_{16}, t_{1}, \ldots t_{32}, z_{1}, \ldots, z_{64} \}$.
\end{theorem}

 Below we give a construction of finding the point in the non-signaling space that achieves the appropriate closeness, itemized by the different measures corresponding to where this could be achieved. 

a) First, we give a way to show that $p({\bf v}, \mathcal{P}_{\scriptscriptstyle NS})$ is bounded above by all 120 measures, which allows structural insight into these solutions, not available in the proof in main text. In each case, we demonstrate the proof just for a single case, but these generalize through the simple renumbering mentioned earlier.

(i) $p({\bf v}, \mathcal{P}_{\scriptscriptstyle NS})$ $\leq f_{i}$. This is trivially true, since $p({\bf v}, \mathcal{P}_{\scriptscriptstyle e}) \leq f_{i}$ and $\mathcal{P}_{\scriptscriptstyle NS}$ is a subspace of $\mathcal{P}_{\scriptscriptstyle e}$.

(ii) $p({\bf v}, \mathcal{P}_{\scriptscriptstyle NS})$ $\leq g_{ij}$. We do the calculation for $g_{11}=v_{11}+v_{23}+v_{24}+v_{32}+v_{34}$. \\
We know from (\ref{eq:plessvoveru}) that for the closeness $p$ to any element of $\mathcal{P}_{\scriptscriptstyle NS}$, we have:
\begin{eqnarray}\label{eq:thefivetlambda}
1-\lambda_{2}-\lambda_{4}-\lambda_{8} \leq \frac{v_{11}}{p}, \lambda_{8}-\lambda_{5} \leq \frac{v_{23}}{p}, 
\lambda_{5} \leq \frac{v_{24}}{p}, \nonumber \\
\lambda_{2}+\lambda_{6} \leq \frac{v_{32}}{p}, \lambda_{4} \leq \frac{v_{34}}{p}.
\end{eqnarray}
Summing all the terms on both sides of the inequalities gives:
\begin{equation}\label{eq:lambdaping}
1+\lambda_{6}\leq \frac{g_{11}}{p}.
\end{equation}
We know that $\lambda_{6}=u_{14}$ and so is non-negative. Thus 
\begin{equation}
1 \leq \frac{g_{11}}{p} \Rightarrow p \leq g_{11}.
\end{equation}
Calculations from the remaining 15 terms are similar, with the equivalent left hand side to (\ref{eq:lambdaping}) being $1+u_{ik}$ for all of the other different 15 $u_{ik}$ terms in turn. We note that the $u_{ik}$ term has a different second index to $g_{ij}$. It is the partner term in the same row, so here we have $g_{11}$ and $u_{14}$.

(iii) $p({\bf v}, \mathcal{P}_{\scriptscriptstyle NS})$ $\leq t_{ij}$. For $T_{1}$ we do the calculation for $t_{11}=v_{11}+v_{12}+v_{23}+v_{31}+v_{32}+v_{44}$. We know from (\ref{eq:plessvoveru}) that:
\begin{eqnarray}\label{eq:thesixtlambda}
1-\lambda_{2}-\lambda_{4}-\lambda_{8} \leq \frac{v_{11}}{p}, \lambda_{2}+\lambda_{4} \leq \frac{v_{12}}{p}, 
\lambda_{8}-\lambda_{5} \leq \frac{v_{23}}{p}, \nonumber \\
1-\lambda_{2}-\lambda_{6}-\lambda_{7} \leq \frac{v_{31}}{p}, \lambda_{2}+\lambda_{6} \leq \frac{v_{32}}{p},  \lambda_{3} \leq \frac{v_{44}}{p}.
\end{eqnarray}
Summing all the terms on both sides of the inequalities gives:
\begin{equation}\label{eq:lambdapineqt}
1+(1-\lambda_{5}-\lambda_{7}+\lambda_{3}) \leq \frac{t_{11}}{p}.
\end{equation}
We know that $1-\lambda_{1}-\lambda_{5}-\lambda_{7}=u_{41}$ and $\lambda_{1}+\lambda_{3}=u_{22}$ so the sum of these, which is the bracketed term in the left hand side of (\ref{eq:lambdapineqt2}), is non-negative. Thus 
\begin{equation}
1 \leq \frac{t_{11}}{p} \Rightarrow p \leq t_{11}.
\end{equation}
Calculations for the remaining 15 terms are similar, with the equivalent left hand side to (\ref{eq:lambdapineqt}) being 1 plus the sum of two $u_{ik}$ terms. These are again the partner terms of the lone elements (this time plural) in the same row. Thus, for the above we had $v_{23}$ and $v_{44}$ leading to $u_{22}$ and $u_{41}$. \\
For $T_{2}$ we do the calculation for $t_{21}=v_{11}+v_{12}+v_{23}+v_{33}+v_{34}+v_{42}$. We now get
\begin{eqnarray}\label{eq:thesixt2lambda}
1-\lambda_{2}-\lambda_{4}-\lambda_{8} \leq \frac{v_{11}}{p}, \lambda_{2}+\lambda_{4} \leq \frac{v_{12}}{p}, 
\lambda_{8}-\lambda_{5} \leq \frac{v_{23}}{p}, \nonumber \\
\lambda_{7} -\lambda_{4} \leq \frac{v_{33}}{p}, \lambda_{4} \leq \frac{v_{34}}{p},  \lambda_{1}+\lambda_{5} \leq \frac{v_{42}}{p}.
\end{eqnarray}
Summing all the terms on both sides of the inequalities gives:
\begin{equation}\label{eq:lambdapineqt2}
1+(\lambda_{1}+\lambda_{7}) \leq \frac{t_{21}}{p}.
\end{equation}
We know that $\lambda_{1}+\lambda_{3}=u_{22}$ and $\lambda_{7}-\lambda_{3}=u_{43}$ so the sum of these, which is the bracketed term in the left hand side of (\ref{eq:lambdapineqt2}), is non-negative. Thus 
\begin{equation}
1 \leq \frac{t_{21}}{p} \Rightarrow p \leq t_{21}.
\end{equation}

(iv) $p({\bf v}, \mathcal{P}_{\scriptscriptstyle NS})$ $\leq z_{ijk}$. For case $Z_{1}$ we do the calculation for $z_{111}=2v_{11}+v_{12}+v_{13}+v_{23}+v_{32}+v_{44}$.
We know from (\ref{eq:plessvoveru}) that:
\begin{eqnarray}\label{eq:thesixnewslambda}
1-\lambda_{2}-\lambda_{4}-\lambda_{8} \leq \frac{v_{11}}{p}, \lambda_{2}+\lambda_{4} \leq \frac{v_{12}}{p}, 
\lambda_{8}-\lambda_{6} \leq \frac{v_{13}}{p}, \nonumber \\
\lambda_{8}-\lambda_{5} \leq \frac{v_{23}}{p}, \lambda_{2}+\lambda_{6} \leq \frac{v_{32}}{p},  \lambda_{3} \leq \frac{v_{44}}{p}.
\end{eqnarray}
Summing all the terms on both sides of the inequalities gives:
\begin{equation}\label{eq:lambdapineqnews}
1+(1-\lambda_{4}-\lambda_{5}+\lambda_{3}) \leq \frac{z_{111}}{p}.
\end{equation}
We know from (\ref{eq:ugreat}) that the bracketed term in (\ref{eq:lambdapineqnews}) is the sum of three non-negative terms $(1-\lambda_{1}-\lambda_{5}-\lambda_{7})+(\lambda_{1}+\lambda_{3})+(\lambda_{7}-\lambda_{4})=u_{41}+u_{22}+u_{33}$. Thus 
\begin{equation}
1 \leq \frac{z_{111}}{p} \Rightarrow p \leq z_{111}.
\end{equation}
Calculations for the remaining 15 terms are similar, with non-negative bracketed terms in each case. \\
For case $Z_{2}$ we do the calculation for $z_{112}=2v_{11}+v_{12}+v_{13}+v_{23}+v_{34}+v_{42}$.
We know from (\ref{eq:plessvoveru}) that:
\begin{eqnarray}\label{eq:thesixnewslambda2}
1-\lambda_{2}-\lambda_{4}-\lambda_{8} \leq \frac{v_{11}}{p}, \lambda_{2}+\lambda_{4} \leq \frac{v_{12}}{p}, 
\lambda_{8}-\lambda_{6} \leq \frac{v_{13}}{p}, \nonumber \\
\lambda_{8}-\lambda_{5} \leq \frac{v_{23}}{p}, \lambda_{4} \leq \frac{v_{34}}{p},  \lambda_{1}+\lambda_{5} \leq \frac{v_{42}}{p}.
\end{eqnarray}
Summing all the terms on both sides of the inequalities gives:
\begin{equation}\label{eq:lambdapineqnews2}
1+(1-\lambda_{2}+\lambda_{1}-\lambda_{6}) \leq \frac{z_{112}}{p}.
\end{equation}
We know from (\ref{eq:ugreat}) that the bracketed term in (\ref{eq:lambdapineqnews2}) is the sum of three non-negative terms $(\lambda_{7}-\lambda_{3})+(\lambda_{1}+\lambda_{3})+(1-\lambda_{2}-\lambda_{6}-\lambda_{7})=u_{43}+u_{22}+u_{31}$. Thus 
\begin{equation}
1 \leq \frac{z_{112}}{p} \Rightarrow p \leq z_{112}.
\end{equation}
Calculations for the remaining 31 terms are similar. \\
For case $Z_{3}$ we do the calculation for $z_{114}=2v_{11}+v_{12}+v_{13}+v_{24}+v_{34}+v_{41}$.
We know from (\ref{eq:plessvoveru}) that:
\begin{eqnarray}\label{eq:thesixnewslambda3}
1-\lambda_{2}-\lambda_{4}-\lambda_{8} \leq \frac{v_{11}}{p}, \lambda_{2}+\lambda_{4} \leq \frac{v_{12}}{p}, 
\lambda_{8}-\lambda_{6} \leq \frac{v_{13}}{p}, \nonumber \\
\lambda_{5} \leq \frac{v_{24}}{p}, \lambda_{4} \leq \frac{v_{34}}{p}, 1-\lambda_{1}-\lambda_{5}-\lambda_{7} \leq \frac{v_{41}}{p}.
\end{eqnarray}
Summing all the terms on both sides of the inequalities gives:
\begin{equation}\label{eq:lambdapineqnews3}
1+(2-\lambda_{1}-\lambda_{2}-\lambda_{6}-\lambda_{7}-\lambda_{8}) \leq \frac{z_{114}}{p}.
\end{equation}
We know from (\ref{eq:ugreat}) that the bracketed term in (\ref{eq:lambdapineqnews3}) is the sum of three non-negative terms $\lambda_{3}+(1-\lambda_{1}-\lambda_{3}-\lambda_{8})+(1-\lambda_{2}-\lambda_{6}-\lambda_{7})=u_{44}+u_{21}+u_{31}$. Thus 
\begin{equation}
1 \leq \frac{z_{114}}{p} \Rightarrow p \leq z_{114}.
\end{equation}
Calculations for the remaining 15 terms are similar.

b) 
We now show how to achieve the bound when a closeness measure from each of the five classes is the minimum. We note that this is not a full proof, because of the ``solution nudging'' that is sometimes required (this is explained below at appropriate places).

(i) We start with the $f_{i}$ measures. There are eight measures, and the construction from Lemma 1 involved eight cases for each. We show the process for $f_{1}$ and the case when $f_{3}, f_{5}$ and $f_{7}$ are less than 1. Thus, following the proof from Lemma 1, we consider the local vector from Case A. We need to show that when $f_{1}$ is the smallest of all 120 closeness measures, so that $p=f_{1}$, (and also that $f_{3}, f_{5}, f_{7}$ are less than 1), we have that the solution specified is in the appropriate space. In the above, we concentrated on meeting the conditions. In fact, some of the conditions can be met by a sufficient margin that the deficit is larger than the original target, implying that the corresponding $u_{ij}$ term is negative. This was not a problem for the space $\mathcal{P}_{\scriptscriptstyle e}$ where such terms were allowed, but cannot occur here. 

We then need to ``nudge'' the solution back into the space. The process is similar for each case. Here, the potential elements that we need to address are those with a non-zero entry in the Case A column of Table \ref{tab:fvectorsTheorem1}. How to shift the $\lambda_{i}$ values is clear from observing the entries corresponding to zeros in this column and so which cannot be increased. Starting from the bottom of the column, $\lambda_{3}$ can only be decreased and $\lambda_{7}$ can also only be decreased, and by at least as much as $\lambda_{3}$. $\lambda_{1}+\lambda_{5}$ can only be decreased. $\lambda_{5}$ is fixed from the original set up to hit the target, and so $\lambda_{1}$ can only be decreased. $\lambda_{4}$ can only be decreased. $\lambda_{2}+\lambda_{6}$ can only be decreased. From the original requirement  to hit the target, we know that $\lambda_{2}+\lambda_{4}$ is fixed, so as $\lambda_{4}$ can only be decreased, $\lambda_{2}$ can only be increased, so in turn $\lambda_{6}$ can only be decreased. We know that $\lambda_{8}$ is fixed due to the original set up. Thus, we have various shifts in direction, which we can make to nudge our solutions back into place.

When is this possible? As we can only decrease $\lambda_{6}$, this means that we must have $v_{14} \geq 1-f_{5}$, as $u_{14}$ must already be non-negative (we cannot increase it). $v_{14}-(1-f_{5})=v_{14}+v_{11}+v_{13}+v_{32}+v_{34}-1=v_{32}+v_{34}-v_{12}=g_{11}-f_{1} \geq 0$, so this holds. We can then reduce $\lambda_{6} p$ by up to $g_{11}-f_{1}$, which ensures $u_{13}$ is also non-negative. A similar argument holds for $\lambda_{1}+\lambda_{3}$ (their summation can be reduced by up to 
$v_{42}+v_{44}-v_{24}=g_{23}-f_{1}$, which then ensures both $u_{21}$ and $u_{22}$ are positive. The remaining three free parameters can then be adjusted suitably to move any of the other three cases into the correct space.

(ii) Now consider the $z_{ijk}$ measures, starting with $Z_{1}$ and choosing $z_{111}$. For this measure, we need the conditions from (\ref{eq:thesixnewslambda}) satisfied with equality. 
Since here $p=z_{111}$ there are eight distinct equations with eight variables. This yields a unique set of $\lambda_{i}$ values, as shown in Table \ref{tab:newsvectorslambda}. 

\begin{table}[hbt!]
\begin{center}
\begin{tabular}{||c|c||}
\hline
$\lambda_{1}$ & $-v_{44}/p$ \\
$\lambda_{2}$ & $-(v_{11}+v_{23}+v_{44})/p$ \\
$\lambda_{3}$ & $v_{44}/p$ \\
$\lambda_{4}$ & $(v_{11}+v_{12}+v_{23}+v_{44})/p$ \\
$\lambda_{5}$ & $(v_{11}+v_{13}+v_{32}+v_{44})/p$ \\
$\lambda_{6}$ & $(v_{11}+v_{23}+v_{32}+v_{44})/p$\\
$\lambda_{7}$ & $(v_{11}+v_{12}+v_{23}+v_{44})/p$ \\
$\lambda_{8}$ & $(v_{11}+v_{13}+v_{23}+v_{32}+v_{44})/p$ \\
\hline
\end{tabular}
\caption{The values of $\lambda_{i}$ when the minimum occurs at $z_{111}$.}
\label{tab:newsvectorslambda}
\end{center}
\end{table}

\begin{table}[hbt!]
\begin{center}
\begin{tabular}{||c|c|c|c||}
\hline
Vector & Local param. $u_{ij}$ & Target deviation & Comparison \\
\hline
$v_{11}$ & $1-\lambda_{2}-\lambda_{4}-\lambda_{8}$ & 0 & 0 \\
$v_{12}$ & $\lambda_{2}+\lambda_{4}$ & 0 & 0 \\
$v_{13}$ & $\lambda_{8}-\lambda_{6}$ & 0 & 0 \\
$v_{14}$ & $\lambda_{6}$ & $v_{14}-v_{11}-v_{23}-v_{32}-v_{44}$ & $1-z_{111} \geq 0$ \\
$v_{21}$ & $1-\lambda_{1}-\lambda_{3}-\lambda_{8}$ & $v_{21}-v_{11}-v_{12}$ & $t_{51}-z_{111} \geq 0$ \\
$v_{22}$ & $\lambda_{1}+\lambda_{3}$ & $v_{22}$  & $ v_{22} \geq 0$ \\
$v_{23}$ & $\lambda_{8}-\lambda_{5}$ & 0 & 0 \\
$v_{24}$ & $\lambda_{5}$ & $v_{24}-v_{11}-v_{13}-v_{32}-v_{44}$ & $f_{1}-z_{111} \geq 0$ \\
$v_{31}$ & $1-\lambda_{2}-\lambda_{6}-\lambda_{7}$ & $v_{31}-v_{11}-v_{13}$ & $t_{11}-z_{111} \geq 0$ \\
$v_{32}$ & $\lambda_{2}+\lambda_{6}$ & 0 & 0 \\
$v_{33}$ & $\lambda_{7}-\lambda_{4}$ & $v_{33}$ & $v_{33} \geq 0$ \\
$v_{34}$ & $\lambda_{4}$ & $v_{34}-v_{11}-v_{12}-v_{23}-v_{44}$ & $f_{5}-z_{111} \geq 0$ \\
$v_{41}$ & $1-\lambda_{1}-\lambda_{5}-\lambda_{7}$ & $v_{41}$ & $v_{41} \geq 0$ \\
$v_{42}$ & $\lambda_{1}+\lambda_{5}$ & $v_{42}-v_{11}-v_{13}-v_{32}$ & $g_{23}-z_{111} \geq 0$ \\
$v_{43}$ & $\lambda_{7}-\lambda_{3}$ & $v_{43}-v_{11}-v_{12}-v_{23}$ & $g_{32}-z_{111} \geq 0$ \\
$v_{44}$ & $\lambda_{3}$ & 0 & 0 \\
\hline
\end{tabular}
\caption{Demonstration of the conditions for a 6 element closeness measure from the first class, here $z_{111}$.The $\lambda_{i}$ values come from Table \ref{tab:newsvectorslambda}.}
\label{tab:newsvectors}
\end{center}
\end{table}

In Table \ref{tab:newsvectors} the measures for comparison with the benchmark $z_{111}$ (not including those that are zero as well as those that achieve the target) are as follows: \\
$1$ \\
$v_{11}+v_{13}+v_{21}+v_{23}+v_{32}+v_{44}$ \\
$v_{11}+v_{12}+v_{23}+v_{24}$ \\
$v_{11}+v_{12}+v_{23}+v_{31}+v_{32}+v_{44}$ \\
$v_{11}+v_{13}+v_{32}+v_{34}$ \\
$v_{11}+v_{12}+v_{23}+v_{42}+v_{44}$ \\
$v_{11}+v_{13}+v_{32}+v_{43}+v_{44}$

Now for $Z_{2}$ we choose $z_{112}$. For this measure, we need the conditions from (\ref{eq:thesixnewslambda2}) satisfied with equality. 
Since here $p=z_{112}$ there are eight distinct equations with eight variables. As before, this yields a unique set of $\lambda_{i}$ values, as shown in Table \ref{tab:newsvectorslambda2}. 

\begin{table}[hbt!]
\begin{center}
\begin{tabular}{||c|c||}
\hline
$\lambda_{1}$ & $-(v_{11}+v_{13}+v_{34})/p$ \\
$\lambda_{2}$ & $(v_{12}-v_{34})/p$ \\
$\lambda_{3}$ & $(v_{11}+v_{13}+v_{34})/p$ \\
$\lambda_{4}$ & $v_{34}/p$ \\
$\lambda_{5}$ & $(v_{11}+v_{13}+v_{34}+v_{42})/p$ \\
$\lambda_{6}$ & $(v_{11}+v_{23}+v_{34}+v_{42})/p$\\
$\lambda_{7}$ & $(v_{11}+v_{13}+v_{34})/p$ \\
$\lambda_{8}$ & $(v_{11}+v_{13}+v_{23}+v_{34}+v_{42})/p$ \\
\hline
\end{tabular}
\caption{The values of $\lambda_{i}$ when the minimum occurs at $z_{112}$.}
\label{tab:newsvectorslambda2}
\end{center}
\end{table}

\begin{table}[hbt!]
\begin{center}
\begin{tabular}{||c|c|c|c||}
\hline
Vector & Local param. $u_{ij}$ & Target deviation & Comparison \\
\hline
$v_{11}$ & $1-\lambda_{2}-\lambda_{4}-\lambda_{8}$ & 0 & 0 \\
$v_{12}$ & $\lambda_{2}+\lambda_{4}$ & 0 & 0 \\
$v_{13}$ & $\lambda_{8}-\lambda_{6}$ & 0 & 0 \\
$v_{14}$ & $\lambda_{6}$ & $v_{14}-v_{11}-v_{23}-v_{34}-v_{42}$ & $1-z_{112} \geq 0$ \\
$v_{21}$ & $1-\lambda_{1}-\lambda_{3}-\lambda_{8}$ & $v_{21}-v_{11}-v_{12}$ & $t_{52}-z_{112} \geq 0$ \\
$v_{22}$ & $\lambda_{1}+\lambda_{3}$ & $v_{22}$  & $v_{22} \geq 0$ \\
$v_{23}$ & $\lambda_{8}-\lambda_{5}$ & 0 & 0 \\
$v_{24}$ & $\lambda_{5}$ & $v_{24}-v_{11}-v_{13}-v_{34}-v_{42}$ & $f_{1}-z_{112} \geq 0$ \\
$v_{31}$ & $1-\lambda_{2}-\lambda_{6}-\lambda_{7}$ & $v_{31}$ & $v_{31} \geq 0$ \\
$v_{32}$ & $\lambda_{2}+\lambda_{6}$ & $v_{32}-v_{11}-v_{12}-v_{23}-v_{42}$ & $f_{5}-z_{112} \geq 0$ \\
$v_{33}$ & $\lambda_{7}-\lambda_{4}$ & $v_{33}-v_{11}-v_{13}$ & $t_{21}-z_{112} \geq 0$ \\
$v_{34}$ & $\lambda_{4}$ & 0 & 0 \\
$v_{41}$ & $1-\lambda_{1}-\lambda_{5}-\lambda_{7}$ & $v_{41}-v_{11}-v_{12}-v_{23}$ & $g_{34}-z_{112} \geq 0$ \\
$v_{42}$ & $\lambda_{1}+\lambda_{5}$ & 0 & 0 \\
$v_{43}$ & $\lambda_{7}-\lambda_{3}$ & $v_{43}$ & $v_{43} \geq 0$ \\
$v_{44}$ & $\lambda_{3}$ & $v_{44}-v_{11}-v_{13}-v_{34}$ & $g_{23}-z_{112} \geq 0$ \\
\hline
\end{tabular}
\caption{Demonstration of the conditions for a 6 element closeness measure from the second class, here $z_{112}$.The $\lambda_{i}$ values come from Table \ref{tab:newsvectorslambda2}.}
\label{tab:newsvectors2}
\end{center}
\end{table}

In Table \ref{tab:newsvectors2} the measures for comparison with the benchmark $z_{112}$ (not including those that are zero as well as those that achieve the target) are as follows: \\
$1$ \\
$v_{11}+v_{13}+v_{21}+v_{23}+v_{34}+v_{42}$ \\
$v_{11}+v_{12}+v_{23}+v_{24}$ \\
$v_{11}+v_{13}+v_{32}+v_{34}$ \\
$v_{11}+v_{12}+v_{23}+v_{33}+v_{34}+v_{42}$ \\
$v_{11}+v_{13}+v_{34}+v_{41}+v_{42}$ \\
$v_{11}+v_{12}+v_{23}+v_{42}+v_{44}$

Finally, for $Z_{3}$ we choose $z_{114}$. For this measure, we need the conditions from (\ref{eq:thesixnewslambda3}) satisfied with equality. Since here $p=z_{114}$ there are eight distinct equations with eight variables. As before, this yields a unique set of $\lambda_{i}$ values, as shown in Table \ref{tab:newsvectorslambda3}.

\begin{table}[hbt!]
\begin{center}
\begin{tabular}{||c|c||}
\hline
$\lambda_{1}$ & $(v_{11}+v_{12})/p$ \\
$\lambda_{2}$ & $(v_{12}-v_{34})/p$ \\
$\lambda_{3}$ & 0 \\
$\lambda_{4}$ & $v_{34}/p$ \\
$\lambda_{5}$ & $v_{24}/p$ \\
$\lambda_{6}$ & $(v_{11}+v_{24}+v_{34}+v_{41})/p$\\
$\lambda_{7}$ & $(v_{11}+v_{13}+v_{34})/p$ \\
$\lambda_{8}$ & $(v_{11}+v_{13}+v_{24}+v_{34}+v_{41})/p$ \\
\hline
\end{tabular}
\caption{The values of $\lambda_{i}$ when the minimum occurs at $z_{114}$.}
\label{tab:newsvectorslambda3}
\end{center}
\end{table}

\begin{table}[hbt!]
\begin{center}
\begin{tabular}{||c|c|c|c||}
\hline
Vect. & Local param. $u_{ij}$ & Target deviation & Comparison \\
\hline
$v_{11}$ & $1-\lambda_{2}-\lambda_{4}-\lambda_{8}$ & 0 & 0 \\
$v_{12}$ & $\lambda_{2}+\lambda_{4}$ & 0 & 0 \\
$v_{13}$ & $\lambda_{8}-\lambda_{6}$ & 0 & 0 \\
$v_{14}$ & $\lambda_{6}$ & $v_{14}-v_{11}-v_{24}-v_{34}-v_{41}$ & $1-z_{114} \geq 0$ \\
$v_{21}$ & $1-\lambda_{1}-\lambda_{3}-\lambda_{8}$ & $v_{21}$ & $v_{21} \geq 0$ \\
$v_{22}$ & $\lambda_{1}+\lambda_{3}$ & $v_{22}-v_{11}-v_{12}$  & $t_{62}-z_{114} \geq 0$ \\
$v_{23}$ & $\lambda_{8}-\lambda_{5}$ & $v_{23}-v_{11}-v_{13}-v_{34}-v_{41}$ & $f_{1}-z_{114} \geq 0$ \\
$v_{24}$ & $\lambda_{5}$ & 0 & 0 \\
$v_{31}$ & $1-\lambda_{2}-\lambda_{6}-\lambda_{7}$ & $v_{31}$  & $v_{31} \geq 0$ \\
$v_{32}$ & $\lambda_{2}+\lambda_{6}$ & $v_{32}-v_{11}-v_{12}-v_{24}-v_{41}$ & $f_{5}-z_{114} \geq 0$ \\
$v_{33}$ & $\lambda_{7}-\lambda_{4}$ & $v_{33}-v_{11}-v_{13}$ & $t_{22}-z_{114} \geq 0$ \\
$v_{34}$ & $\lambda_{4}$ & 0 & 0 \\
$v_{41}$ & $1-\lambda_{1}-\lambda_{5}-\lambda_{7}$ & 0 & 0 \\
$v_{42}$ & $\lambda_{1}+\lambda_{5}$ & $v_{42}-v_{11}-v_{12}-v_{24}$ & $g_{34}-z_{114} \geq 0$ \\
$v_{43}$ & $\lambda_{7}-\lambda_{3}$ & $v_{43}-v_{11}-v_{13}-v_{34}$ & $g_{23}-z_{114} \geq 0$ \\
$v_{44}$ & $\lambda_{3}$ & $v_{44}$ & $v_{44} \geq 0$ \\
\hline
\end{tabular}
\caption{Demonstration of the conditions for a 6 element closeness measure from the third class, here $z_{114}$. The $\lambda_{i}$ values come from Table \ref{tab:newsvectorslambda3}.}
\label{tab:newsvectors3}
\end{center}
\end{table}

In Table \ref{tab:newsvectors3} the measures for comparison with the benchmark $z_{114}$ (not including those that are zero, as well as those that achieve the target) are as follows: \\
$1$ \\
$v_{11}+v_{13}+v_{22}+v_{24}+v_{34}+v_{41}$ \\
$v_{11}+v_{12}+v_{23}+v_{24}$ \\
$v_{11}+v_{13}+v_{32}+v_{34}$ \\
$v_{11}+v_{12}+v_{24}+v_{33}+v_{34}+v_{41}$ \\
$v_{11}+v_{13}+v_{34}+v_{41}+v_{42}$ \\
$v_{11}+v_{12}+v_{24}+v_{41}+v_{43}$

(iii) Now consider the $t_{ij}$ closeness measures, choosing $t_{11}$ from $T_{1}$. For this measure, we need the conditions from (\ref{eq:thesixtlambda}) satisfied with equality. Since here $p=t_{11}$ there are only five distinct equations with eight variables. However, we also need the condition (\ref{eq:lambdapineqt}) to be satisfied with equality, in turn implying that $1-\lambda_{1}-\lambda_{5}-\lambda_{7}=0$ and $\lambda_{1}+\lambda_{3}=0$. This yields a set of $\lambda_{i}$ values with only one free parameter, as shown in Table \ref{tab:tvectorslambda}. In this table we leave this as a general value, but in the subsequent solution we choose $\lambda_{2} p=-\min(v_{34}-v_{12},v_{14}-v_{32})$. \\
As before the focus has been on achieving the conditions on the more challenging cases, and for some parameter values the above conditions are achieved with a sufficient margin that the initially selected point is actually outside of the space (i.e., some 'nudging' back into the correct space is needed). We omit this calculation here, but shall show this for the corresponding case $T_{2}$ below as illustration.

\begin{table}[hbt!]
\begin{center}
\begin{tabular}{||c|c||}
\hline
$\lambda_{1}$ & $-v_{44}/p$ \\
$\lambda_{2}$ & free \\
$\lambda_{3}$ & $v_{44}/p$ \\
$\lambda_{4}$ & $v_{12}/p-\lambda_{2}$ \\
$\lambda_{5}$ & $(v_{31}+v_{32}+v_{44})/p$ \\
$\lambda_{6}$ & $v_{32}/p-\lambda_{2}$\\
$\lambda_{7}$ & $(v_{11}+v_{12}+v_{23}+v_{44})/p$ \\
$\lambda_{8}$ & $(v_{23}+v_{31}+v_{32}+v_{44})/p$ \\
\hline
\end{tabular}
\caption{The values of $\lambda_{i}$ when the minimum occurs at $t_{11}$.}
\label{tab:tvectorslambda}
\end{center}
\end{table}

\begin{table}[hbt!]
\begin{center}
\begin{tabular}{||c|c|c|c||}
\hline
Vector & Local param. $u_{ij}$ & Target deviation & Comparison \\
\hline
$v_{11}$ & $1-\lambda_{2}-\lambda_{4}-\lambda_{8}$ & 0 & 0 \\
$v_{12}$ & $\lambda_{2}+\lambda_{4}$ & 0 & 0 \\
$v_{13}$ & $\lambda_{8}-\lambda_{6}$ & $v_{13}-v_{23}-v_{31}-v_{44}+\min(v_{34}-v_{12},v_{14}-v_{32})$ & $\geq 1-t_{11} \geq 0$ \\
$v_{14}$ & $\lambda_{6}$ & $v_{14}-v_{32}-\min(v_{34}-v_{12},v_{14}-v_{32})$ & $\geq 0$ \\
$v_{21}$ & $1-\lambda_{1}-\lambda_{3}-\lambda_{8}$ & $v_{21}-v_{11}-v_{12}$ & $g_{44}-t_{11} \geq 0$ \\
$v_{22}$ & $\lambda_{1}+\lambda_{3}$ & $v_{22}$ & $v_{22} \geq 0$ \\
$v_{23}$ & $\lambda_{8}-\lambda_{5}$ & 0 & 0 \\
$v_{24}$ & $\lambda_{5}$ & $v_{24}-v_{31}-v_{32}-v_{44}$ & $f_{1}-t_{11} \geq0$ \\
$v_{31}$ & $1-\lambda_{2}-\lambda_{6}-\lambda_{7}$ & 0 & 0 \\
$v_{32}$ & $\lambda_{2}+\lambda_{6}$ & 0 & 0 \\
$v_{33}$ & $\lambda_{7}-\lambda_{4}$ & $v_{33}-v_{11}-v_{23}-v_{44}+\min(v_{34}-v_{12},v_{14}-v_{32})$ & $\geq 1-t_{11} \geq 0$ \\
$v_{34}$ & $\lambda_{4}$ & $v_{34}-v_{12}-\min(v_{34}-v_{12},v_{14}-v_{32})$ & $\geq 0$ \\
$v_{41}$ & $1-\lambda_{1}-\lambda_{5}-\lambda_{7}$ & $v_{41}$ & $v_{41} \geq 0$ \\
$v_{42}$ & $\lambda_{1}+\lambda_{5}$ & $v_{42}-v_{31}-v_{32}$ & $g_{23}-t_{11} \geq 0$ \\
$v_{43}$ & $\lambda_{7}-\lambda_{3}$ & $v_{43}-v_{11}-v_{12}-v_{23}$ & $f_{3}-t_{11} \geq0$ \\
$v_{44}$ & $\lambda_{3}$ & 0 & 0 \\
\hline
\end{tabular}
\caption{Demonstration of the conditions for a 6 element closeness measure from the first class, here $t_{11}$. The $\lambda_{i}$ values come from Table \ref{tab:tvectorslambda}.}
\label{tab:tvectors}
\end{center}
\end{table}

In Table \ref{tab:tvectors} the measures for comparison with the benchmark $t_{11}$ are as follows: \\
$v_{11}+v_{12}+v_{13}+v_{32}-L$ \\
$v_{11}+v_{12}+v_{14}+v_{23}+v_{31}+v_{44}+L$ \\
$v_{21}+v_{23}+v_{31}+v_{32}+v_{44}$ \\
$v_{11}+v_{12}+v_{23}+v_{24}$ \\
$v_{12}+v_{31}+v_{32}+v_{33}-L$ \\
$v_{11}+v_{23}+v_{31}+v_{32}+v_{34}+v_{44}+L$ \\
$v_{11}+v_{12}+v_{23}+v_{42}+v_{44}$ \\
$v_{31}+v_{32}+v_{43}+v_{44}$. \\
Here, $L$ represents $\lambda_{2}p=\lambda_{2} t_{11}$ and is the only free term after the allocation of the parameters, as described in Table \ref{tab:tvectorslambda}. As we see in Table \ref{tab:tvectors}, the selected value of $L=-\min(v_{34}-v_{12},v_{14}-v_{32})$ achieves all of the inequalities (and in fact in some cases it is the unique term that does this for the $v_{13}$ and $v_{33}$ cases). 

There are four conditions that we need to check which relate to nudging the vector back in the correct space; those related to  $v_{13}, v_{14}, v_{33}, v_{34}$ (we need to see that the target deviations from Table \ref{tab:tvectors2} are less than for the original $v_{ij}$ vector element; the other conditions are automatically satisfied in this regard). This is easy to see by considering the appropriate $v_{ij}+t_{11}$ minus this deviation, when we obtain a closeness measure larger than $t_{11}$ in each case.

From $T_{2}$ we select $t_{21}$. Following the same process, we need the conditions from (\ref{eq:thesixt2lambda}) satisfied with equality. Again, there are five distinct equations with eight variables, together with the condition (\ref{eq:lambdapineqt2}) to be satisfied with equality, in turn implying that $\lambda_{1}+\lambda_{3}=0$ and $\lambda_{7}-\lambda_{3}=0$. This yields a set of $\lambda_{i}$ values with only one free parameter, as shown in Table \ref{tab:tvectorslambda2}. In this table, we leave this as a general value, but in the subsequent solution we choose $L=\lambda_{2}p=\max(v_{23}+v_{33}+v_{34}+v_{42}-v_{13}, v_{11}+v_{23}+v_{34}+v_{42}-v_{31})$. 

\begin{table}[hbt!]
\begin{center}
\begin{tabular}{||c|c||}
\hline
$\lambda_{1}$ & $-(v_{33}+v_{34})/p$ \\
$\lambda_{2}$ & $(v_{12}-v_{34})/p$ \\
$\lambda_{3}$ & $(v_{33}+v_{34})/p$ \\
$\lambda_{4}$ & $v_{34}/p$ \\
$\lambda_{5}$ & $v_{33}+v_{34}+v_{42})/p$ \\
$\lambda_{6}$ & free \\
$\lambda_{7}$ & $(v_{33}+v_{34})/p$ \\
$\lambda_{8}$ & $(v_{23}+v_{33}+v_{34}+v_{42})/p$ \\
\hline
\end{tabular}
\caption{The values of $\lambda_{i}$ when the minimum occurs at $t_{21}$.}
\label{tab:tvectorslambda2}
\end{center}
\end{table}

\begin{table}[hbt!]
\begin{center}
\begin{tabular}{||c|c|c|c||}
\hline
Vt. & Local param. $u_{ij}$ & Target deviation & Comparison \\
\hline
$v_{11}$ & $1-\lambda_{2}-\lambda_{4}-\lambda_{8}$ & 0 & 0 \\
$v_{12}$ & $\lambda_{2}+\lambda_{4}$ & 0 & 0 \\
$v_{13}$ & $\lambda_{8}-\lambda_{6}$ & $\max(0, v_{11}+v_{13}-v_{31}-v_{33})$ & $\max (1-f_{6},0)$ \\
$v_{14}$ & $\lambda_{6}$ & $v_{14}-v_{23}-v_{34}-v_{42}-\max(v_{33}-v_{13}, v_{11}-v_{31})$ & $\min(1,f_{6})-t_{21}$ \\
$v_{21}$ & $1-\lambda_{1}-\lambda_{3}-\lambda_{8}$ & $v_{21}-v_{11}-v_{12}$ & $g_{42}-t_{21} \geq 0$ \\
$v_{22}$ & $\lambda_{1}+\lambda_{3}$ & $v_{22}$ & $v_{22} \geq 0$ \\
$v_{23}$ & $\lambda_{8}-\lambda_{5}$ & 0 & 0 \\
$v_{24}$ & $\lambda_{5}$ & $v_{24}-v_{33}-v_{34}-v_{42}$ & $f_{1}-t_{21} \geq 0$ \\
$v_{31}$ & $1-\lambda_{2}-\lambda_{6}-\lambda_{7}$ & $\max(0, v_{31}+v_{33}-v_{11}-v_{13})$ &  $\max (1-f_{5},0)$ \\
$v_{32}$ & $\lambda_{2}+\lambda_{6}$ & $v_{32}-v_{12}-v_{23}-v_{42}-\max(v_{33}-v_{13}, v_{11}-v_{31})$ & $\min (1,f_{5})-t_{21}$ \\
$v_{33}$ & $\lambda_{7}-\lambda_{4}$ & 0 & 0 \\
$v_{34}$ & $\lambda_{4}$ & 0 & 0 \\
$v_{41}$ & $1-\lambda_{1}-\lambda_{5}-\lambda_{7}$ & $v_{41}-v_{11}-v_{12}-v_{23}$ & $f_{4}-t_{21} \geq 0$ \\
$v_{42}$ & $\lambda_{1}+\lambda_{5}$ & 0 & 0 \\
$v_{43}$ & $\lambda_{7}-\lambda_{3}$ & $v_{43}$ & $v_{43} \geq 0$ \\
$v_{44}$ & $\lambda_{3}$ & $v_{44}-v_{33}-v_{34}$ & $g_{23}-t_{21} \geq 0$ \\
\hline
\end{tabular}
\caption{Demonstration of the conditions for a 6 element closeness measure from the second class, here $t_{21}$. The $\lambda_{i}$ values come from Table \ref{tab:tvectorslambda2}.}
\label{tab:tvectors2}
\end{center}
\end{table}

In Table \ref{tab:tvectors2} the measures for comparison with the benchmark $t_{21}$ are as follows: \\
$v_{11}+v_{12}+v_{13}+L$ \\
$v_{11}+v_{12}+v_{14}+v_{23}+v_{33}+v_{34}+v_{42}-L$ \\
$v_{21}+v_{23}+v_{33}+v_{34}+v_{42}$ \\
$v_{11}+v_{12}+v_{23}+v_{24}$ \\
$v_{12}+v_{31}+v_{33}+L$ \\
$v_{11}+v_{23}+v_{32}+v_{33}+2v_{34}+v_{42}-L$ \\
$v_{33}+v_{34}+v_{41}+v_{42}$ \\
$v_{11}+v_{12}+v_{23}+v_{42}+v_{44}$ \\
Here $L$ represents $\lambda_{2}p=\lambda_{2} t_{21}$ and is the only free term after the allocation of the parameters as described in Table \ref{tab:tvectorslambda2}. As we see in Table \ref{tab:tvectors2}, the selected value of $L=\max(v_{23}+v_{33}+v_{34}+v_{42}-v_{13}, v_{11}+v_{23}+v_{34}+v_{42}-v_{31})$ achieves all of the inequalities (and in fact is the unique term that does this for the $v_{13}$ and $v_{33}$ terms). 

As before, we can see that the selected vector is within the correct space, $\mathcal{P}_{\scriptscriptstyle NS}$.


(iv) Now consider the $g_{ij}$ measures, choosing $g_{11}$. For this measure, we need the conditions from (\ref{eq:thefivetlambda}) satisfied with equality. Since here $p=g_{11}$, there are only four distinct equations with eight variables. However, we also need the condition (\ref{eq:lambdaping}) to be satisfied with equality, implying that $\lambda_{6}=0$. This yields a set of $\lambda_{i}$ values with three free parameters, as shown in Table \ref{tab:gvectorslambda}. In this table, we leave these as general values, but in the subsequent solution we choose several cases: \\
$f_{3}<1, f_{7}<1$ then $\lambda_{1}=(v_{42}-v_{24})/p, \lambda_{3}=v_{44}/p, \lambda_{7}=(v_{43}+v_{44})/p$. \\
$f_{3}<1, f_{7}>1$ then $\lambda_{1}=(v_{42}-v_{24})/p+(1-f_{7}), \lambda_{3}=v_{44}/p, \lambda_{7}=(v_{43}+v_{44})/p$. \\
$f_{3}>1, f_{7}<1$ then $\lambda_{1}=(v_{42}-v_{24})/p, \lambda_{3}=v_{44}/p, \lambda_{7}=(v_{43}+v_{44})/p+1-f_{3}$. \\
$f_{3}>f_{7}>1$ then $\lambda_{1}=(v_{42}-v_{24})/p, \lambda_{3}=v_{44}/p+1-f_{7}, \lambda_{7}=(v_{43}+v_{44})/p+1-f_{3}$. \\
$f_{7}>f_{3}>1$ then $\lambda_{1}=(v_{42}-v_{24})/p, \lambda_{3}=v_{44}/p+1-f_{7}, \lambda_{7}=(v_{43}+v_{44})/p+1-f_{7}$. \\

\begin{table}[hbt!]
\begin{center}
\begin{tabular}{||c|c||}
\hline
$\lambda_{1}$ & free \\
$\lambda_{2}$ & $v_{32}/p$ \\
$\lambda_{3}$ & free \\
$\lambda_{4}$ & $v_{34}/p$ \\
$\lambda_{5}$ & $v_{24}/p$ \\
$\lambda_{6}$ & 0\\
$\lambda_{7}$ & free \\
$\lambda_{8}$ & $(v_{23}+v_{24})/p$ \\
\hline
\end{tabular}
\caption{The values of $\lambda_{i}$ when the minimum occurs at $g_{11}$.}
\label{tab:gvectorslambda}
\end{center}
\end{table}

\begin{table}[hbt!]
\begin{center}
\begin{tabular}{||c|c|c|c||}
\hline
Vector & Local param. $u_{ij}$ & Target deviation & Comparison \\
\hline
$v_{11}$ & $1-\lambda_{2}-\lambda_{4}-\lambda_{8}$ & 0 & 0 \\
$v_{12}$ & $\lambda_{2}+\lambda_{4}$ & $v_{12}-v_{32}-v_{34}$ & $f_{1}-g_{11}\geq 0$ \\
$v_{13}$ & $\lambda_{8}-\lambda_{6}$ & $v_{13}-v_{23}-v_{24}$ & $f_{3}-g_{11} \geq 0$ \\
$v_{14}$ & $\lambda_{6}$ & $v_{14}$ & $v_{14} \geq 0$ \\
$v_{21}$ & $1-\lambda_{1}-\lambda_{3}-\lambda_{8}$ & $\min(1,f_{7})-g_{11}$ &  $\geq 0$ \\
$v_{22}$ & $\lambda_{1}+\lambda_{3}$ & $\max(v_{22}+v_{24}-v_{42}-v_{44},0)$ & $\max(1-f_{7},0) \geq 0$ \\
$v_{23}$ & $\lambda_{8}-\lambda_{5}$ & 0 & 0 \\
$v_{24}$ & $\lambda_{5}$ & 0 & 0 \\
$v_{31}$ & $1-\lambda_{2}-\lambda_{6}-\lambda_{7}$ & $f_{3}/1/(1+f_{3}-f_{7})-g_{11}$ & $\geq 0$ \\
$v_{32}$ & $\lambda_{2}+\lambda_{6}$ & 0 & 0 \\
$v_{33}$ & $\lambda_{7}-\lambda_{4}$ & $(1-f_{3})/0/(f_{7}-f_{3})$ & $\geq 0$ \\
$v_{34}$ & $\lambda_{4}$ & 0 & 0 \\
$v_{41}$ & $1-\lambda_{1}-\lambda_{5}-\lambda_{7}$ & $f_{4}/f_{8}/1-g_{11}$  & $\geq 0$ \\
$v_{42}$ & $\lambda_{1}+\lambda_{5}$ & $\max(v_{42}+v_{44}-v_{22}-v_{24},0)$ & $\max(f_{7}-1, 0) \geq 0$ \\
$v_{43}$ & $\lambda_{7}-\lambda_{3}$ & $(f_{3}-1)/(f_{3}-f_{7})/0$ & $\geq 0$ \\
$v_{44}$ & $\lambda_{3}$ & $\max(v_{42}+v_{44}-v_{22}-v_{24},0)$ &  $\max(f_{7}-1, 0) \geq 0$ \\
\hline
\end{tabular}
\caption{Demonstration of the conditions for a 5 element closeness measure, here $g_{11}$. The $\lambda_{i}$ values come from Table \ref{tab:gvectorslambda}. }
\label{tab:gvectors}
\end{center}
\end{table}

In Table \ref{tab:gvectors} the measures for comparison with the benchmark $g_{11}$ are as follows: \\
$v_{21}+v_{23}+v_{24}+L_{1}+L_{3}$ \\
$v_{11}+v_{22}+v_{23}+v_{24}+v_{32}+v_{34}-L_{1}-L_{3}$ \\
$v_{31}+v_{32}+L_{7}$ \\
$v_{11}+v_{23}+v_{24}+v_{32}+v_{33}+2v_{34}-L_{7}$ \\
$v_{24}+v_{41}+L_{1}+L_{7}$ \\
$v_{11}+v_{23}+v_{32}+v_{34}+v_{42}-L_{1}$ \\
$v_{11}+v_{23}+v_{24}+v_{32}+v_{34}+v_{43}+L_{3}-L_{7}$ \\
$v_{31}+v_{23}+v_{24}+v_{32}+v_{34}+v_{44}-L_{3}$. \\
Here $L_{i}$ represents $\lambda_{i}p=\lambda_{i} g_{11}$ for $i=1,3,7$. These are the free terms after the allocation of the parameters as described in Table \ref{tab:gvectorslambda}. The values must be selected appropriately, depending upon other parameter values, to ensure the conditions are satisfied. As we see in Table \ref{tab:gvectors}, the selected values $L_{1}, L_{2}, L_{3}$ achieve all of the inequalities.
We note that, as before, the selected point might be outside of the allowed space and a similar nudging procedure as previously described would be required.

\vspace{0.3cm}
\noindent\textbf{\textsf{\text{C.6. {Finding the measure of locality}}}}
\vspace{0.2cm}

The non-signaling space is characterized by the equality constraints (\ref{eq:nonsignaling}) as well as the constraints related to the laws of probability (\ref{eq:ugreat}). The local space involves the additional inequality constraints (\ref{eq:parametrizationlocal}). Here are then eight different classes, $F, E, G, S_{1}, S_{2}, S_{3}, T_{1}, T_{2}$.

\vspace{0.3cm}
\setcounter{theorem}{1}
\begin{theorem}\label{T2}\ \\
    $p({\bf v}, \mathcal{P}_{\scriptscriptstyle Loc})$ = $\min \{f_{1}, \ldots , f_{8}, e_{1}, \ldots, e_{8}, g_{1}, \ldots g_{16},  t_{1}, \ldots t_{32}, s_{1}, \ldots, s_{64} \}$.
\end{theorem}

As for \textbf{Theorem~\ref{T3}}, a proof of this theorem is elsewhere. Below we give a construction for finding the appropriate point in the non-signaling space that achieves the best closeness, itemized by the different measures concerning where this is possible. 
  
First, we give a way to show that $p({\bf v}, \mathcal{P}_{\scriptscriptstyle Loc})$ is bounded above by all 128 closeness measures.  In each case, we demonstrate the proof just for a single case (or two cases for $e_{i}$) but these generalize through the simple renumbering mentioned earlier.

(i) $p({\bf v},  \mathcal{P}_{\scriptscriptstyle Loc})$ $\leq f_{i}$. This is trivially true, since $p({\bf v},  \mathcal{P}_{\scriptscriptstyle NS})$ $\leq f_{i}$ and $\mathcal{P}_{\scriptscriptstyle Loc}$ is a subspace of $ \mathcal{P}_{\scriptscriptstyle NS}$.

(ii) $p({\bf v},  \mathcal{P}_{\scriptscriptstyle Loc})$ $\leq e_{i}$. We do the calculation for both $e_{1}$ and $e_{2}$.
We know from (\ref{eq:plessvoveru}) that for the closeness $p$ to any element of $U_{L}$, we have:
\begin{eqnarray}\label{eq:theeightlambda}
1-\lambda_{2}-\lambda_{4}-\lambda_{8} \leq \frac{v_{11}}{p}, \lambda_{6} \leq \frac{v_{14}}{p}, \lambda_{1}+\lambda_{3} \leq \frac{v_{22}}{p}, \lambda_{8}-\lambda_{5} \leq \frac{v_{23}}{p}, \nonumber \\
\lambda_{2}+\lambda_{6} \leq \frac{v_{32}}{p}, \lambda_{7}-\lambda_{4} \leq \frac{v_{33}}{p}, \lambda_{1}+\lambda_{5} \leq \frac{v_{42}}{p}, \lambda_{7}-\lambda_{3} \leq \frac{v_{43}}{p}.
\end{eqnarray}
Summing all the terms on both sides of the inequalities gives:
\begin{equation}\label{eq:lambdasumsine}
1+ 2(\lambda_{1}+\lambda_{6}+\lambda_{7}-\lambda_{4}) \leq \frac{e_{1}}{p}.
\end{equation}
We note that from (\ref{eq:parametrizationlocal}) we have that the bracketed term in (\ref{eq:lambdasumsine}) is greater than 0. Thus 
\begin{equation}
1 \leq \frac{e_{1}}{p} \Rightarrow p \leq e_{1}.
\end{equation}
Repeating the calculation for $e_{2}$ gives:
\begin{eqnarray}\label{eq:theeightlambda2}
\lambda_{2}+\lambda_{4} \leq \frac{v_{12}}{p}, \lambda_{8}-\lambda_{6} \leq \frac{v_{13}}{p}, 1-\lambda_{1}-\lambda_{3}-\lambda_{8} \leq \frac{v_{21}}{p}, \lambda_{5} \leq \frac{v_{24}}{p}, \nonumber \\
1-\lambda_{2}-\lambda_{6}-\lambda_{7} \leq \frac{v_{31}}{p}, \lambda_{4} \leq \frac{v_{34}}{p}, 1-\lambda_{1}-\lambda_{5}-\lambda_{7} \leq \frac{v_{41}}{p}, \lambda_{3} \leq \frac{v_{44}}{p}.
\end{eqnarray}
Summing all the terms on both sides of the inequalities gives:
\begin{equation}\label{eq:lambdasumsineagain}
3- 2(\lambda_{1}+\lambda_{6}+\lambda_{7}-\lambda_{4}) \leq \frac{e_{2}}{p}.
\end{equation}
We note that from (\ref{eq:parametrizationlocal}) we have that $\lambda_{1}+\lambda_{6}+\lambda_{7}-\lambda_{4}$ is less than 1. Thus 
\begin{equation}\label{eq:strictforeight2}
1 \leq \frac{e_{2}}{p} \Rightarrow p \leq e_{2}.
\end{equation}
Calculations for the remaining six terms are similar, with each of the eight inequalities summarized in (\ref{eq:parametrizationlocal}) appearing exactly once.

(iii) (i) $p({\bf v},  \mathcal{P}_{\scriptscriptstyle Loc})$ $\leq g_{ij}$. This is trivially true, since $p({\bf v},  \mathcal{P}_{\scriptscriptstyle NS})$ $\leq g_{ij}$ and $\mathcal{P}_{\scriptscriptstyle Loc}$ is a subspace of $ \mathcal{P}_{\scriptscriptstyle NS}$.

(iv) (i) $p({\bf v},  \mathcal{P}_{\scriptscriptstyle Loc})$ $\leq t_{ji}$. This is trivially true, since $p({\bf v}, U_{NS})$ $\leq t_{ij}$ and $\mathcal{P}_{\scriptscriptstyle Loc}$ is a subspace of $ \mathcal{P}_{\scriptscriptstyle NS}$.

(v) $p({\bf v},  \mathcal{P}_{\scriptscriptstyle Loc})$ $\leq s_{ijk}$. For case $S_{1}$ we do the calculation for $s_{111}=v_{11}+v_{12}+v_{13}+v_{23}+v_{32}+v_{44}$.
We know from (\ref{eq:plessvoveru}) that:
\begin{eqnarray}\label{eq:thesixslambda}
1-\lambda_{2}-\lambda_{4}-\lambda_{8} \leq \frac{v_{11}}{p}, \lambda_{2}+\lambda_{4} \leq \frac{v_{12}}{p}, 
\lambda_{8}-\lambda_{6} \leq \frac{v_{13}}{p}, \nonumber \\
\lambda_{8}-\lambda_{5} \leq \frac{v_{23}}{p}, \lambda_{2}+\lambda_{6} \leq \frac{v_{32}}{p},  \lambda_{3} \leq \frac{v_{44}}{p}.
\end{eqnarray}
Summing all the terms on both sides of the inequalities gives:
\begin{equation}\label{eq:lambdapineqs}
1+(\lambda_{2}+\lambda_{3}+\lambda_{8}-\lambda_{5}) \leq \frac{s_{111}}{p}.
\end{equation}
We know from (\ref{eq:parametrizationlocal}) that the bracketed term in (\ref{eq:lambdapineqs}) is non-negative. Thus 
\begin{equation}
1 \leq \frac{s_{111}}{p} \Rightarrow p \leq s_{111}.
\end{equation}
Calculations from the remaining 15 terms are similar, with non-negative bracketed terms in each case. \\
For case $S_{2}$ we do the calculation for $s_{112}=v_{11}+v_{12}+v_{13}+v_{23}+v_{34}+v_{42}$.
We know from (\ref{eq:plessvoveru}) that:
\begin{eqnarray}\label{eq:thesixslambda2}
1-\lambda_{2}-\lambda_{4}-\lambda_{8} \leq \frac{v_{11}}{p}, \lambda_{2}+\lambda_{4} \leq \frac{v_{12}}{p}, 
\lambda_{8}-\lambda_{6} \leq \frac{v_{13}}{p}, \nonumber \\
\lambda_{8}-\lambda_{5} \leq \frac{v_{23}}{p}, \lambda_{4} \leq \frac{v_{34}}{p},  \lambda_{1}+\lambda_{5} \leq \frac{v_{42}}{p}.
\end{eqnarray}
Summing all the terms on both sides of the inequalities gives:
\begin{equation}\label{eq:lambdapineqs2}
1+(\lambda_{1}+\lambda_{4}+\lambda_{8}-\lambda_{6}) \leq \frac{s_{112}}{p}.
\end{equation}
We know from (\ref{eq:parametrizationlocal}) that the bracketed term in (\ref{eq:lambdapineqs2}) is non-negative. Thus 
\begin{equation}
1 \leq \frac{s_{112}}{p} \Rightarrow p \leq s_{112}.
\end{equation}
Calculations from the remaining 31 terms are similar. \\
For case $S_{3}$, we do the calculation for $s_{114}=v_{11}+v_{12}+v_{13}+v_{24}+v_{34}+v_{41}$.
We know from (\ref{eq:plessvoveru}) that:
\begin{eqnarray}\label{eq:thesixslambda3}
1-\lambda_{2}-\lambda_{4}-\lambda_{8} \leq \frac{v_{11}}{p}, \lambda_{2}+\lambda_{4} \leq \frac{v_{12}}{p}, 
\lambda_{8}-\lambda_{6} \leq \frac{v_{13}}{p}, \nonumber \\
\lambda_{5} \leq \frac{v_{24}}{p}, \lambda_{4} \leq \frac{v_{34}}{p}, 1-\lambda_{1}-\lambda_{5}-\lambda_{7} \leq \frac{v_{41}}{p}.
\end{eqnarray}
Summing all the terms on both sides of the inequalities gives:
\begin{equation}\label{eq:lambdapineqs3}
1+(1-\lambda_{1}-\lambda_{6}-\lambda_{7}+\lambda_{4}) \leq \frac{s_{114}}{p}.
\end{equation}
We know from (\ref{eq:parametrizationlocal}) that the bracketed term in (\ref{eq:lambdapineqs3}) is non-negative. Thus 
\begin{equation}
1 \leq \frac{s_{114}}{p} \Rightarrow p \leq s_{114}.
\end{equation}
Calculations from the remaining 15 terms are similar.

b) 
We now demonstrate that the bound described above is tight, by showing how to achieve it when a measure from each of the five classes is the minimum. 

(i) We start with the $f_{i}$ measures. We again start with the construction from Lemma 1 and show the process for $f_{1}$ and when $f_{3}, f_{5}$ and $f_{7}$ are less than 1. Thus, following the proof from Lemma 1, we consider the local vector from Case A. We need to show that when $f_{1}$ is the smallest of all 128 closeness measures, so that $p=f_{1}$ (and also that $f_{3}, f_{5}, f_{7}$ are less than 1), we have that the inequalities in \ref{eq:parametrizationlocal} hold.

Using the $\lambda_{i}$ values from equation \ref{eq:fvectorlambdasfirst} and Table \ref{tab:CaseAtheorem 1}, we have:
\begin{eqnarray}\label{eq:whenf1smallest128}
\lambda_{1}+\lambda_{4}+\lambda_{8}-\lambda_{6}=(v_{42}+v_{12}+v_{23}-v_{32})/p=(s_{331}-1)/f_{1} , \nonumber \\
\lambda_{2}+\lambda_{3}+\lambda_{8}-\lambda_{5}=(v_{12}-v_{34}+v_{44}+v_{23})/p=(s_{311}-1)/f_{1} , \nonumber \\
\lambda_{1}+\lambda_{6}+\lambda_{7}-\lambda_{4}=(v_{42}-v_{24}+v_{32}-v_{12}+v_{43}+v_{44})/p=(s_{441}-f_{1})/f_{1} , \nonumber \\
\lambda_{2}+\lambda_{5}+\lambda_{7}-\lambda_{3} =(v_{12}-v_{34}+v_{24}+v_{43})/p=(s_{312}-1)/f_{1} .
\end{eqnarray}

We need to see that all expressions above are greater than 0 and less than 1. For this, we need to see that each of the first, second, and fourth $s_{ijk}$ values are more than $1$ and less than $1+f_{1}$, and the third $s_{ijk}$ value is between $f_{1}$ and $2f_{1}$. We thus have eight conditions to satisfy. \\
As before, the solutions may be just outside the correct space, in which case we need to ``nudge'' them back. From the first equation, if $s_{331}<1$ we can increase this by decreasing $\lambda_{6}$, potentially down to 0, while leaving all other $\lambda_{i}$ values unchanged. We obtain at the maximum limit \\
\begin{equation}
\lambda_{1}+\lambda_{4}+\lambda_{8}-\lambda_{6}=(v_{23}+v_{34}+v_{42})/p,
\end{equation}
which is clearly greater than 0, so we just need to stop at some appropriate value. This does not violate any of the other conditions, except potentially the lower bound in the third inequality, since $\lambda_{6}$ works in the opposite direction here. We note that these two terms add to
\begin{equation}
(v_{11}+v_{12}+2v_{23}+2v_{42}+v_{43}+v_{44} -f_{1})/f_{1}=(g_{23}-f_{1}+v_{23}+v_{42}+v_{43})/f_{1}.
\end{equation}
Since this sum is positive, we can achieve both expressions being greater than zero simultaneously. A similar procedure applies to each of the other conditions. 

(ii) Now consider the measures $e_{i}$, in particular $e_{2}$. For this measure, we need the conditions from (\ref{eq:theeightlambda2}) satisfied with equality. Since here $p=e_{2}$ there are only seven distinct equations with eight variables. However, we also need the condition (\ref{eq:strictforeight2}) satisfied with equality, which implies that $\lambda_{1}+\lambda_{6}+\lambda_{7}-\lambda_{4}=1$. This yields a unique set of $\lambda_{i}$ values, as shown in Table \ref{tab:evectorslambda}. We can clearly see from the table that the target deviations lie between 0 and the respective $v_{ij}$, so corresponding to a vector within the local space.

\begin{table}[hbt!]
\begin{center}
\begin{tabular}{||c|c||}
\hline
$\lambda_{1}$ & $(v_{12}+v_{31})/p$ \\
$\lambda_{2}$ & $(v_{12}-v_{34})/p$ \\
$\lambda_{3}$ & $v_{44}/p$ \\
$\lambda_{4}$ & $v_{34}/p$ \\
$\lambda_{5}$ & $v_{24}/p$ \\
$\lambda_{6}$ & $(v_{24}+v_{34}+v_{41})/p$ \\
$\lambda_{7}$ & $(v_{13}+v_{21}+v_{34}+v_{44})/p$ \\
$\lambda_{8}$ & $(v_{13}+v_{24}+v_{34}+v_{41})/p$ \\
\hline
\end{tabular}
\caption{The values of $\lambda_{i}$ when the minimum occurs at $e_{2}$.}
\label{tab:evectorslambda}
\end{center}
\end{table}

\begin{table}[hbt!]
\begin{center}
\begin{tabular}{||c|c|c|c||}
\hline
Vector $v_{ij}$ & Local parametrization $u_{ij}$ & Target deviation & Comparison \\
\hline
$v_{11}$ & $1-\lambda_{2}-\lambda_{4}-\lambda_{8}$ & $v_{11}-v_{21}-v_{31}-v_{44}$ & $s_{114}-e_{2} \geq 0$ \\
$v_{12}$ & $\lambda_{2}+\lambda_{4}$ & 0 & 0 \\
$v_{13}$ & $\lambda_{8}-\lambda_{6}$ & 0 & 0 \\
$v_{14}$ & $\lambda_{6}$ & $v_{14}-v_{24}-v_{34}-v_{41}$ & $s_{141}-e_{2} \geq 0$ \\
$v_{21}$ & $1-\lambda_{1}-\lambda_{3}-\lambda_{8}$ & 0 & 0 \\
$v_{22}$ & $\lambda_{1}+\lambda_{3}$ & $v_{22}-v_{12}-v_{31}-v_{44}$ & $s_{221}-e_{2} \geq 0$ \\
$v_{23}$ & $\lambda_{8}-\lambda_{5}$ & $v_{23}-v_{13}-v_{34}-v_{41}$ & $s_{234}-e_{2} \geq 0$ \\
$v_{24}$ & $\lambda_{5}$ & 0 & 0 \\
$v_{31}$ & $1-\lambda_{2}-\lambda_{6}-\lambda_{7}$ & 0 & 0 \\
$v_{32}$ & $\lambda_{2}+\lambda_{6}$ & $v_{32}-v_{12}-v_{24}-v_{41}$ & $s_{323}-e_{2} \geq 0$ \\
$v_{33}$ & $\lambda_{7}-\lambda_{4}$ & $v_{33}-v_{13}-v_{21}-v_{44}$ & $s_{331}-e_{2} \geq 0$ \\
$v_{34}$ & $\lambda_{4}$ & 0 & 0 \\
$v_{41}$ & $1-\lambda_{1}-\lambda_{5}-\lambda_{7}$ & 0 & 0 \\
$v_{42}$ & $\lambda_{1}+\lambda_{5}$ & $v_{42}-v_{12}-v_{24}-v_{31}$ & $s_{423}-e_{2} \geq 0$ \\
$v_{43}$ & $\lambda_{7}-\lambda_{3}$ & $v_{43}-v_{13}-v_{21}-v_{34}$ & $s_{432}-e_{2} \geq 0$ \\
$v_{44}$ & $\lambda_{3}$ & 0 & 0 \\
\hline
\end{tabular}
\caption{Demonstration of the conditions for an 8 element closeness measure, here $e_{2}$. The $\lambda_{i}$ values come from Table \ref{tab:evectorslambda}.}
\label{tab:evectors}
\end{center}
\end{table}

In Table \ref{tab:evectors} the eight $s_{ijk}$ measures are as follows: \\
$s_{114}=v_{11}+v_{12}+v_{13}+v_{24}+v_{34}+v_{41}$ \\
$s_{141}=v_{12}+v_{13}+v_{14}+v_{21}+v_{31}+v_{44}$ \\
$s_{221}=v_{13}+v_{21}+v_{22}+v_{24}+v_{34}+v_{41}$ \\
$s_{234}=v_{12}+v_{21}+v_{23}+v_{24}+v_{31}+v_{44}$ \\
$s_{323}=v_{13}+v_{21}+v_{31}+v_{32}+v_{34}+v_{44}$ \\
$s_{331}=v_{12}+v_{24}+v_{31}+v_{33}+v_{34}+v_{41}$ \\
$s_{423}=v_{13}+v_{21}+v_{34}+v_{41}+v_{42}+v_{44}$ \\
$s_{432}=v_{12}+v_{24}+v_{31}+v_{41}+v_{43}+v_{44}$ \\
Note that the first two measures are in $S_{3}$, the next four are in $S_{2}$ and the last two are in $S_{1}$.

(iii) Now consider the $t_{ij}$ measures. We start with the construction from \textbf{Theorem~\ref{T3}}. We need to show that when $t_{11}$ is the smallest of all 128 measures, so that $p=t_{11}$, we have that the inequalities in (\ref{eq:parametrizationlocal}) hold.

Using the $\lambda_{i}$ values from equation (\ref{tab:tvectorslambda}), we have:
\begin{eqnarray}\label{eq:whent11smallest128}
\lambda_{1}+\lambda_{4}+\lambda_{8}-\lambda_{6}= \frac{1}{t_{11}} (v_{12}+v_{23}+v_{31}), \nonumber \\
\lambda_{2}+\lambda_{3}+\lambda_{8}-\lambda_{5}= \frac{1}{t_{11}} (v_{23}+v_{44}-\min(v_{34}-v_{12}, v_{14}-v_{32})), \nonumber \\
\lambda_{1}+\lambda_{6}+\lambda_{7}-\lambda_{4} = \frac{1}{t_{11}} (v_{11}+v_{23}+v_{32}), \nonumber \\
\lambda_{2}+\lambda_{5}+\lambda_{7}-\lambda_{3} = \frac{1}{t_{11}} (v_{11}+v_{12}+v_{23}+v_{31}+v_{32}+v_{44}-\min(v_{34}-v_{12}, v_{14}-v_{32})).
\end{eqnarray}
The first and third expressions clearly lie between 0 and 1. The second expression is clearly smaller than the fourth, so we need to show that the second is greater than 0 and the fourth is less than 1. Now consider the fourth expression. Since $s_{321}-t_{11}=v_{34}-v_{12}$ and $s_{121}-t_{11}=v_{14}-v_{32}$, both of these expressions are non-negative, so that $\min(v_{34}-v_{12}, v_{14}-v_{32})$ is non-negative and the second expression does not exceed 1 as required. For the second expression, it is possible that the selected value of $L$ gives a value for this that is negative, if $-(v_{23}+v_{44})>L=-\min(v_{34}-v_{12}, v_{14}-v_{32}$. We need to nudge this into the region by increasing $L$ to equal $-(v_{23}+v_{44})$ (here we can see that all other conditions still hold).

A similar argument holds for $t_{12}$.

(iv) Now consider the $g_{ij}$ measures, choosing $g_{11}$. We start with the construction from \textbf{Theorem~\ref{T3}} from Table \ref{tab:gvectors}. We would need to show that when $g_{11}$ is the smallest of all 128 measures, so that $p=g_{11}$, we have that the inequalities in \ref{eq:parametrizationlocal} hold. As above, this again might not be the case without suitable nudging of the solution.

(v) Now consider the $s_{ijk}$ measures, starting with $S_{1}$ and choosing $s_{111}$. For this measure, we need the conditions from (\ref{eq:thesixslambda}) satisfied with equality. We also need the condition (\ref{eq:lambdapineqs}) to be satisfied with equality, in turn implying that $\lambda_{2}+\lambda_{3}+\lambda_{8}-\lambda_{5}=0$. For the vectors in $S_{1}$ (and also $S_{2}, S_{3}$), as above, some nudging may be needed.

Since here $p=s_{111}$ there are six distinct equations with eight variables. This yields a set of $\lambda_{i}$ values with only two free parameters, as shown in Table \ref{tab:svectorslambda}. In this table, we leave these as general values, but in the subsequent solution we choose $\lambda_{1} p=\min(v_{22}-v_{44}, v_{42}-v_{13}-v_{32}-v_{44})$ and $\lambda_{7} p = \min(v_{43}+v_{44}, v_{12}+v_{23}+v_{33}+v_{44}).$

\begin{table}[hbt!]
\begin{center}
\begin{tabular}{||c|c||}
\hline
$\lambda_{1}$ & free \\
$\lambda_{2}$ & $-(v_{23}+v_{44})/p$ \\
$\lambda_{3}$ & $v_{44}/p$ \\
$\lambda_{4}$ & $(v_{12}+v_{23}+v_{44})/p$ \\
$\lambda_{5}$ & $(v_{13}+v_{32}+v_{44})/p$ \\
$\lambda_{6}$ & $(v_{23}+v_{32}+v_{44})/p$\\
$\lambda_{7}$ & free \\
$\lambda_{8}$ & $(v_{13}+v_{23}+v_{32}+v_{44})/p$ \\
\hline
\end{tabular}
\caption{The values of $\lambda_{i}$ when the minimum occurs at $s_{111}$.}
\label{tab:svectorslambda}
\end{center}
\end{table}

\begin{table}[hbt!]
\begin{center}
\begin{tabular}{||c|c|c|c||}
\hline
Vector & Local param. $u_{ij}$ & Target deviation & Comparison \\
\hline
$v_{11}$ & $1-\lambda_{2}-\lambda_{4}-\lambda_{8}$ & 0 & 0 \\
$v_{12}$ & $\lambda_{2}+\lambda_{4}$ & 0 & 0 \\
$v_{13}$ & $\lambda_{8}-\lambda_{6}$ & 0 & 0 \\
$v_{14}$ & $\lambda_{6}$ & $v_{14}-v_{23}-v_{32}-v_{44}$ & $1-s_{111} \geq 0$ \\
$v_{21}$ & $1-\lambda_{1}-\lambda_{3}-\lambda_{8}$ & $v_{21}-v_{11}-v_{12}+\min(v_{22}, v_{42}-v_{13}-v_{32})$ & $\min(s_{211},f_{7})-s_{111} \geq 0$ \\
$v_{22}$ & $\lambda_{1}+\lambda_{3}$ & $\max(0,v_{13}+v_{22}+v_{32}-v_{42})$  & $\geq 0$ \\
$v_{23}$ & $\lambda_{8}-\lambda_{5}$ & 0 & 0 \\
$v_{24}$ & $\lambda_{5}$ & $v_{24}-v_{13}-v_{32}-v_{44}$ & $f_{1}-s_{111} \geq 0$ \\
$v_{31}$ & $1-\lambda_{2}-\lambda_{6}-\lambda_{7}$ & $v_{31}-v_{11}-v_{13}-\min(v_{33},v_{43}-v_{12}-v_{23})$ & $\min(f_{3},s_{311})-s_{111} \geq 0$ \\
$v_{32}$ & $\lambda_{2}+\lambda_{6}$ & 0 & 0 \\
$v_{33}$ & $\lambda_{7}-\lambda_{4}$ & $\max(v_{12}+v_{23}+v_{33}-v_{43},0)$ & $\geq 0$ \\
$v_{34}$ & $\lambda_{4}$ & $v_{34}-v_{12}-v_{23}-v_{44}$ & $f_{5}-s_{111} \geq 0$ \\
$v_{41}$ & $1-\lambda_{1}-\lambda_{5}-\lambda_{7}$ & $\min (s_{422}, e_{7}, 1, s_{433})-s_{111}$ & $\geq 0$ \\
$v_{42}$ & $\lambda_{1}+\lambda_{5}$ & $\max(v_{42}-v_{13}-v_{22}-v_{32},0)$ & $\geq 0$ \\
$v_{43}$ & $\lambda_{7}-\lambda_{3}$ & $\max(0, v_{43}-v_{12}-v_{23}-v_{33})$ & $\geq 0$ \\
$v_{44}$ & $\lambda_{3}$ & 0 & 0 \\
\hline
\end{tabular}
\caption{Demonstration of the conditions for a 6 element closeness measure from the first class, here $s_{111}$. The $\lambda_{i}$ values come from Table \ref{tab:svectorslambda}.}
\label{tab:svectors}
\end{center}
\end{table}

In Table \ref{tab:svectors} the measures for comparison with the benchmark $s_{111}$ are as follows: \\
$1$ \\
$v_{13}+v_{21}+v_{23}+v_{32}+2v_{44}+L_{1}$ \\
$v_{11}+v_{12}+v_{13}+v_{22}+v_{23}+v_{32}-L_{1}$ \\
$v_{11}+v_{12}+v_{23}+v_{24}$ \\
$v_{31}+v_{32}+L_{7}$ \\
$v_{11}+2v_{12}+v_{13}+2v_{23}+v_{32}+v_{33}+2v_{44}-L_{7}$ \\
$v_{11}+v_{13}+v_{32}+v_{34}$ \\
$v_{13}+v_{32}+v_{41}+v_{44}+L_{1}+L_{7}$ \\
$v_{11}+v_{12}+v_{23}+v_{42}-L_{1}$ \\
$v_{11}+v_{12}+v_{13}+v_{23}+v_{32}+v_{43}+2v_{44}-L_{7}$ \\
Here $L_{1}=\min(v_{22}-v_{44}, v_{42}-v_{13}-v_{32}-v_{44})$ and $L_{7}=\min(v_{43}+v_{44}, v_{12}+v_{23}+v_{33}+v_{44})$ \\

Now for $S_{2}$ we choose $s_{112}$. For this measure, we need the conditions from (\ref{eq:thesixslambda2}) satisfied with equality. As before, we also have the condition (\ref{eq:lambdapineqs2}) to be satisfied with equality, implying that $\lambda_{1}+\lambda_{4}+\lambda_{8}-\lambda_{6}=0$, so that with for $p=s_{112}$ there are six distinct equations here with eight variables. This yields the set of $\lambda_{i}$ values shown in Table \ref{tab:svectorslambda2}. In this table, we leave these as general values, but in the subsequent solution we choose $\lambda_{3} p=\min(v_{44}, v_{13}+v_{22}+v_{34})$ and $\lambda_{7} p=v_{11}+v_{13}+v_{34}-\min(v_{31}, v_{41}-v_{12}-v_{23}).$

\begin{table}[hbt!]
\begin{center}
\begin{tabular}{||c|c||}
\hline
$\lambda_{1}$ & $-(v_{13}+v_{34})/p$ \\
$\lambda_{2}$ & $(v_{12}-v_{44})/p$ \\
$\lambda_{3}$ & free \\
$\lambda_{4}$ & $v_{34}/p$ \\
$\lambda_{5}$ & $(v_{13}+v_{34}+v_{42})/p$ \\
$\lambda_{6}$ & $(v_{23}+v_{34}+v_{42})/p$\\
$\lambda_{7}$ & free \\
$\lambda_{8}$ & $(v_{13}+v_{23}+v_{34}+v_{42})/p$ \\
\hline
\end{tabular}
\caption{The values of $\lambda_{i}$ when the minimum occurs at $s_{112}$.}
\label{tab:svectorslambda2}
\end{center}
\end{table}

\begin{table}[hbt!]
\begin{center}
\begin{tabular}{||c|c|c|c||}
\hline
Vector & Local param. $u_{ij}$ & Target deviation & Comparison \\
\hline
$v_{11}$ & $1-\lambda_{2}-\lambda_{4}-\lambda_{8}$ & 0 & 0 \\
$v_{12}$ & $\lambda_{2}+\lambda_{4}$ & 0 & 0 \\
$v_{13}$ & $\lambda_{8}-\lambda_{6}$ & 0 & 0 \\
$v_{14}$ & $\lambda_{6}$ & $v_{14}-v_{23}-v_{34}-v_{42}$ & $1-s_{112} \geq 0$ \\
$v_{21}$ & $1-\lambda_{1}-\lambda_{3}-\lambda_{8}$ & $
v_{21}-v_{11}-v_{12}+\min(v_{22}, v_{13}+v_{22}+v_{34})$ & $\min(f_{7},s_{212})-s_{112} \geq 0$ \\
$v_{22}$ & $\lambda_{1}+\lambda_{3}$ & $v_{13}+v_{22}+v_{34}-\min(v_{44}-v_{13}-v_{34}, v_{22})$  & $\geq 0$ \\
$v_{23}$ & $\lambda_{8}-\lambda_{5}$ & 0 & 0 \\
$v_{24}$ & $\lambda_{5}$ & $v_{24}-v_{13}-v_{34}-v_{42}$ & $f_{1}-s_{112} \geq 0$ \\
$v_{31}$ & $1-\lambda_{2}-\lambda_{6}-\lambda_{7}$ & $v_{31}-\min(v_{31}, v_{41}-v_{12}-v_{23})$ & $\geq 0$ \\
$v_{32}$ & $\lambda_{2}+\lambda_{6}$ & $v_{32}-v_{12}-v_{23}-v_{42}$ & $f_{5}-s_{112} \geq 0$ \\
$v_{33}$ & $\lambda_{7}-\lambda_{4}$ & $v_{33}-v_{11}-v_{13}+\min(v_{31}, v_{41}-v_{12}-v_{23})$ & $\min(s_{331},f_{4})-s_{112} \geq 0$ \\
$v_{34}$ & $\lambda_{4}$ & 0 & 0 \\
$v_{41}$ & $1-\lambda_{1}-\lambda_{5}-\lambda_{7}$ & $v_{41}-v_{12}-v_{23}-\min(v_{31}, v_{41}-v_{12}-v_{23})$ & $\geq 0$ \\
$v_{42}$ & $\lambda_{1}+\lambda_{5}$ & 0 & 0 \\
$v_{43}$ & $\lambda_{7}-\lambda_{3}$ & $\min(s{413},1,e_{5},s_{442})-s_{112}$ & $\geq 0$ \\
$v_{44}$ & $\lambda_{3}$ & $v_{44}-\min(v_{44}, v_{13}+v_{22}+v_{34})$ & $\geq 0$ \\
\hline
\end{tabular}
\caption{Demonstration of the conditions for a 6 element closeness measure from the second class, here $s_{112}$. The $\lambda_{i}$ values come from Table \ref{tab:svectorslambda2}.}
\label{tab:svectors2}
\end{center}
\end{table}

In Table \ref{tab:svectors2} the measures for comparison with the benchmark $s_{112}$ are as follows: \\
$1$ \\
$v_{21}+v_{23}+v_{42}+L_{3}$ \\
$v_{11}+v_{12}+2v_{13}+v_{22}+v_{23}+2v_{34}+v_{42}-L_{3}$ \\
$v_{11}+v_{12}+v_{23}+v_{24}$ \\
$v_{12}+v_{23}+v_{31}+v_{42}+L_{7}$ \\
$v_{11}+v_{13}+v_{32}+v_{34}$ \\
$v_{11}+v_{12}+v_{13}+v_{23}+v_{33}+2v_{34}+v_{42}-L_{7}$ \\
$v_{41}+v_{42}+L_{7}$ \\
$v_{11}+v_{12}+v_{13}+v_{23}+v_{34}+v_{42}+v_{43}-L_{7}+L_{3}$ \\
$v_{11}+v_{12}+v_{13}+v_{23}+v_{34}+v_{42}+v_{44}-L_{3}$ \\
Here $L_{3}=\min(v_{44}, v_{13}+v_{22}+v_{34})$ and $L_{7}=v_{11}+v_{13}+v_{34}-\min(v_{31}, v_{41}-v_{12}-v_{23})$. \\

Finally, regarding $S_{3}$ we choose as an example $s_{114}$. For this measure, we need the conditions from (\ref{eq:thesixslambda3}) satisfied with equality. As before, we also have the condition (\ref{eq:lambdapineqs3}) to be satisfied with equality, implying that $\lambda_{1}+\lambda_{6}+\lambda_{7}-\lambda_{4}=1$, so that for $p=s_{114}$ there are six distinct equations with eight variables. This yields the set of $\lambda_{i}$ values shown in Table \ref{tab:svectorslambda3}. In this table, we leave these as general values, but in the subsequent solution we choose $\lambda_{1} p=\min(v_{42}-v_{24}, v_{12}+v_{31})$ and $\lambda_{3} p=\max(v_{11}+v_{12}-v_{21}, v_{11}+v_{12}+v_{13}+v_{34}-v_{43})-\min(v_{42}-v_{24}, v_{12}+v_{31})$.

\begin{table}[hbt!]
\begin{center}
\begin{tabular}{||c|c||}
\hline
$\lambda_{1}$ & free \\
$\lambda_{2}$ & $(v_{12}-v_{32})/p$ \\
$\lambda_{3}$ & free \\
$\lambda_{4}$ & $v_{34}/p$ \\
$\lambda_{5}$ & $v_{24}/p$ \\
$\lambda_{6}$ & $(v_{24}+v_{34}+v_{41})/p$\\
$\lambda_{7}$ & $(v_{11}+v_{12}+v_{13}+v_{34})/p-\lambda_{1}$ \\
$\lambda_{8}$ & $(v_{13}+v_{24}+v_{34}+v_{41})/p$ \\
\hline
\end{tabular}
\caption{The values of $\lambda_{i}$ when the minimum occurs at $s_{114}$.}
\label{tab:svectorslambda3}
\end{center}
\end{table}

\begin{table}[hbt!]
\begin{center}
\begin{tabular}{||c|c|c|c||}
\hline
Vect. & Local param. $u_{ij}$ & Target deviation & Comparison \\
\hline
$v_{11}$ & $1-\lambda_{2}-\lambda_{4}-\lambda_{8}$ & 0 & 0 \\
$v_{12}$ & $\lambda_{2}+\lambda_{4}$ & 0 & 0 \\
$v_{13}$ & $\lambda_{8}-\lambda_{6}$ & 0 & 0 \\
$v_{14}$ & $\lambda_{6}$ & $v_{14}-v_{24}-v_{34}-v_{41}$ & $1-s_{114} \geq 0$ \\
$v_{21}$ & $1-\lambda_{1}-\lambda_{3}-\lambda_{8}$ & $v_{21}+\max(-v_{21}, v_{13}+v_{34}-v_{43})$ & $\geq 0$ \\
$v_{22}$ & $\lambda_{1}+\lambda_{3}$ & $v_{22}-v_{11}-v_{12}+\min(v_{21}, v_{43}-v_{13}-v_{34})$  & $\min(s_{222},f_{8})-s_{114} \geq 0$ \\
$v_{23}$ & $\lambda_{8}-\lambda_{5}$ & $v_{23}-v_{13}-v_{34}-v_{41}$ & $f_{1}-s_{114} \geq 0$ \\
$v_{24}$ & $\lambda_{5}$ & 0 & 0 \\
$v_{31}$ & $1-\lambda_{2}-\lambda_{6}-\lambda_{7}$ & $v_{12}+v_{31}-\min(v_{42}-v_{24}, v_{12}+v_{31})$  & $\geq 0$ \\
$v_{32}$ & $\lambda_{2}+\lambda_{6}$ & $v_{32}-v_{12}-v_{24}-v_{41}$ & $f_{5}-s_{114} \geq 0$ \\
$v_{33}$ & $\lambda_{7}-\lambda_{4}$ & $v_{33}-v_{11}-v_{13}+\min(v_{42}-v_{12}-v_{24}, v_{31})$ & $\min(f_{4},s_{332})-s_{114} \geq 0$ \\
$v_{34}$ & $\lambda_{4}$ & 0 & 0 \\
$v_{41}$ & $1-\lambda_{1}-\lambda_{5}-\lambda_{7}$ & 0 & 0 \\
$v_{42}$ & $\lambda_{1}+\lambda_{5}$ & $v_{42}-v_{24}-\min(v_{42}-v_{24}, v_{12}+v_{31})$ & $\geq 0$ \\
$v_{43}$ & $\lambda_{7}-\lambda_{3}$ & $v_{43}-v_{13}-v_{34}+\max(-v_{21}, v_{13}+v_{34}-v_{43})$ & $\geq 0$ \\
$v_{44}$ & $\lambda_{3}$ & $\min(s_{423}, e_{2},1, s_{432})-s_{114}$ & $\geq 0$ \\
\hline
\end{tabular}
\caption{Demonstration of the conditions for a 6 element closeness measure from the third class, here $s_{114}$. The $\lambda_{i}$ values come from Table \ref{tab:svectorslambda3}.}
\label{tab:svectors3}
\end{center}
\end{table}

In Table \ref{tab:svectors3} the measures for comparison with the benchmark $s_{114}$ are as follows: \\
$1$ \\
$v_{13}+v_{21}+v_{24}+v_{34}+v_{41}+L_{1}+L_{3}$ \\
$v_{11}+v_{12}+v_{13}+v_{22}+v_{24}+v_{34}+v_{41}-L_{1}-L_{3}$ \\
$v_{11}+v_{12}+v_{23}+v_{24}$ \\
$v_{11}+2v_{12}+v_{13}+v_{24}+v_{31}+v_{34}+v_{41}-L_{1}$ \\
$v_{11}+v_{13}+v_{32}+v_{34}$ \\
$v_{24}+v_{33}+v_{34}+v_{41}+L_{1}$ \\
$v_{11}+v_{12}+v_{13}+v_{34}+v_{41}+v_{42}-L_{1}$ \\
$v_{24}+v_{41}+v_{43}+L_{1}+L_{3}$ \\
$v_{11}+v_{12}+v_{13}+v_{24}+v_{34}+v_{41}+v_{44}-L_{3}$ \\
Here $L_{1}=\min(v_{42}-v_{24}, v_{12}+v_{31})$, we set $L_{1}+L_{3}=\max(v_{11}+v_{12}-v_{21}, v_{11}+v_{12}+v_{13}+v_{34}-v_{43})$, which gives $L_{3}=\max(v_{11}+v_{12}-v_{21}, v_{11}+v_{12}+v_{13}+v_{34}-v_{43})-\min(v_{42}-v_{24}, v_{12}+v_{31})$, with $L_{7}=v_{11}+v_{12}+v_{13}+v_{34}-L_{1}$.

\newpage

\noindent\textsf{\Large \textbf{\underline{Part D}}: Insights into solutions}\vspace{0.5cm}

We show in Table~\ref{long-table} the full set of 120 vectors for the non-signaling subspace $f$'s, $g$'s, $t$'s and  $z$'s in Eq.~(\ref{S}) and the 128 vectors for the locality subspace $f$'s, $g$'s, $t$'s, $s$'s and $e$'s in Eq.~(\ref{Q}). See also \textbf{Part A}. The vectors are labeled (e.g. $f_1$ or $g_{11}$) in a way following the explicit construction in \textbf{Part C}, rather than the simpler one in \textbf{Part B}. 

\begin{center}
\begin{longtable}{||c| c| c| c| c||}
\caption{List of vectors in $\mathbb{Q}$ and $\mathbb{S}$.\\}\label{long-table}\\
\hline
 \textbf{ID} & \textbf{Vectors} & \textbf{Directions} & \textbf{Category}& \textbf{Calculation}\\ 
 \hline
 \endhead
 
 f1 & (1, 1, 0, 0, 0, 0, 1, 1, 0, 0, 0, 0, 0, 0, 0, 0) & LNNN & F & f1 \\
 \hline
 f2 & (0, 0, 1, 1, 1, 1, 0, 0, 0, 0, 0, 0, 0, 0, 0, 0) & RNNN & F & f2 \\
 \hline
 f3 & (0, 0, 0, 0, 0, 0, 0, 0, 1, 1, 0, 0, 0, 0, 1, 1) & NNLN & F & f3 \\
 \hline
 f4 & (0, 0, 0, 0, 0, 0, 0, 0, 0, 0, 1, 1, 1, 1, 0, 0) & NNRN & F & f4 \\
 \hline
 f5 & (1, 0, 1, 0, 0, 0, 0, 0, 0, 1, 0, 1, 0, 0, 0, 0) & NLNN & F & f5 \\
 \hline
 f6 & (0, 1, 0, 1, 0, 0, 0, 0, 1, 0, 1, 0, 0, 0, 0, 0) & NRNN & F & f6 \\
 \hline
 f7 & (0, 0, 0, 0, 1, 0, 1, 0, 0, 0, 0, 0, 0, 1, 0, 1) & NNNL & F & f7 \\
 \hline
 f8 & (0, 0, 0, 0, 0, 1, 0, 1, 0, 0, 0, 0, 1, 0, 1, 0) & NNNR & F & f8 \\
 \hline
 g11 & (1, 0, 0, 0, 0, 0, 1, 1, 0, 1, 0, 1, 0, 0, 0, 0) & LLNN & G & f1 + f5 - (1 - v14) \\
 \hline 
 g12 & (0, 1, 0, 0, 0, 0, 1, 1, 1, 0, 1, 0, 0, 0, 0, 0) & LRNN & G & f1 + f6 - (1 - v13) \\
 \hline
 g13 & (0, 0, 1, 0, 1, 1, 0, 0, 0, 1, 0, 1, 0, 0, 0, 0) & RLNN & G & f2 + f5 - (1 - v12) \\
 \hline
 g14 & (0, 0, 0, 1, 1, 1, 0, 0, 1, 0, 1, 0, 0, 0, 0, 0) & RRNN & G & f2 + f6 - (1 - v11) \\
 \hline
 g21 & (0, 0, 1, 1, 1, 0, 0, 0, 0, 0, 0, 0, 0, 1, 0, 1) & RNNL & G & f2 + f7 - (1 - v24) \\
 \hline
 g22 & (0, 0, 1, 1, 0, 1, 0, 0, 0, 0, 0, 0, 1, 0, 1, 0) & RNNR & G & f2 + f8 - (1 - v23) \\
 \hline
 g23 & (1, 1, 0, 0, 0, 0, 1, 0, 0, 0, 0, 0, 0, 1, 0, 1) & LNNL & G & f1 + f7 - (1 - v22) \\
 \hline
 g24 & (1, 1, 0, 0, 0, 0, 0, 1, 0, 0, 0, 0, 1, 0, 1, 0) & LNNR & G & f1 + f8 - (1 - v21) \\
 \hline
 g31 & (0, 1, 0, 1, 0, 0, 0, 0, 1, 0, 0, 0, 0, 0, 1, 1) & NRLN & G & f3 + f6 - (1 - v34) \\
 \hline
 g32 & (1, 0, 1, 0, 0, 0, 0, 0, 0, 1, 0, 0, 0, 0, 1, 1) & NLLN & G & f3 + f5 - (1 - v33) \\
 \hline
 g33 & (0, 1, 0, 1, 0, 0, 0, 0, 0, 0, 1, 0, 1, 1, 0, 0) & NRRN & G & f4 + f6 - (1 - v32) \\
 \hline
 g34 & (1, 0, 1, 0, 0, 0, 0, 0, 0, 0, 0, 1, 1, 1, 0, 0) & NLRN & G & f4 + f5 - (1 - v31) \\
 \hline
 g41 & (0, 0, 0, 0, 0, 1, 0, 1, 0, 0, 1, 1, 1, 0, 0, 0) & NNRR & G & f4 + f8 - (1 - v44) \\
 \hline
 g42 & (0, 0, 0, 0, 1, 0, 1, 0, 0, 0, 1, 1, 0, 1, 0, 0) & NNRL & G & f4 + f7 - (1 - v43) \\
 \hline
 g43 & (0, 0, 0, 0, 0, 1, 0, 1, 1, 1, 0, 0, 0, 0, 1, 0) & NNLR & G & f3 + f8 - (1 - v42) \\
 \hline
 g44 & (0, 0, 0, 0, 1, 0, 1, 0, 1, 1, 0, 0, 0, 0, 0, 1) & NNLL & G & f3 + f7 - (1 - v41) \\
 \hline
 t11 & (1, 1, 0, 0, 0, 0, 1, 0, 1, 1, 0, 0, 0, 0, 0, 1) & LNLL & T1 & g23 + g44 – f7 \\
 \hline
 t12 & (1, 1, 0, 0, 0, 0, 0, 1, 1, 1, 0, 0, 0, 0, 1, 0) & LNLR & T1 & g24 + g43 – f8 \\
 \hline
 t13 & (1, 0, 0, 0, 0, 0, 1, 1, 0, 1, 0, 0, 0, 0, 1, 1) & LLLN & T1 & g11 + g32 – f5 \\
 \hline
 t14 & (0, 1, 0, 0, 0, 0, 1, 1, 1, 0, 0, 0, 0, 0, 1, 1) & LRLN & T1 & g12 + g31 – f6 \\
 \hline
 t21 & (1, 1, 0, 0, 0, 0, 1, 0, 0, 0, 1, 1, 0, 1, 0, 0) & LNRL & T2 & g23 + g42 – f7 \\
 \hline
 t22 & (1, 1, 0, 0, 0, 0, 0, 1, 0, 0, 1, 1, 1, 0, 0, 0) & LNRR & T2 & g24 + g41 – f8 \\
 \hline
 t23 & (1, 0, 0, 0, 0, 0, 1, 1, 0, 0, 0, 1, 1, 1, 0, 0) & LLRN & T2 & g11 + g34 – f5 \\
 \hline
 t24 & (0, 1, 0, 0, 0, 0, 1, 1, 0, 0, 1, 0, 1, 1, 0, 0) & LRRN & T2 & g12 + g33 – f6 \\
 \hline
 t31 & (0, 0, 1, 1, 1, 0, 0, 0, 1, 1, 0, 0, 0, 0, 0, 1) & RNLL & T2 & g21 + g44 – f7 \\
 \hline
 t32 & (0, 0, 1, 1, 0, 1, 0, 0, 1, 1, 0, 0, 0, 0, 1, 0) & RNLR & T2 & g22 + g43 – f8 \\
 \hline
 t33 & (0, 0, 1, 0, 1, 1, 0, 0, 0, 1, 0, 0, 0, 0, 1, 1) & RLLN & T2 & g13 + g32 – f5 \\
 \hline
 t34 & (0, 0, 0, 1, 1, 1, 0, 0, 1, 0, 0, 0, 0, 0, 1, 1) & RRLN & T2 & g14 + g31 – f6 \\
 \hline
 t41 & (0, 0, 1, 1, 1, 0, 0, 0, 0, 0, 1, 1, 0, 1, 0, 0) & RNRL & T1 & g21 + g42 – f7 \\
 \hline
 t42 & (0, 0, 1, 1, 0, 1, 0, 0, 0, 0, 1, 1, 1, 0, 0, 0) & RNRR & T1 & g22 + g41 – f8 \\
 \hline
 t43 & (0, 0, 1, 0, 1, 1, 0, 0, 0, 0, 0, 1, 1, 1, 0, 0) & RLRN & T1 & g13 + g34 – f5 \\
 \hline
 t44 & (0, 0, 0, 1, 1, 1, 0, 0, 0, 0, 1, 0, 1, 1, 0, 0) & RRRN & T1 & g14 + g33 – f6 \\
 \hline
 t51 & (1, 0, 1, 0, 1, 0, 1, 0, 0, 1, 0, 0, 0, 0, 0, 1) & NLLL & T1 & g32 + g44 – f3 \\
 \hline
 t52 & (1, 0, 1, 0, 1, 0, 1, 0, 0, 0, 0, 1, 0, 1, 0, 0) & NLRL & T1 & g34 + g42 – f4 \\
 \hline
 t53 & (1, 0, 0, 0, 0, 0, 1, 0, 0, 1, 0, 1, 0, 1, 0, 1) & LLNL & T1 & g11 + g23 – f1 \\
 \hline
 t54 & (0, 0, 1, 0, 1, 0, 0, 0, 0, 1, 0, 1, 0, 1, 0, 1) & RLNL & T1 & g13 + g21 – f2 \\
 \hline
 t61 & (1, 0, 1, 0, 0, 1, 0, 1, 0, 1, 0, 0, 0, 0, 1, 0) & NLLR & T2 & g32 + g43 – f3 \\
 \hline
 t62 & (1, 0, 1, 0, 0, 1, 0, 1, 0, 0, 0, 1, 1, 0, 0, 0) & NLRR & T2 & g34 + g41 – f4 \\
 \hline
 t63 & (1, 0, 0, 0, 0, 0, 0, 1, 0, 1, 0, 1, 1, 0, 1, 0) & LLNR & T2 & g11 + g24 – f1 \\
 \hline
 t64 & (0, 0, 1, 0, 0, 1, 0, 0, 0, 1, 0, 1, 1, 0, 1, 0) & RLNR & T2 & g13 + g22 – f2 \\
 \hline
 t71 & (0, 1, 0, 1, 1, 0, 1, 0, 1, 0, 0, 0, 0, 0, 0, 1) & NRLL & T2 & g31 + g44 – f3 \\
 \hline
 t72 & (0, 1, 0, 1, 1, 0, 1, 0, 0, 0, 1, 0, 0, 1, 0, 0) & NRRL & T2 & g33 + g42 – f4 \\
 \hline
 t73 & (0, 1, 0, 0, 0, 0, 1, 0, 1, 0, 1, 0, 0, 1, 0, 1) & LRNL & T2 & g12 + g23 – f1 \\
 \hline
 t74 & (0, 0, 0, 1, 1, 0, 0, 0, 1, 0, 1, 0, 0, 1, 0, 1) & RRNL & T2 & g14 + g21 – f2 \\
 \hline
 t81 & (0, 1, 0, 1, 0, 1, 0, 1, 1, 0, 0, 0, 0, 0, 1, 0) & NRLR & T1 & g31 + g43 – f3 \\
 \hline
 t82 & (0, 1, 0, 1, 0, 1, 0, 1, 0, 0, 1, 0, 1, 0, 0, 0) & NRRR & T1 & g33 + g41 – f4 \\
 \hline
 t83 & (0, 1, 0, 0, 0, 0, 0, 1, 1, 0, 1, 0, 1, 0, 1, 0) & LRNR & T1 & g12 + g24 – f1 \\
 \hline
 t84 & (0, 0, 0, 1, 0, 1, 0, 0, 1, 0, 1, 0, 1, 0, 1, 0) & RRNR & T1 & g14 + g22 – f2 \\
 \hline
 z111 & (2, 1, 1, 0, 0, 0, 1, 0, 0, 1, 0, 0, 0, 0, 0, 1) & LLLL (3) & Z3 & t11 + t51 – g44 \\
 \hline
 z112 & (2, 1, 1, 0, 0, 0, 1, 0, 0, 0, 0, 1, 0, 1, 0, 0) & LLRL (3) & Z1 & t21 + t52 – g42 \\
 \hline
 z113 & (2, 1, 1, 0, 0, 0, 0, 1, 0, 1, 0, 0, 0, 0, 1, 0) & LLLR (3) & Z1 & t12 + t61 – g43 \\
 \hline
 z114 & (2, 1, 1, 0, 0, 0, 0, 1, 0, 0, 0, 1, 1, 0, 0, 0) & LLRR (3) & Z2 & t22 + t62 – g41 \\
 \hline
 z121 & (1, 2, 0, 1, 0, 0, 1, 0, 1, 0, 0, 0, 0, 0, 0, 1) & LRLL (3) & Z1 & t11 + t71 – g44 \\
 \hline
 z122 & (1, 2, 0, 1, 0, 0, 1, 0, 0, 0, 1, 0, 0, 1, 0, 0) & LRRL (3) & Z2 & t21 + t72 – g42 \\
 \hline
 z123 & (1, 2, 0, 1, 0, 0, 0, 1, 1, 0, 0, 0, 0, 0, 1, 0) & LRLR (3) & Z3 & t12 + t81 – g43 \\
 \hline
 z124 & (1, 2, 0, 1, 0, 0, 0, 1, 0, 0, 1, 0, 1, 0, 0, 0) & LRRR (3) & Z1 & t22 + t82 – g41 \\
 \hline
 z131 & (1, 0, 2, 1, 1, 0, 0, 0, 0, 1, 0, 0, 0, 0, 0, 1) & RLLL (3) & Z1 & t31 + t51 – g44 \\
 \hline
 z132 & (1, 0, 2, 1, 1, 0, 0, 0, 0, 0, 0, 1, 0, 1, 0, 0) & RLRL (3) & Z3 & t41 + t52 – g42 \\
 \hline
 z133 & (1, 0, 2, 1, 0, 1, 0, 0, 0, 1, 0, 0, 0, 0, 1, 0) & RLLR (3) & Z2 & t32 + t61 – g43 \\
 \hline
 z134 & (1, 0, 2, 1, 0, 1, 0, 0, 0, 0, 0, 1, 1, 0, 0, 0) & RLRR (3) & Z1 & t42 + t62 – g41 \\
 \hline
 z141 & (0, 1, 1, 2, 1, 0, 0, 0, 1, 0, 0, 0, 0, 0, 0, 1) & RRLL (3) & Z2 & t31 + t71 – g44 \\
 \hline
 z142 & (0, 1, 1, 2, 1, 0, 0, 0, 0, 0, 1, 0, 0, 1, 0, 0) & RRRL (3) & Z1 & t41 + t72 – g42 \\
 \hline
 z143 & (0, 1, 1, 2, 0, 1, 0, 0, 1, 0, 0, 0, 0, 0, 1, 0) & RRLR (3) & Z1 & t32 + t81 – g43 \\
 \hline
 z144 & (0, 1, 1, 2, 0, 1, 0, 0, 0, 0, 1, 0, 1, 0, 0, 0) & RRRR (3) & Z3 & t42 + t82 – g41 \\
 \hline
 z211 & (0, 0, 1, 0, 2, 1, 1, 0, 0, 1, 0, 0, 0, 0, 0, 1) & RLLL (2) & Z1 & t33 + t51 – g32 \\
 \hline
 z212 & (0, 0, 1, 0, 2, 1, 1, 0, 0, 0, 0, 1, 0, 1, 0, 0) & RLRL (2) & Z3 & t43 + t52 – g34 \\
 \hline
 z213 & (0, 0, 0, 1, 2, 1, 1, 0, 1, 0, 0, 0, 0, 0, 0, 1) & RRLL (2) & Z2 & t34 + t71 – g31 \\
 \hline
 z214 & (0, 0, 0, 1, 2, 1, 1, 0, 0, 0, 1, 0, 0, 1, 0, 0) & RRRL (2) & Z1 & t44 + t72 – g33 \\
 \hline
 z221 & (0, 0, 1, 0, 1, 2, 0, 1, 0, 1, 0, 0, 0, 0, 1, 0) & RLLR (2) & Z2 & t33 + t61 – g32 \\
 \hline
 z222 & (0, 0, 1, 0, 1, 2, 0, 1, 0, 0, 0, 1, 1, 0, 0, 0) & RLRR (2) & Z1 & t43 + t62 – g34 \\
 \hline
 z223 & (0, 0, 0, 1, 1, 2, 0, 1, 1, 0, 0, 0, 0, 0, 1, 0) & RRLR (2) & Z1 & t34 + t81 – g31 \\
 \hline
 z224 & (0, 0, 0, 1, 1, 2, 0, 1, 0, 0, 1, 0, 1, 0, 0, 0) & RRRR (2) & Z3 & t44 + t82 – g33 \\
 \hline
 z231 & (1, 0, 0, 0, 1, 0, 2, 1, 0, 1, 0, 0, 0, 0, 0, 1) & LLLL (2) & Z3 & t13 + t51 – g32 \\
 \hline
 z232 & (1, 0, 0, 0, 1, 0, 2, 1, 0, 0, 0, 1, 0, 1, 0, 0) & LLRL (2) & Z1 & t23 + t52 – g34 \\
 \hline
 z233 & (0, 1, 0, 0, 1, 0, 2, 1, 1, 0, 0, 0, 0, 0, 0, 1) & LRLL (2) & Z1 & t14 + t71 – g31 \\
 \hline
 z234 & (0, 1, 0, 0, 1, 0, 2, 1, 0, 0, 1, 0, 0, 1, 0, 0) & LRRL (2) & Z2 & t24 + t72 – g33 \\
 \hline
 z241 & (1, 0, 0, 0, 0, 1, 1, 2, 0, 1, 0, 0, 0, 0, 1, 0) & LLLR (2) & Z1 & t13 + t61 – g32 \\
 \hline
 z242 & (1, 0, 0, 0, 0, 1, 1, 2, 0, 0, 0, 1, 1, 0, 0, 0) & LLRR (2) & Z2 & t23 + t62 – g34 \\
 \hline
 z243 & (0, 1, 0, 0, 0, 1, 1, 2, 1, 0, 0, 0, 0, 0, 1, 0) & LRLR (2) & Z3 & t14 + t81 – g31 \\
 \hline
 z244 & (0, 1, 0, 0, 0, 1, 1, 2, 0, 0, 1, 0, 1, 0, 0, 0) & LRRR (2) & Z1 & t24 + t82 – g33 \\
 \hline
 z311 & (0, 1, 0, 0, 0, 0, 1, 0, 2, 1, 1, 0, 0, 0, 0, 1) & LRLL (4) & Z1 & t11 + t73 – g23 \\
 \hline
 z312 & (0, 1, 0, 0, 0, 0, 0, 1, 2, 1, 1, 0, 0, 0, 1, 0) & LRLR (4) & Z3 & t12 + t83 – g24 \\
 \hline
 z313 & (0, 0, 0, 1, 1, 0, 0, 0, 2, 1, 1, 0, 0, 0, 0, 1) & RRLL (4) & Z2 & t31 + t74 – g21 \\
 \hline
 z314 & (0, 0, 0, 1, 0, 1, 0, 0, 2, 1, 1, 0, 0, 0, 1, 0) & RRLR (4) & Z1 & t32 + t84 – g22 \\
 \hline
 z321 & (1, 0, 0, 0, 0, 0, 1, 0, 1, 2, 0, 1, 0, 0, 0, 1) & LLLL (4) & Z3 & t11 + t53 – g23 \\
 \hline
 z322 & (1, 0, 0, 0, 0, 0, 0, 1, 1, 2, 0, 1, 0, 0, 1, 0) & LLLR (4) & Z1 & t12 + t63 – g24 \\
 \hline
 z323 & (0, 0, 1, 0, 1, 0, 0, 0, 1, 2, 0, 1, 0, 0, 0, 1) & RLLL (4) & Z1 & t31 + t54 – g21 \\
 \hline
 z324 & (0, 0, 1, 0, 0, 1, 0, 0, 1, 2, 0, 1, 0, 0, 1, 0) & RLLR (4) & Z2 & t32 + t64 – g22 \\
 \hline
 z331 & (0, 1, 0, 0, 0, 0, 1, 0, 1, 0, 2, 1, 0, 1, 0, 0) & LRRL (4) & Z2 & t21 + t73 – g23 \\
 \hline
 z332 & (0, 1, 0, 0, 0, 0, 0, 1, 1, 0, 2, 1, 1, 0, 0, 0) & LRRR (4) & Z1 & t22 + t83 – g24 \\
 \hline
 z333 & (0, 0, 0, 1, 1, 0, 0, 0, 1, 0, 2, 1, 0, 1, 0, 0) & RRRL (4) & Z1 & t41 + t74 – g21 \\
 \hline
 z334 & (0, 0, 0, 1, 0, 1, 0, 0, 1, 0, 2, 1, 1, 0, 0, 0) & RRRR (4) & Z3 & t42 + t84 – g22 \\
 \hline
 z341 & (1, 0, 0, 0, 0, 0, 1, 0, 0, 1, 1, 2, 0, 1, 0, 0) & LLRL (4) & Z1 & t21 + t53 – g23 \\
 \hline
 z342 & (1, 0, 0, 0, 0, 0, 0, 1, 0, 1, 1, 2, 1, 0, 0, 0) & LLRR (4) & Z2 & t22 + t63 – g24 \\
 \hline
 z343 & (0, 0, 1, 0, 1, 0, 0, 0, 0, 1, 1, 2, 0, 1, 0, 0) & RLRL (4) & Z3 & t41 + t54 – g21 \\
 \hline
 z344 & (0, 0, 1, 0, 0, 1, 0, 0, 0, 1, 1, 2, 1, 0, 0, 0) & RLRR (4) & Z1 & t42 + t64 – g22 \\
 \hline
 z411 & (1, 0, 0, 0, 0, 0, 0, 1, 0, 0, 0, 1, 2, 1, 1, 0) & LLRR (1) & Z2 & t23 + t63 – g11 \\
 \hline
 z412 & (0, 1, 0, 0, 0, 0, 0, 1, 0, 0, 1, 0, 2, 1, 1, 0) & LRRR (1) & Z1 & t24 + t83 – g12 \\
 \hline
 z413 & (0, 0, 1, 0, 0, 1, 0, 0, 0, 0, 0, 1, 2, 1, 1, 0) & RLRR (1) & Z1 & t43 + t64 – g13 \\
 \hline
 z414 & (0, 0, 0, 1, 0, 1, 0, 0, 0, 0, 1, 0, 2, 1, 1, 0) & RRRR (1) & Z3 & t44 + t84 – g14 \\
 \hline
 z421 & (1, 0, 0, 0, 0, 0, 1, 0, 0, 0, 0, 1, 1, 2, 0, 1) & LLRL (1) & Z1 & t23 + t53 – g11 \\
 \hline
 z422 & (0, 1, 0, 0, 0, 0, 1, 0, 0, 0, 1, 0, 1, 2, 0, 1) & LRRL (1) & Z2 & t24 + t73 – g12 \\
 \hline
 z423 & (0, 0, 1, 0, 1, 0, 0, 0, 0, 0, 0, 1, 1, 2, 0, 1) & RLRL (1) & Z3 & t43 + t54 – g13 \\
 \hline
 z424 & (0, 0, 0, 1, 1, 0, 0, 0, 0, 0, 1, 0, 1, 2, 0, 1) & RRRL (1) & Z1 & t44 + t74 – g14 \\
 \hline
 z431 & (1, 0, 0, 0, 0, 0, 0, 1, 0, 1, 0, 0, 1, 0, 2, 1) & LLLR (1) & Z1 & t13 + t63 – g11 \\
 \hline
 z432 & (0, 1, 0, 0, 0, 0, 0, 1, 1, 0, 0, 0, 1, 0, 2, 1) & LRLR (1) & Z3 & t14 + t83 – g12 \\
 \hline
 z433 & (0, 0, 1, 0, 0, 1, 0, 0, 0, 1, 0, 0, 1, 0, 2, 1) & RLLR (1) & Z2 & t33 + t64 – g13 \\
 \hline
 z434 & (0, 0, 0, 1, 0, 1, 0, 0, 1, 0, 0, 0, 1, 0, 2, 1) & RRLR (1) & Z1 & t34 + t84 – g14 \\
 \hline
 z441 & (1, 0, 0, 0, 0, 0, 1, 0, 0, 1, 0, 0, 0, 1, 1, 2) & LLLL (1) & Z3 & t13 + t53 – g11 \\
 \hline
 z442 & (0, 1, 0, 0, 0, 0, 1, 0, 1, 0, 0, 0, 0, 1, 1, 2) & LRLL (1) & Z1 & t14 + t73 – g12 \\
 \hline
 z443 & (0, 0, 1, 0, 1, 0, 0, 0, 0, 1, 0, 0, 0, 1, 1, 2) & RLLL (1) & Z1 & t33 + t54 – g13 \\
 \hline
 z444 & (0, 0, 0, 1, 1, 0, 0, 0, 1, 0, 0, 0, 0, 1, 1, 2) & RRLL (1) & Z2 & t34 + t74 – g14 \\
 \hline
 s111 & (1, 1, 1, 0, 0, 0, 1, 0, 0, 1, 0, 0, 0, 0, 0, 1) & LLLL (3) & S3 & z111 – v11 \\
 \hline
 s112 & (1, 1, 1, 0, 0, 0, 1, 0, 0, 0, 0, 1, 0, 1, 0, 0) & LLRL (3) & S1 & z112 – v11 \\
 \hline
 s113 & (1, 1, 1, 0, 0, 0, 0, 1, 0, 1, 0, 0, 0, 0, 1, 0) & LLLR (3) & S1 & z113 – v11 \\
 \hline
 s114 & (1, 1, 1, 0, 0, 0, 0, 1, 0, 0, 0, 1, 1, 0, 0, 0) & LLRR (3) & S2 & z114 – v11 \\
 \hline
 s121 & (1, 1, 0, 1, 0, 0, 1, 0, 1, 0, 0, 0, 0, 0, 0, 1) & LRLL (3) & S1 & z121 – v12 \\
 \hline
 s122 & (1, 1, 0, 1, 0, 0, 1, 0, 0, 0, 1, 0, 0, 1, 0, 0) & LRRL (3) & S2 & z122 – v12 \\
 \hline
 s123 & (1, 1, 0, 1, 0, 0, 0, 1, 1, 0, 0, 0, 0, 0, 1, 0) & LRLR (3) & S3 & z123 - v12 \\
 \hline
 s124 & (1, 1, 0, 1, 0, 0, 0, 1, 0, 0, 1, 0, 1, 0, 0, 0) & LRRR (3) & S1 & z124 – v12 \\
 \hline
 s131 & (1, 0, 1, 1, 1, 0, 0, 0, 0, 1, 0, 0, 0, 0, 0, 1) & RLLL (3) & S1 & z131 – v13 \\
 \hline
 s132 & (1, 0, 1, 1, 1, 0, 0, 0, 0, 0, 0, 1, 0, 1, 0, 0) & RLRL (3) & S3 & z132 – v13 \\
 \hline
 s133 & (1, 0, 1, 1, 0, 1, 0, 0, 0, 1, 0, 0, 0, 0, 1, 0) & RLLR (3) & S2 & z133 – v13 \\
 \hline
 s134 & (1, 0, 1, 1, 0, 1, 0, 0, 0, 0, 0, 1, 1, 0, 0, 0) & RLRR (3) & S1 & z134 – v13 \\
 \hline
 s141 & (0, 1, 1, 1, 1, 0, 0, 0, 1, 0, 0, 0, 0, 0, 0, 1) & RRLL (3) & S2 & z141 – v14 \\
 \hline
 s142 & (0, 1, 1, 1, 1, 0, 0, 0, 0, 0, 1, 0, 0, 1, 0, 0) & RRRL (3) & S1 & z142 – v14 \\
 \hline
 s143 & (0, 1, 1, 1, 0, 1, 0, 0, 1, 0, 0, 0, 0, 0, 1, 0) & RRLR (3) & S1 & z143 – v14 \\
 \hline
 s144 & (0, 1, 1, 1, 0, 1, 0, 0, 0, 0, 1, 0, 1, 0, 0, 0) & RRRR (3) & S3 & z144 – v14 \\
 \hline
 s211 & (0, 0, 1, 0, 1, 1, 1, 0, 0, 1, 0, 0, 0, 0, 0, 1) & RLLL (2) & S1 & z211 – v21 \\
 \hline
 s212 & (0, 0, 1, 0, 1, 1, 1, 0, 0, 0, 0, 1, 0, 1, 0, 0) & RLRL (2) & S3 & z212 – v21 \\
 \hline
 s213 & (0, 0, 0, 1, 1, 1, 1, 0, 1, 0, 0, 0, 0, 0, 0, 1) & RRLL (2) & S2 & z213 – v21 \\
 \hline
 s214 & (0, 0, 0, 1, 1, 1, 1, 0, 0, 0, 1, 0, 0, 1, 0, 0) & RRRL (2) & S1 & z214 – v21 \\
 \hline
 s221 & (0, 0, 1, 0, 1, 1, 0, 1, 0, 1, 0, 0, 0, 0, 1, 0) & RLLR (2) & S2 & z221 – v22 \\
 \hline
 s222 & (0, 0, 1, 0, 1, 1, 0, 1, 0, 0, 0, 1, 1, 0, 0, 0) & RLRR (2) & S1 & z222 – v22 \\
 \hline
 s223 & (0, 0, 0, 1, 1, 1, 0, 1, 1, 0, 0, 0, 0, 0, 1, 0) & RRLR (2) & S1 & z223 – v22 \\
 \hline
 s224 & (0, 0, 0, 1, 1, 1, 0, 1, 0, 0, 1, 0, 1, 0, 0, 0) & RRRR (2) & S3 & z224 – v22 \\
 \hline
 s231 & (1, 0, 0, 0, 1, 0, 1, 1, 0, 1, 0, 0, 0, 0, 0, 1) & LLLL (2) & S3 & z231 – v23 \\
 \hline
 s232 & (1, 0, 0, 0, 1, 0, 1, 1, 0, 0, 0, 1, 0, 1, 0, 0) & LLRL (2) & S1 & z232 – v23 \\
 \hline
 s233 & (0, 1, 0, 0, 1, 0, 1, 1, 1, 0, 0, 0, 0, 0, 0, 1) & LRLL (2) & S1 & z233 – v23 \\
 \hline
 s234 & (0, 1, 0, 0, 1, 0, 1, 1, 0, 0, 1, 0, 0, 1, 0, 0) & LRRL (2) & S2 & z234 – v23 \\
 \hline
 s241 & (1, 0, 0, 0, 0, 1, 1, 1, 0, 1, 0, 0, 0, 0, 1, 0) & LLLR (2) & S1 & z241 – v24 \\
 \hline
 s242 & (1, 0, 0, 0, 0, 1, 1, 1, 0, 0, 0, 1, 1, 0, 0, 0) & LLRR (2) & S2 & z242 – v24 \\
 \hline
 s243 & (0, 1, 0, 0, 0, 1, 1, 1, 1, 0, 0, 0, 0, 0, 1, 0) & LRLR (2) & S3 & z243 – v24 \\
 \hline
 s244 & (0, 1, 0, 0, 0, 1, 1, 1, 0, 0, 1, 0, 1, 0, 0, 0) & LRRR (2) & S1 & z244 – v24 \\
 \hline
 s311 & (0, 1, 0, 0, 0, 0, 1, 0, 1, 1, 1, 0, 0, 0, 0, 1) & LRLL (4) & S1 & z311 – v31 \\
 \hline
 s312 & (0, 1, 0, 0, 0, 0, 0, 1, 1, 1, 1, 0, 0, 0, 1, 0) & LRLR (4) & S3 & z312 – v31 \\
 \hline
 s313 & (0, 0, 0, 1, 1, 0, 0, 0, 1, 1, 1, 0, 0, 0, 0, 1) & RRLL (4) & S2 & z313 – v31 \\
 \hline
 s314 & (0, 0, 0, 1, 0, 1, 0, 0, 1, 1, 1, 0, 0, 0, 1, 0) & RRLR (4) & S1 & z314 – v31 \\
 \hline
 s321 & (1, 0, 0, 0, 0, 0, 1, 0, 1, 1, 0, 1, 0, 0, 0, 1) & LLLL (4) & S3 & z321 – v32 \\
 \hline
 s322 & (1, 0, 0, 0, 0, 0, 0, 1, 1, 1, 0, 1, 0, 0, 1, 0) & LLLR (4) & S1 & z322 – v32 \\
 \hline
 s323 & (0, 0, 1, 0, 1, 0, 0, 0, 1, 1, 0, 1, 0, 0, 0, 1) & RLLL (4) & S1 & z323 – v32 \\
 \hline
 s324 & (0, 0, 1, 0, 0, 1, 0, 0, 1, 1, 0, 1, 0, 0, 1, 0) & RLLR (4) & S2 & z324 – v32 \\
 \hline
 s331 & (0, 1, 0, 0, 0, 0, 1, 0, 1, 0, 1, 1, 0, 1, 0, 0) & LRRL (4) & S2 & z331 – v33 \\
 \hline
 s332 & (0, 1, 0, 0, 0, 0, 0, 1, 1, 0, 1, 1, 1, 0, 0, 0) & LRRR (4) & S1 & z332 – v33 \\
 \hline
 s333 & (0, 0, 0, 1, 1, 0, 0, 0, 1, 0, 1, 1, 0, 1, 0, 0) & RRRL (4) & S1 & z333 – v33 \\
 \hline
 s334 & (0, 0, 0, 1, 0, 1, 0, 0, 1, 0, 1, 1, 1, 0, 0, 0) & RRRR (4) & S3 & z334 – v33 \\
 \hline
 s341 & (1, 0, 0, 0, 0, 0, 1, 0, 0, 1, 1, 1, 0, 1, 0, 0) & LLRL (4) & S1 & z341 – v34 \\
 \hline
 s342 & (1, 0, 0, 0, 0, 0, 0, 1, 0, 1, 1, 1, 1, 0, 0, 0) & LLRR (4) & S2 & z342 – v34 \\
 \hline
 s343 & (0, 0, 1, 0, 1, 0, 0, 0, 0, 1, 1, 1, 0, 1, 0, 0) & RLRL (4) & S3 & z343 – v34 \\
 \hline
 s344 & (0, 0, 1, 0, 0, 1, 0, 0, 0, 1, 1, 1, 1, 0, 0, 0) & RLRR (4) & S1 & z344 – v34 \\
 \hline
 s411 & (1, 0, 0, 0, 0, 0, 0, 1, 0, 0, 0, 1, 1, 1, 1, 0) & LLRR (1) & S2 & z411 – v41 \\
 \hline
 s412 & (0, 1, 0, 0, 0, 0, 0, 1, 0, 0, 1, 0, 1, 1, 1, 0) & LRRR (1) & S1 & z412 – v41 \\
 \hline
 s413 & (0, 0, 1, 0, 0, 1, 0, 0, 0, 0, 0, 1, 1, 1, 1, 0) & RLRR (1) & S1 & z413 – v41 \\
 \hline
 s414 & (0, 0, 0, 1, 0, 1, 0, 0, 0, 0, 1, 0, 1, 1, 1, 0) & RRRR (1) & S3 & z414 – v41 \\
 \hline
 s421 & (1, 0, 0, 0, 0, 0, 1, 0, 0, 0, 0, 1, 1, 1, 0, 1) & LLRL (1) & S1 & z421 – v42 \\
 \hline
 s422 & (0, 1, 0, 0, 0, 0, 1, 0, 0, 0, 1, 0, 1, 1, 0, 1) & LRRL (1) & S2 & z422 – v42 \\
 \hline
 s423 & (0, 0, 1, 0, 1, 0, 0, 0, 0, 0, 0, 1, 1, 1, 0, 1) & RLRL (1) & S3 & z423 – v42 \\
 \hline
 s424 & (0, 0, 0, 1, 1, 0, 0, 0, 0, 0, 1, 0, 1, 1, 0, 1) & RRRL (1) & S1 & z424 – v42 \\
 \hline
 s431 & (1, 0, 0, 0, 0, 0, 0, 1, 0, 1, 0, 0, 1, 0, 1, 1) & LLLR (1) & S1 & z431 – v43 \\
 \hline
 s432 & (0, 1, 0, 0, 0, 0, 0, 1, 1, 0, 0, 0, 1, 0, 1, 1) & LRLR (1) & S3 & z432 – v43 \\
 \hline
 s433 & (0, 0, 1, 0, 0, 1, 0, 0, 0, 1, 0, 0, 1, 0, 1, 1) & RLLR (1) & S2 & z433 – v43 \\
 \hline
 s434 & (0, 0, 0, 1, 0, 1, 0, 0, 1, 0, 0, 0, 1, 0, 1, 1) & RRLR (1) & S1 & z434 – v43 \\
 \hline
 s441 & (1, 0, 0, 0, 0, 0, 1, 0, 0, 1, 0, 0, 0, 1, 1, 1) & LLLL (1) & S3 & z441 – v44 \\
 \hline
 s442 & (0, 1, 0, 0, 0, 0, 1, 0, 1, 0, 0, 0, 0, 1, 1, 1) & LRLL (1) & S1 & z442 – v44 \\
 \hline
 s443 & (0, 0, 1, 0, 1, 0, 0, 0, 0, 1, 0, 0, 0, 1, 1, 1) & RLLL (1) & S1 & z443 – v44 \\
 \hline
 s444 & (0, 0, 0, 1, 1, 0, 0, 0, 1, 0, 0, 0, 0, 1, 1, 1) & RRLL (1) & S2 & z444 – v44 \\
 \hline
 e1 & (1, 0, 0, 1, 0, 1, 1, 0, 0, 1, 1, 0, 0, 1, 1, 0) & - & E & s122 + s133 – 1 \\
 \hline
 &  &  &  & s214 + s241 – 1 \\
 \hline
 &  &  &  & s314 + s341 – 1 \\
 \hline
 &  &  &  & s414 + s441 – 1 \\
 \hline
 e2 & (0, 1, 1, 0, 1, 0, 0, 1, 1, 0, 0, 1, 1, 0, 0, 1) & - & E & s114 + s141 – 1 \\
 \hline
 &  &  &  & s222 + s233 – 1 \\
 \hline
 &  &  &  & s323 + s332 – 1 \\
 \hline
 &  &  &  & s423 + s432 – 1 \\
 \hline
 e3 & (0, 1, 1, 0, 1, 0, 0, 1, 0, 1, 1, 0, 0, 1, 1, 0) & - & E & s113 + s142 – 1 \\
 \hline
 &  &  &  & s221 + s234 – 1 \\
 \hline
 &  &  &  & s312 + s343 – 1 \\
 \hline
 &  &  &  & s412 + s443 – 1 \\
 \hline
 e4 & (1, 0, 0, 1, 0, 1, 1, 0, 1, 0, 0, 1, 1, 0, 0, 1) & - & E & s121 + s134 – 1 \\
 \hline
 &  &  &  & s213 + s242 – 1 \\
 \hline
 &  &  &  & s321 + s334 – 1 \\
 \hline
 &  &  &  & s421 + s434 – 1 \\
 \hline
 e5 & (0, 1, 1, 0, 0, 1, 1, 0, 1, 0, 0, 1, 0, 1, 1, 0) & - & E & s112 + s143 – 1 \\
 \hline
 &  &  &  & s212 + s243 – 1 \\
 \hline
 &  &  &  & s331 + s324 – 1 \\
 \hline
 &  &  &  & s413 + s442 – 1 \\
 \hline
 e6 & (1, 0, 0, 1, 1, 0, 0, 1, 0, 1, 1, 0, 1, 0, 0, 1) & - & E & s124 + s131 – 1 \\
 \hline
 &  &  &  & s224 + s231 – 1 \\
 \hline
 &  &  &  & s313 + s342 – 1 \\
 \hline
 &  &  &  & s424 + s431 – 1 \\
 \hline
 e7 & (0, 1, 1, 0, 0, 1, 1, 0, 0, 1, 1, 0, 1, 0, 0, 1) & - & E & s111 + s144 – 1 \\
 \hline
 &  &  &  & s211 + s244 – 1 \\
 \hline
 &  &  &  & s311 + s344 – 1 \\
 \hline
 &  &  &  & s422 + s433 – 1 \\
 \hline
 e8 & (1, 0, 0, 1, 1, 0, 0, 1, 1, 0, 0, 1, 0, 1, 1, 0) & - & E & s123 + s132 – 1 \\
 \hline
 &  &  &  & s223 + s232 – 1 \\
 \hline
 &  &  &  & s322 + s333 – 1 \\
 \hline
 &  &  &  & s411 + s444 – 1 \\
 \hline
 
\end{longtable}
\end{center}

It is possible to describe the vectors in each class, in terms of which non-signaling constraints are violated -- for the $f$, $g$, $t$, and $z$/$s$ classes, the vectors are associated with violations of 1, 2, 3, and 4 non-signaling constraints, respectively. 

Specifically, we order the signaling constraints as $(\Delta_{\scriptscriptstyle  1}, \Delta_{\scriptscriptstyle  3},  \Delta_{\scriptscriptstyle  2}, \Delta_{\scriptscriptstyle  4})$, following Eqs.~(\ref{NS-AB0})-(\ref{NS-BA1}), that is alternating the constraints corresponding to the Alice to Bob direction. Then, all vectors can be identified through the heuristic rule of violating non-signaling constraints, in the order of alternating directions, gradually increasing the number of non-signaling constraints violated. In Table~\ref{long-table}  (‘Directions’ column), no constraint violation is indicated as N, positive violations as L, and negative ones as R. Then, violating a single constraint leads to the four $f$-vectors, $\{(f_1, f_2), (f_5, f_6), (f_3, f_4), (f_7, f_8)\}$. For example, $\Delta_{\scriptscriptstyle 1}= v_{11} + v_{12} -v_{21}-v_{22} = f_1-1=1-f_2$ so that $f_1$ corresponds to $\Delta_{\scriptscriptstyle  1-}$ and $f_2$ corresponds to $\Delta_{\scriptscriptstyle  1+}$. For the $z$-vectors, because all four constraints are violated, there are four possibilities for which constraint is violated first (the numbers in parentheses). For example, $\Delta_{\scriptscriptstyle  1-}$, $\Delta_{\scriptscriptstyle  3-}$, $\Delta_{\scriptscriptstyle  2-}$, $\Delta_{\scriptscriptstyle  4-}$ corresponds to LLLL(1), LLLL(2), LLLL(3) and LLLL(4) i.e. $z_{114}$, $z_{234}$, $z_{324}$  and $z_{444}$ respectively. This heuristic rule leads to the construction of all 120 vectors to compute the non-signaling fraction. In order to construct the 128 vectors for the locality fraction, we add the $e$’s and replace the $z$’s with $s$’s – the cardinalities and correspondence of signaling constraint violations are the same. 

\begin{figure}[th]
\centering
\includegraphics[width=0.5\columnwidth]{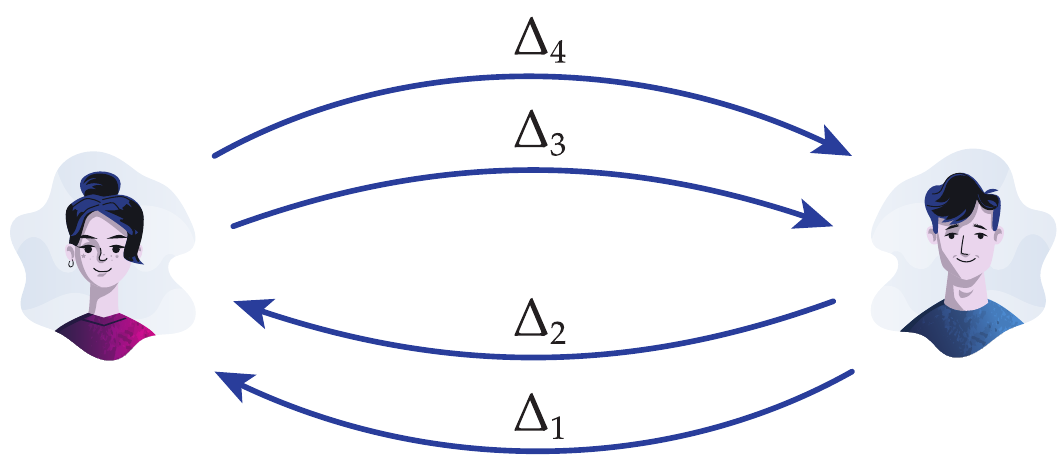}\caption
{\label{Fig-LNNR}{\bf\textsf{'Alternating' direction of signaling violation.}} 1) There are two parties, A and B, with two arrows from A to B (4,3) and two arrows from B to A (2,1).
2) We traverse a path starting from B or A using each arrow only once.
3) When traversing an arrow we carry a flag L or R (these may differ along a composite path).
4) Path reversals are equivalent, e.g. L along 1 then R along 3 is equivalent to R along 3 then L along 1.
5) Each path has a unique start and end.
}
\end{figure}

A more complete description of this heuristic rule is given in Fig.~\ref{Fig-LNNR}. This also helps to describe aspects of the symmetry orbits across the different vector classes, as indicated by the “Cat.” (category) column in Table~\ref{long-table}, where F, G, T1, T2, S1, S2, S3, Z1, Z2, Z3, and E denote the orbits of the symmetry group in Eqs.(\ref{sym-1})-(\ref{sym-5}), corresponding to the representatives $\vec{f}_\alpha$, $\vec{g}_\beta$, $\vec{t}_\beta$, $\vec{t}_\gamma$, $\vec{s}_\beta$, $\vec{s}_\gamma$, $\vec{s}_\delta$, $\vec{z}_\beta$, $\vec{z}_\gamma$, $\vec{z}_\delta$ and $\vec{e}_\alpha$, respectively, as listed in Eq.~(\ref{representatives}). For the $g$-vectors, these correspond to only two violations of non-signaling constraints, one in each direction, so there is only one class of vectors. The $t$-vectors correspond to the violation of three non-signaling constraints, two paths from Alice to Bob and one from Bob to Alice, or vice versa (see Fig.~\ref{Fig-LNNR}). There are symmetry classes corresponding to whether the two non-signaling constraints in the same direction have the same sign (T1) or not (T2). The $z$-vectors correspond to violations of all four NS constraints. Considering the two pairs of non-signaling constraints in each direction ($\Delta_{\scriptscriptstyle  1}$, $\Delta_{\scriptscriptstyle  2}$ and $\Delta_{\scriptscriptstyle  3}$, $\Delta_{\scriptscriptstyle  4}$), symmetry orbits correspond to whether the signs are the same in one pair and opposite in the other pair (Z1), the signs are opposite in both pairs (Z2), and the signs are the same in both pairs (Z3).

The rule above can be used to compute the cardinalities for each vector class. Note, the ‘alternating’ part of the rule explains why cardinalities do not just follow the number of non-signaling constraints in each class. For example, without the alternating rule, regarding the cardinality of the $g$ class there would be $(_2^4)\cdot2\cdot2=24$ vectors, but instead we have only 16.

Another way to illustrate the set of vectors is with a set of 4x4 tables, each representing a vector in the solution sets. For the $f$-vectors, this is depicted in Fig.~\ref{4x4_f's}, while the relationship between $f$-measures and associated $g$-, $t$-, $z$-measures is depicted in Fig.~\ref{4x4_relations}.

\begin{figure}[h]
\centering
\includegraphics[page=1,
        trim=2cm 7cm 2cm 4cm,
        clip,
        width=\textwidth]{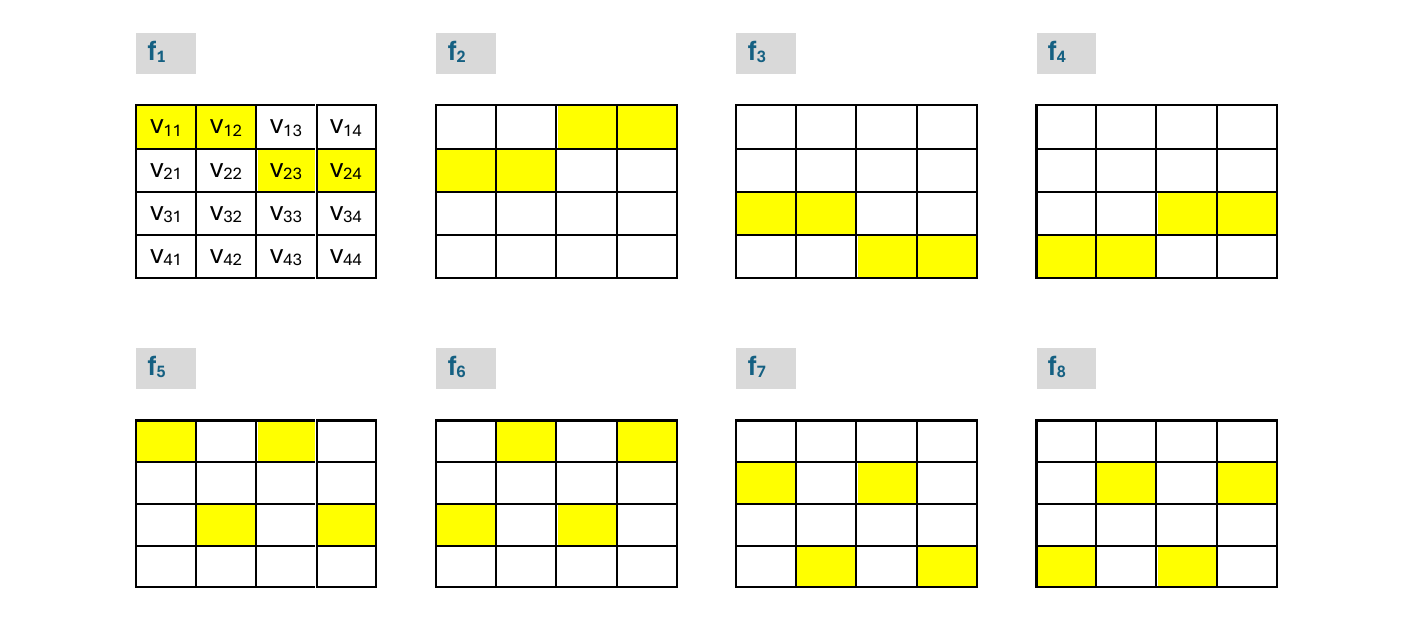}\caption
{\label{4x4_f's}{\bf\textsf{4x4 representation of the $\textit{f}$-\textit{measures}.}} The above figure is a 4x4 representation of the $f$-vectors, whereby the highlighted cells are the 4 components out of the behavior’s vector’s 16 components that make up each $f$-\textit{measure}. The $f$-\textit{measures} are divided into two sets, and each set consists of two pairs; these are $(f_1, f_2)$, $(f_3, f_4)$  and $(f_5, f_6)$, $(f_7, f_8)$, with the elements of each pair summing to 2 (e.g. $f_1+f_2=2$).}
\end{figure}

\textbf{\myuline{Symmetries of the solutions}:} these can be found using the five symmetry conditions from Eqs.~(\ref{sym-1})-(\ref{sym-5}) of the main text. Applying them in turn, we have for the \textit{f-measures} (where “to” denotes a swap of the two elements):\vspace{0.1cm}

Eq.~(\ref{sym-1}) - $a$ to $1-a$: $f_1$ to $f_2$, $f_3$ to $f_4$; others are (swapping columns 1 and 3, 2 and 4 in Fig.~\ref{4x4_f's})\,,

Eq.~(\ref{sym-2}) - $b$ to $1-b$: $f_5$ to $f_6$, $f_7$ to $f_8$; others are unchanged (swapping columns 1 and 2, 3 and 4 in Fig.~\ref{4x4_f's})\,,

Eq.~(\ref{sym-3}) - $x$ to $1-x$: $f_1$ to $f_3$\,, $f_2$ to $f_4$, $f_5$ to $f_6$, $f_7$ to $f_8$ (swapping rows 1 and 3, 2 and 4 in Fig.~\ref{4x4_f's})\,,

Eq.~(\ref{sym-4}) - $y$ to $1-y$: $f_5$ to $f_7$\,, $f_6$ to $f_8$, $f_1$ to $f_2$, $f_3$ to $f_4$ (swapping rows 1 and 2, 3 and 4 in Fig.~\ref{4x4_f's})\,,

Eq.~(\ref{sym-5}) - $(x,a)$ to $(y,b)$: $f_1$ to $f_5$, $f_2$ to $f_6$, $f_3$ to $f_7$, $f_4$ to $f_8$ (swapping rows 2 and 3, columns 2 and 3 in Fig.~\ref{4x4_f's})\,.\vspace{0.1cm}

\noindent We can thus get from $f_1$ to any other via a sequence of moves (for instance, $f_1\rightarrow f_8$ can be achieved by $f_1\rightarrow f_2$, then $f_2\rightarrow f_4$, and finally $f_4\rightarrow f_8$). We note that using the symmetries above we can similarly get from any vector to all of the members of its symmetry group.

\begin{figure}[h]
\centering
\includegraphics[page=1,
        trim=0.5cm 3cm 0.5cm 0.5cm,
        clip,
        width=0.9\textwidth]{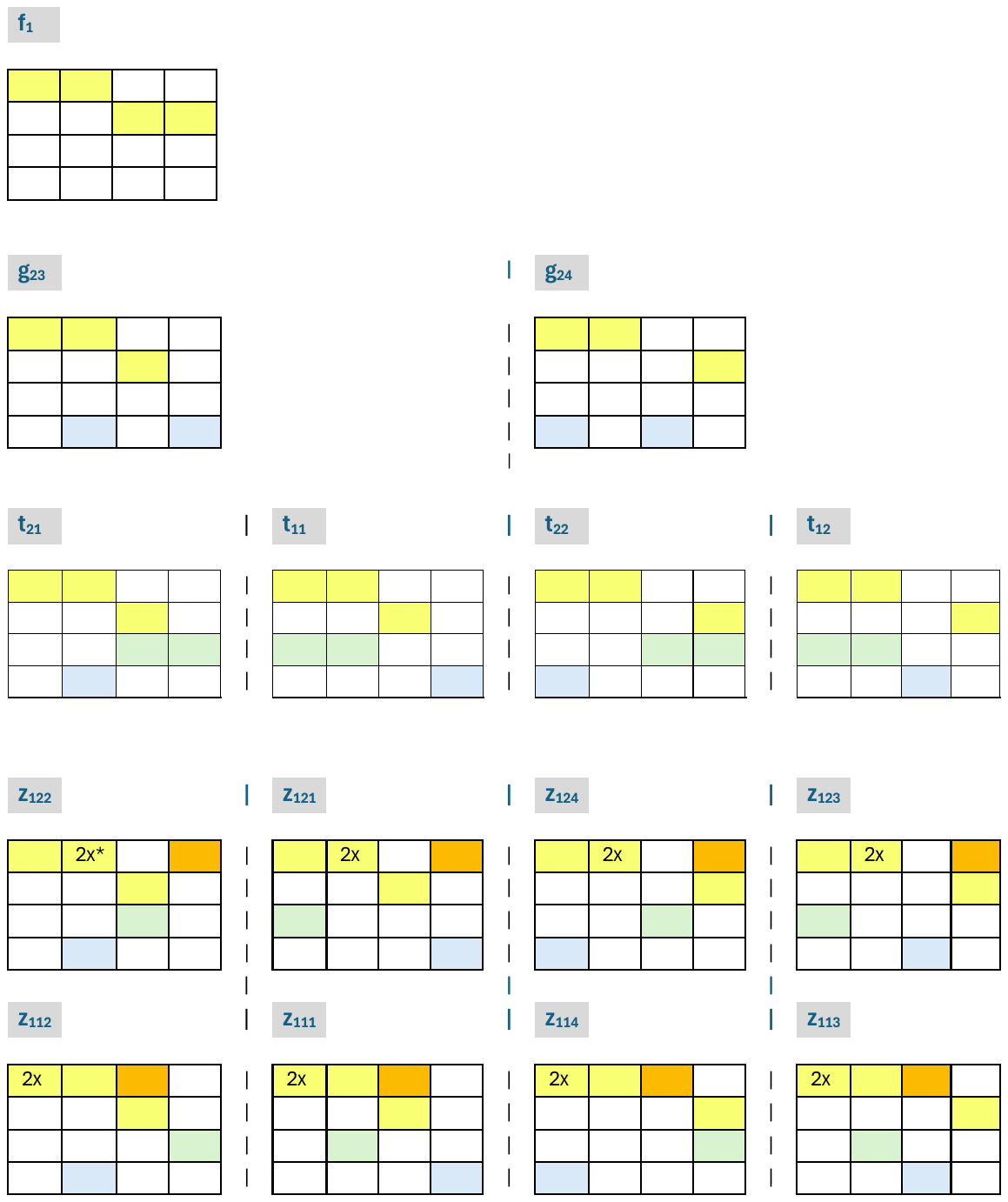}\caption
{\label{4x4_relations}{\textsf{\textbf{Relationship between $\textit{f}$-\textit{measures} and associated $\textit{g}$-, $\textit{t}$-, $\textit{z}$-\textit{measures}.}} Here we show a subset of the connections highlighted in Table~\ref{long-table} (also see Fig~\ref{Connections}). The term "2x'' indicates that a factor of 2 is applied to that component, when computing the corresponding measures (these factors of 2 occur only for the $z$-vectors).}}
\end{figure}

Fig.~\ref{4x4_relations} shows how $f_1$ connects to $g_{23}$ and $g_{24}$. Each $g_{ij}$ term can be constructed using two $f_i$ terms, as we see in the final column of Table~\ref{long-table}. For example, $g_{23}$ uses $f_1$ and $f_7$. In general, each $f_i$ term is involved in four $g_{ij}$ terms ($f_1$ is used for $g_{23}$ and $g_{24}$, as above, and also $g_{11}$ and $g_{12}$). 

Similarly, the figure shows how e.g. $g_{23}$ connects to $t_{21}$ and $t_{11}$ uses $g_{23}$ and $g_{11}$. Each $t_{ij}$ term uses two $g_{ij}$ terms, again see Table~\ref{long-table}.
Each $g_{ij}$ is used in four $t_{ij}$ terms ($g_{23}$ is used for $t_{21}$ and $t_{11}$, as above, and also $t_{53}$ and $t_{73}$). 

We then see how e.g. $t_{21}$ connects to $z_{122}$ and $z_{112}$. Each $z_{ijk}$ term can be constructed using two $t_{ij}$ terms (see Table~\ref{long-table}). For example, $z_{122}$ uses $t_{21}$ and $t_{72}$. Each $t_{ij}$ term is used in four $z_{ijk}$ terms ($t_{21}$ is used for $z_{122}$ and $z_{112}$, as above, and also $z_{331}$ and $z_{341}$).

These relations are further illustrated in the 'Directions' and 'Calculation' column in Table~\ref{long-table}. Recall, each vector in the $f$, $g$, $t$, $z/s$ classes can be characterized by violations of particular non-signaling constraints. If one takes a target vector in any class apart from $f$, then the non-signaling constraints corresponding to this vector are exactly the ones of the vectors which make up the target vector from earlier classes. For example, each $t_{ij}$ is characterized by three violations of non-signaling constraints, which are exactly the ones of the two $g_{kl}$ vectors, which make up this particular $t_{ij}$ one (two $g_{ij}$ vectors will be characterized by four violations, but one will be common between the two). See also Fig~\ref{Connections} below -- the highlighted connections are the same ones as used in the example in Fig.~\ref{4x4_relations}.

\begin{figure}[h]
\centering
\includegraphics[width=0.9\columnwidth]{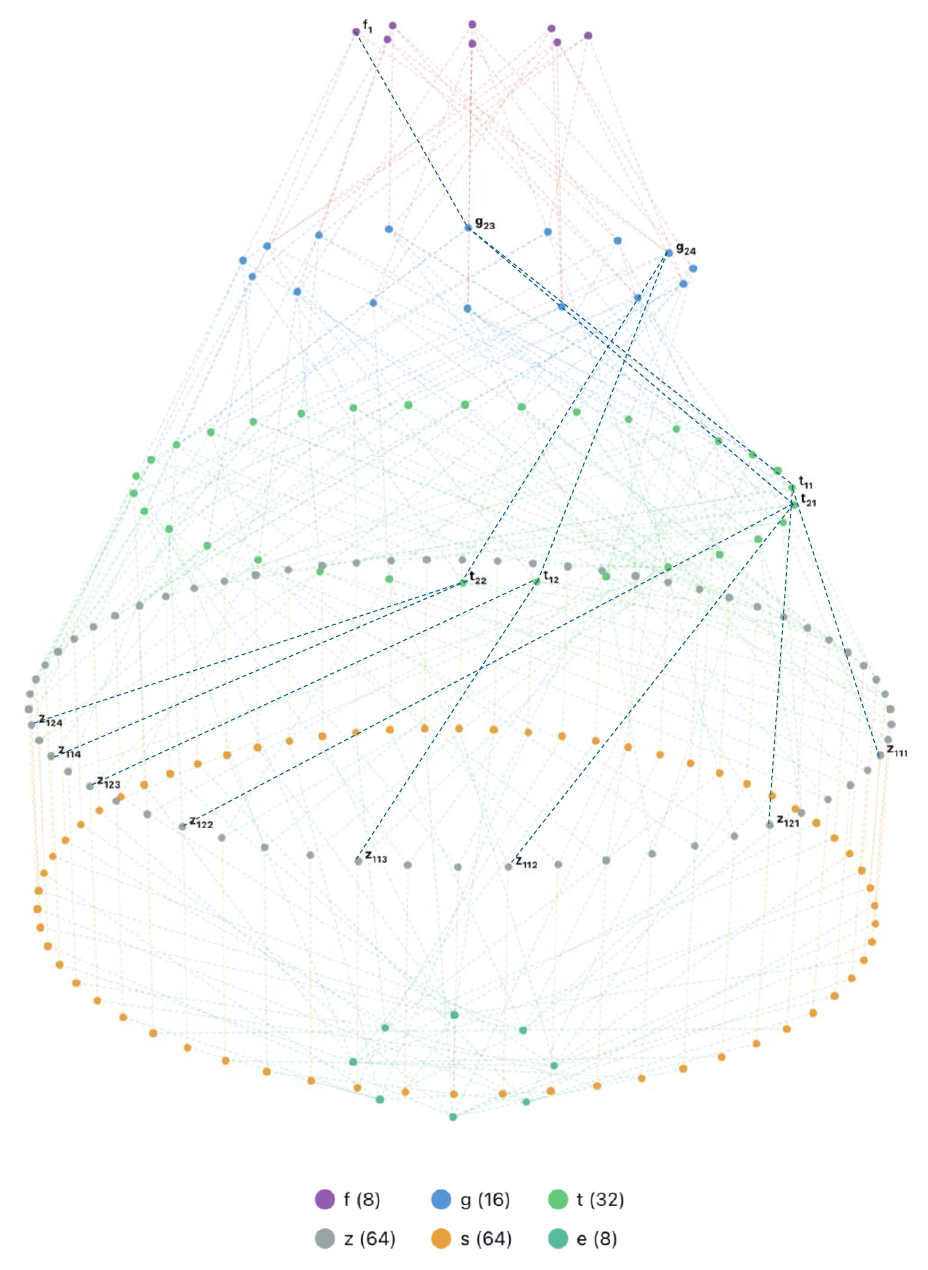}\caption
{\label{Connections}\textsf{\textbf{Connections between measures.}} This is a graph representation of the connections between \textit{f}-, \textit{g}-, \textit{t}- , \textit{z}- and \textit{s}-\textit{measures}, which reflects the algebraic expressions, i.e. how they can be derived from one class of the directions to the next:
each layer represents one category of \textit{f}-, \textit{g}-, \textit{t}-, \textit{z}- or \textit{s}-\textit{measures}, 
each \textit{g}-, \textit{t}- and \textit{z}-\textit{measure} is connected to two quantities from the layer above. Each \textit{f}-\textit{measure} is connected to four \textit{g}-\textit{measures}, each \textit{g}-\textit{measure} is connected to four \textit{t}-\textit{measures}, and each \textit{t}-\textit{measure} is connected to four \textit{z}-\textit{measures}. Finally, each \textit{s}-\textit{measure} is connected to its associated \textit{z}-\textit{measure} and to one \textit{e}-\textit{measure}. These relationships mirror those in Table~\ref{long-table}, as seen by comparing the entries in the first and fifth columns.
The figure was constructed using \texttt{Python}’s \texttt{Plotly} library and then visually refined using a barycentric optimization routine to minimize edge crossings. The example shown in Fig.~\ref{4x4_relations} is overlaid with a black dashed line.
}
\end{figure}

\end{document}